\documentclass[twocolumn,aps,pra,superscriptaddress,10pt]{revtex4-2}
\usepackage[utf8]{inputenc}
\usepackage[T1]{fontenc}
\usepackage{amsmath}
\usepackage{amsfonts}
\usepackage{amssymb}
\usepackage{amsthm}
\usepackage{bm}
\usepackage{dsfont}
\usepackage{graphicx}
\usepackage[dvipsnames,table]{xcolor}
\usepackage{nicefrac}
\usepackage{enumitem}
\usepackage{hyperref}
\usepackage{makecell}
\usepackage{braket}

\usepackage{pifont}
\newcommand{\cmark}{\ding{51}}%
\newcommand{\xmark}{\ding{55}}%

\usepackage{thmtools}
\declaretheorem[name=Theorem]{thm}

\begin{document}

	\title{Classically-Verifiable Quantum Advantage from a Computational Bell Test}
	\author{Gregory D. Kahanamoku-Meyer}
	\affiliation{Department of Physics, University of California at Berkeley, Berkeley, CA 94720}
	\author{Soonwon Choi}
	\affiliation{Department of Physics, University of California at Berkeley, Berkeley, CA 94720}
	\author{Umesh V. Vazirani}
	\affiliation{Department of Electrical Engineering and Computer Science, University of California at Berkeley, Berkeley, CA 94720}
	\author{Norman Y. Yao}
	\affiliation{Department of Physics, University of California at Berkeley, Berkeley, CA 94720}

	\begin{abstract}
	We propose and analyze a novel interactive protocol for demonstrating quantum computational advantage, which is efficiently classically verifiable.
	Our protocol relies upon the cryptographic hardness of trapdoor claw-free functions (TCFs).
	Through a surprising connection to Bell's inequality, our protocol avoids the need for an adaptive hardcore bit, with essentially no increase in the quantum circuit complexity and no extra cryptographic assumptions.
	Crucially, this expands the set of compatible TCFs, and we propose two new constructions: one based upon the decisional Diffie-Hellman problem and the other based upon Rabin's function, $x^2 \bmod N$.
	We also describe two independent innovations which improve the efficiency of our protocol's implementation: (i) a scheme to discard so-called ``garbage bits'', thereby removing the need for reversibility in the quantum circuits, and (ii) a natural way of performing post-selection which significantly reduces the fidelity needed to demonstrate quantum advantage. 
	These two constructions may also be of independent interest, as they may be applicable to other TCF-based quantum cryptography such as certifiable random number generation.
	Finally, we design several efficient circuits for $x^2 \bmod N$ and describe a blueprint for their implementation on a Rydberg-atom-based quantum computer. 
	\end{abstract}

	\maketitle

	\section{Introduction}

The development of large-scale programmable quantum hardware has opened the door to testing a fundamental question in the theory of computation: can quantum computers outperform classical ones for certain tasks?
This idea, termed quantum computational advantage, has motivated the design of novel algorithms and protocols to demonstrate advantage with minimal quantum resources such as qubit number and gate depth~\cite{aaronson_computational_2011, farhi2016quantum, bremner_average-case_2016, lund2017quantum, harrow2017quantum, terhal2018quantum, boixo_characterizing_2018, bouland_complexity_2019, aaronson_complexity-theoretic_2017, neill_blueprint_2018}.
Such protocols are naturally characterized along two axes: the computational speedup and the ease of verification.
The former distinguishes whether a quantum algorithm exhibits a polynomial or super-polynomial speedup over the best known classical one.
The latter classifies whether the correctness of the quantum computation is \emph{efficiently verifiable} by a classical computer.
Along these axes lie three broad paths to demonstrating advantage:
1) sampling from entangled quantum many-body wavefunctions,
2) solving a deterministic problem, e.g. prime factorization, via a quantum algorithm, and
3) proving quantumness through interactive protocols.

Sampling-based protocols directly rely on the classical hardness of simulating quantum mechanics~\cite{aaronson_computational_2011, bremner_average-case_2016, boixo_characterizing_2018, bouland_complexity_2019, aaronson_complexity-theoretic_2017, neill_blueprint_2018}.
The ``computational task'' is to prepare and measure a generic complex many-body wavefunction with little structure.
As such, these protocols typically require minimal resources and can be implemented on near-term quantum devices~\cite{arute_quantum_2019, zhong_quantum_2020}.
The correctness of the sampling results, however, is exponentially difficult to verify.
This has an important consequence: in the regime beyond the capability of classical computers, the sampling results cannot be explicitly checked, and quantum computational advantage can only be inferred (e.g.~extrapolated from simpler circuits).

Algorithms in the second class of protocols are naturally broken down by whether they exhibit polynomial or super-polynomial speed-ups.
In the case of polynomial speed-ups, there exist notable examples that are \emph{provably} faster than any possible classical algorithm \cite{bravyi_quantum_2018,bravyi_quantum_2019}.
However, polynomial speed-ups are tremendously challenging to demonstrate in practice, due to the slow growth of the separation between classical and quantum run-times \footnote{They also have some other caveats: a provable speedup of $\mathcal{O}(1)$ quantum complexity over $\mathcal{O}(n)$ classical complexity is promising, but just reading the input may require $\mathcal{O}(n)$ time, hiding the computational speedup in practice.}.
Accordingly, the most attractive algorithms for demonstrating advantage tend to be those with a super-polynomial speed-up, including Abelian hidden subgroup problems such as factoring and discrete logarithms \cite{shor_polynomial-time_1997}.
The challenge is that for all known protocols of this type, the quantum circuits required to demonstrate advantage are well beyond the capabilities of near-term experiments.

The final class of protocols demonstrates quantum advantage through an \emph{interactive proof}~\cite{brakerski_cryptographic_2019, brakerski_simpler_2020, aharonov2017interactive, watrous_pspace_1999, kitaev2000parallelization, kobayashi2003quantum, fitzsimons2015multiprover, markov2018quantum}.
At a high level, this type of protocol involves multiple rounds of communication between the classical verifier and the quantum prover; the prover must give self-consistent responses despite not knowing what the verifier will ask next.
This requirement of self-consistency rules out a broad range of classical cheating strategies and can imbue  ``hardness'' into questions that would otherwise be easy to answer.
To this end, interactive protocols expand the space of computational problems that can be used to demonstrate quantum advantage; from a more pragmatic perspective, this can enable the realization of efficiently verifiable quantum advantage on near-term quantum hardware.

	\begin{figure*}
		\begin{center}
		\includegraphics{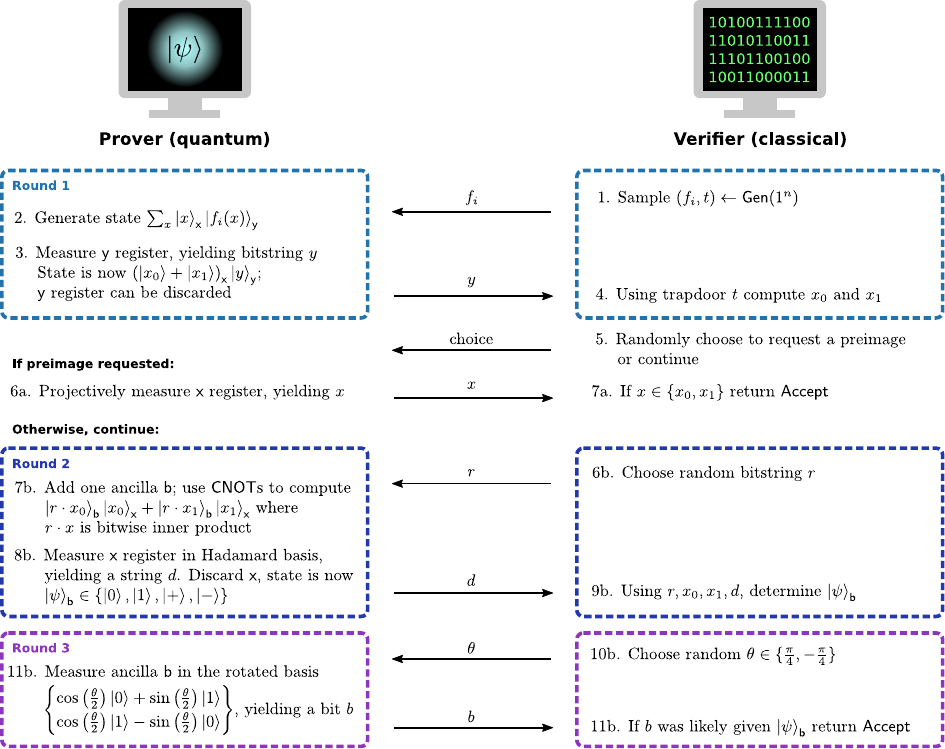}
		\end{center}
		\caption{\label{fig:main-protocol} Schematic representation of the interactive quantum advantage protocol. In the first round of interaction, the classical verifier (right) selects a specific function from a trapdoor claw-free family and the quantum prover (left) evaluates it over a superposition of inputs. 
		 The goal of the second round is to condense the information contained in the prover's superposition state onto a single ancilla qubit. 
		The final round of interaction effectively performs a Bell inequality measurement, whose outcome is cryptographically protected.
		}
	\end{figure*}

    Recently, a beautiful interactive protocol was introduced that can operate both as a test for quantum advantage and as a generator of certifiable quantum randomness \cite{brakerski_cryptographic_2019}.
	The core of the protocol is a two-to-one function, $f$, built on the computational problem known as learning with errors (LWE) \cite{regev_lattices_2005}.
	 The demonstration of advantage leverages two important properties of the function: first, it is  \emph{claw-free}, meaning that it is computationally hard to find a pair of inputs $(x_0, x_1)$ such that $f(x_0) = f(x_1)$.
	\footnote{``Claw-free'' is often used to refer to a \emph{pair} of functions $f_0, f_1$ such that for appropriate $x_0, x_1$ we have $f_0(x_0) = f_1(x_1)$. Here, we use the slightly more general idea of a single 2-to-1 function $f$ for which it is hard to find $x_0, x_1$ such that $f(x_0) = f(x_1)$. This is a special case of a ``collision-resistant function,'' which could potentially be many-to-one. We also note that a claw-free pair of functions can be converted into a single claw-free function by defining $f(b || x) = f_b(x)$, where $||$ denotes concatenation.}.
	Second, there exists a \emph{trapdoor}: given some secret data $t$, it becomes possible to efficiently invert $f$ and reveal the pair of inputs mapping to any output.
	(See supplemental information~\cite{SM} for an overview of trapdoor claw-free functions).
	However, to fully protect against cheating provers, the protocol requires a stronger version of the claw-free property called the \emph{adaptive hardcore bit}, namely, that for any input $x_0$ (which may be chosen by the prover), it is computationally hard to find even a single bit of information about $x_1$~\footnote{To be precise, it is hard to find both $x_0$ and the parity of any subset of the bits of $x_1$.}.
	The need for an adaptive hardcore bit within this protocol severely restricts the class of functions that can operate as verifiable tests of quantum advantage.

	Here, we propose and analyze a novel interactive quantum advantage protocol that removes the need for an adaptive hardcore bit, with essentially zero overhead in the quantum circuit and no extra cryptographic assumptions.
	We present four main results.
	First, we demonstrate how an idea from tests of Bell's inequality can serve the same cryptographic purpose as the adaptive hardcore bit \cite{bell_einstein_1964}.
	In essence, our interactive  protocol is a variant of the CHSH (Clauser,  Horne,  Shimony, Holt) game \cite{clauser_proposed_1969} in which one player is replaced by a cryptographic construction.
	Normally, in CHSH, two quantum parties are asked to produce correlations that would be impossible for classical devices to produce.
	If space-like separation is enforced to rule out communication between the two parties, then the correlations constitute a proof of quantumness.
	In our case, the space-like separation is replaced by the computational hardness of a cryptographic problem.
	In particular, the quantum prover holds a qubit whose state depends on the cryptographic secret in  the same way that the state of one CHSH player's qubit depends on the secret measurement basis of the other player.
	An alternative interpretation, from the perspective of Bell's theorem, is that the protocol can be thought of as a ``single-detector Bell test''---the cryptographic task generates the same single-qubit state as would be produced by entangling a second qubit and measuring it with another detector.
	As in the CHSH game, a quantum device can pass the verifier's test with probability $\sim85\%$, but a classical device can only succeed with probability at most $75\%$.
	This finite gap in success probabilities is precisely what enables a verifiable test of quantum advantage.

	Second, by removing the need for an adaptive hardcore bit, our protocol accepts a broader landscape of functions for interactive tests of quantum advantage (see Table~\ref{tab:tcfs}~and Methods).
	We populate this list with two new constructions.
	The first is based on the decisional Diffie-Hellman problem (DDH) \cite{diffie_new_1976,peikert_lossy_2008,freeman_more_2010}; the second utilizes the function $f_N(x) = x^2 \bmod N$ with $N$ the product of two primes, which forms the backbone of the Rabin cryptosystem \cite{rabin_digitalized_1979, goldwasser_digital_1988}.
	On the one hand, DDH is appealing because the elliptic-curve version of the problem is particularly hard for classical computers~\cite{miller_use_1986, koblitz_elliptic_1987, barker_recommendation_2016}.
	On the other hand,  $x^2 \bmod N$ can be implemented significantly more efficiently, while its hardness is equivalent to factoring.
	We hope that these two constructions will provide a foundation for the search for more TCFs with desirable properties (small key size and efficient quantum circuits).

	Third, we describe two innovations that facilitate our protocol's use in practice: a way to significantly reduce overhead arising from the reversibility requirement of quantum circuits, and a scheme for increasing noisy devices' probability of passing the test.
	Normally, quantum implementations of classical functions like the TCFs used in this protocol have some overhead, due to the need to make the circuit reversible in order to be consistent with unitarity~\cite{bennett_timespace_1989, levine_note_1990, aharonov1998quantum, babu2004synthesis, kotiyal2014circuit}.
	Our protocol exhibits the surprising property that it permits a measurement scheme to discard so-called ``garbage bits'' that arise during the computation, allowing classical circuits to be converted into quantum ones with essentially zero overhead.
	In the case of a noisy quantum device, the protocol also enables an inherent post-selection scheme for detecting and removing certain types of quantum errors.
	With this scheme it is possible for quantum devices to trade off low quantum fidelities for an increase in the overall runtime, while still passing the cryptographic test.
	We note that these constructions are likely applicable to other TCF-based quantum cryptography protocols as well, and thus may be of independent interest for tasks such as certifiable quantum random number generation.

	Finally, focusing on the TCF $x^2 \bmod N$, we provide explicit quantum circuits---both asymptotically optimal (requiring only $\mathcal{O}(n \log n)$ gates and $\mathcal{O}(n)$ qubits), as well as those aimed for near-term quantum devices.
	We show that a verifiable test of quantum advantage can be achieved with $\sim 10^3$ qubits and a gate depth $\sim 10^5$ (see Methods).
	We also co-design a specific implementation of $x^2 \bmod N$ optimized for a programmable Rydberg-based quantum computing platform.
	The native physical interaction corresponding to the Rydberg blockade mechanism enables the direct implementation of multi-qubit-controlled arbitrary phase rotations without the need to decompose such gates into universal two-qubit operations \cite{saffman_quantum_2016,levine_parallel_2019,graham_rydberg-mediated_2019,madjarov_high-fidelity_2020,browaeys2020many}.
	Access to such a native gate immediately reduces the gate depth for achieving quantum advantage by an order of magnitude.

	    \renewcommand{\arraystretch}{1.25}
        \begin{table}
		\begin{tabular}{|c|c|c|c|c|}
			\hline
			Problem & \thead{Trap\\door} & \thead{Claw-\\free} & \thead{Adaptive\\hard-core\\bit} & \thead{Asymptotic\\complexity\\ (gate count)} \\
			\hline
			\hline
			LWE \cite{brakerski_cryptographic_2019}
			& \cmark & \cmark & \cmark & $n^2 \log^2 n$ \\
			\hline
			$x^2 \bmod N$
			& \cmark & \cmark & \xmark & $n \log n$ \\
			\hline
			Ring-LWE
			\cite{brakerski_simpler_2020}
			& \cmark & \cmark & \xmark & $n \log^2 n$ \\
			\hline
			Diffie-Hellman
			& \cmark & \cmark & \xmark & $n^3 \log^2 n$ \\
			\hline
			\hline
			Shor's alg.
			& \hspace{1em} --- \hspace{1em} & \hspace{1em} --- \hspace{1em} & \hspace{1em} --- \hspace{1em} & $n^2 \log n$ \\
			\hline
		\end{tabular}

		\caption{\label{tab:tcfs}Cryptographic constructions for interactive quantum advantage protocols.
		$n$ represents the number of bits in the function's input string. 
		Big-$\mathcal{O}$ notation is implied and factors of $\log \log n$ and smaller are dropped.
		For references and derivations of the circuit complexities, see supplementary information~\cite{SM}.
		}
	\end{table}

	\section{Background and Related work}
	\label{sec:related}
	
	The use of trapdoor claw-free functions for quantum cryptographic tasks was pioneered in two recent breakthrough protocols: (i) giving classical homomorphic encryption for quantum circuits~\cite{mahadev_classical_homomorphic_2018} and (ii) for generating cryptographically certifiable quantum randomness from an untrusted black-box device~\cite{brakerski_cryptographic_2019}; this latter work  also introduced the notion of an adaptive hardcore bit and serves as an efficiently verifiable test of quantum advantage.
	Remarkably, the scheme was further extended to allow a classical server to cryptographically verify the correctness of arbitrary quantum computations~\cite{mahadev_classical_2018}; it has also been applied to  remote state preparation with implications for secure delegated computation~\cite{gheorghiu_computationally-secure_2019}.
	
	Recently, an improvement to the practicality of TCF-based proofs of quantumness was provided in the \emph{random oracle model} (ROM)---a model of computation in which both the quantum prover and classical verifier can query a third-party ``oracle,'' which returns a random (but consistent) output for each input.
	In that work, the authors provide a protocol that both removes the need for the adaptive hardcore bit, and also reduces the interaction to a single round \cite{brakerski_simpler_2020}.
	Because the security of the protocol is proven in the ROM, implementing this protocol in practice requires applying the \emph{random oracle heuristic}, in which the random oracle is replaced by a cryptographic hash function, but the hardness of classically defeating the protocol is taken to still hold
	\footnote{Replacing the random oracle with a hash function is termed a \emph{heuristic} rather than an \emph{assumption} because the security of this procedure generally holds in practice but is not provable---in fact, there exist constructions that are provably secure in the random oracle model but trivially insecure when instantiated with a hash function \cite{canetti_random_1998}.}.
	Only contrived cryptographic schemes have ever been broken by attacking the random oracle heuristic \cite{canetti_random_1998,koblitz_random_2015}, so it seems to be effective in practice and the ROM protocol has significant potential for use as a practical tool for benchmarking untrusted quantum servers.
	On the other hand, for a robust experimental test of the foundational complexity theoretic claims of quantum computing---that quantum mechanics allows for algorithms that are superpolynomially faster than classical Turing machines---we desire the complexity-theoretic backing of the speedup to be as strong as possible (i.e.~provable in the ``standard model'' of computation~\cite{aaronson_complexity-theoretic_2016}), which is the goal pursued in the present work.
	With that said, we emphasize that the various optimizations described below---including the TCF families based on DDH and $x^2 \bmod N$, as well as the schemes for postselection and discarding garbage bits---can be applied to the ROM protocol as well.
		
	Lastly, we also note two recent works which demonstrate that any TCF-based proof of quantumness, including the present work, can be implemented in constant quantum circuit depth (at the cost of more qubits) \cite{liu_depth-efficient_2021, hirahara_test_2021}.

	\section{Interactive Protocol for Quantum Advantage}
	\label{sec:protocol}

	Our full protocol is shown diagrammatically in Figure \ref{fig:main-protocol}.
	It consists of three rounds of interaction between the prover and verifier (with a ``round'' being a challenge from the verifier, followed by a response from the prover).
	The first round generates a multi-qubit superposition over two bit strings that would be cryptographically hard to compute classically.
	The second round maps this superposition onto the state of one ancilla qubit, retaining enough information to ensure that the resulting single-qubit state is also hard to compute classically.
	The third round takes this single qubit as input to a CHSH-type measurement, allowing the prover to generate a bit of data that is correlated with the cryptographic secret in a way that would not be possible classically.
	Having described the intuition behind the protocol, we now lay out each round in detail.

	\subsection{Description of the protocol}

	The goal of the first round is to generate a superposition over two colliding inputs to the trapdoor claw-free function (TCF).
	It begins with the verifier choosing an instance $f_i$ of the TCF along with the associated trapdoor data $t$; $f_i$ is sent to the prover.
	As an example, in the case of $x^2 \mod N$, the ``index'' $i$ is the modulus $N$, and the trapdoor data is its factorization, $p,q$.
	The prover now initializes two registers of qubits, which we denote as the $\mathsf{x}$ and $\mathsf{y}$ registers.
	On these registers, they compute the entangled superposition $\ket{\psi} = \sum_x \ket{x}_\mathsf{x} \ket{f_i(x)}_\mathsf{y}$, over all $x$ in the domain of $f_i$.
	The prover then measures the $\mathsf{y}$ register in the standard basis, collapsing the state to $\left(\ket{x_0} + \ket{x_1} \right)_\mathsf{x} \ket{y}_\mathsf{y}$, with $y = f(x_0) = f(x_1)$.
	The  measured bitstring $y$ is then sent to the verifier, who uses the secret trapdoor to compute $x_0$ and $x_1$ in full.

	At this point, the verifier randomly chooses to either request a projective measurement of the $\mathsf{x}$ register, ending the protocol, or to continue with the second and third rounds.
	In the former case, the prover communicates the result of that measurement, yielding either $x_0$ or $x_1$, and the verifier checks that indeed $f(x)=y$.
	In the latter case, the protocol proceeds with the final two rounds.

	The second round of interaction converts the many-qubit superposition $\ket{\psi} = \ket{x_0}_\mathsf{x} + \ket{x_1}_\mathsf{x}$ into a single-qubit state $\{\ket{0}_\mathsf{b}, \ket{1}_\mathsf{b}, \ket{+}_\mathsf{b}, \ket{-}_\mathsf{b}\}$ on an ancilla qubit $\mathsf{b}$.
	The final state of $\mathsf{b}$ depends on the values of both $x_0$ and $x_1$.
	The round begins with the verifier choosing a random bitstring $r$ of the same length as $x_0$ and $x_1$, and sending it to the prover.
	Using a series of CNOT gates from the $\mathsf{x}$ register to $\mathsf{b}$, the prover  computes the state $\ket{r \cdot x_0}_\mathsf{b} \ket{x_0}_\mathsf{x} + \ket{r \cdot x_1}_\mathsf{b} \ket{x_1}_\mathsf{x}$, where $r \cdot x$ denotes the binary inner product.
	Finally, the prover measures the $\mathsf{x}$ register in the Hadamard basis, storing the result as a bitstring $d$ which is sent to the verifier.
	This measurement disentangles $\mathsf{x}$ from $\mathsf{b}$ without collapsing $\mathsf{b}$'s superposition.
	At the end of the second round, the prover's state is $(-1)^{d \cdot x_0} \ket{r \cdot x_0}_\mathsf{b} + (-1)^{d \cdot x_1} \ket{r \cdot x_1}_\mathsf{b}$, which is one of $\{\ket{0}, \ket{1}, \ket{+}, \ket{-}\}$.
	Crucially, it is cryptographically hard to predict whether this state is one of $\{\ket{0}, \ket{1}\}$ or $\{\ket{+}, \ket{-}\}$.

	The final round of our protocol can be understood in analogy to the CHSH game \cite{clauser_proposed_1969}.
	While the prover cannot extract the polarization axis from their single qubit (echoing the no-signaling property of  CHSH), they can make a measurement that is \emph{correlated} with it.
	This measurement outcome ultimately constitutes the proof of quantumness.
	In particular, the verifier requests a measurement in an intermediate basis, rotated from the $Z$ axis around $Y$, by either $\theta=\pi/4$ or $-\pi/4$.
	Because the measurement basis is never perpendicular to the state, there will always be one outcome that is more likely than the other (specifically, with probability $\cos^2 (\pi/8) \approx 0.85$).
	The verifier returns $\mathsf{Accept}$ if this ``more likely'' outcome is the one received.

	In the next section, we prove that a quantum device can cause the verifier to $\mathsf{Accept}$ with substantially higher probability than any classical prover.
	A full test of quantum advantage would consist of running the protocol many times, until it can be established with high statistical confidence that the device has exceeded the classical probability bound.

\subsection{Completeness and soundness}

We now prove completeness (the noise-free quantum success probability) and soundness (an upper bound on the classical success probability).
Recall that after the first round of the protocol, the verifier chooses to either request a standard basis measurement of the first register, or to continue with the second and third rounds.
In the proofs below, we analyze the prover's success probability across these two cases separately.
We denote the probability that the verifier will accept the prover's string $x$ in the first case as $p_x$, and the probability that the verifier will accept the single-qubit measurement result in the second case as $p_{\mathrm{CHSH}}$.

    \subsubsection{Perfect quantum prover (completeness)}
	\label{subsec:honest-quantum}

	\begin{thm}
	\label{thm:honest-quantum}
		An error-free quantum device honestly following the interactive protocol  will cause the verifier to return $\mathsf{Accept}$ with $p_x = 1$ and $p_\mathrm{CHSH} = \cos^2 \left(\pi/8 \right) \approx 0.85$.
	\end{thm}

	\begin{proof}
		If the verifier chooses to request a projective measurement of $\mathsf{x}$ after the first round, an honest quantum prover succeeds with probability $p_x = 1$ by inspection.

		If the verifier chooses to instead perform the rest of the protocol, the prover will hold one of $\{\ket{0}, \ket{1}, \ket{+}, \ket{-}\}$ after round 2.
		In either measurement basis the verifier may request in round 3, there will be one outcome that occurs with probability $\cos^2 \left(\pi / 8 \right)$, which is by construction the one the verifier accepts.
		Thus, an honest quantum prover has $p_\mathrm{CHSH} = \cos^2(\pi/8) \approx 0.85$.
	\end{proof}

	\subsubsection{Classical prover (soundness)}
	\label{subsec:classical-hardness}

	\begin{restatable}{thm}{classicalhardness}
		\label{thm:classical-success}
		Assume the function family used in the interactive protocol is claw-free. Then, $p_x$ and $p_\mathrm{CHSH}$ for any classical prover must obey the relation
		\begin{equation}
		\label{eq:classical-bound}
		p_x + 4p_\mathrm{CHSH} - 4 < \epsilon(n)
		\end{equation}
		where $\epsilon$ is a negligible function of $n$, the length of the function family's input strings.
	\end{restatable}

	\begin{proof}

	We prove by contradiction. Assume  that there exists a classical machine $\mathcal{A}$ for which $p_x + 4p_\mathrm{CHSH} - 4 \geq \mu(n)$, for a non-negligible function $\mu$. We show that there exists another algorithm $\mathcal{B}$ that uses $\mathcal{A}$ as a subroutine to find a pair of colliding inputs to the claw-free function, a contradiction.

		Given a claw-free function instance $f_i$, $\mathcal{B}$ acts as a simulated verifier for $\mathcal{A}$.
		$\mathcal{B}$ begins by supplying $f_i$ to $\mathcal{A}$, after which $\mathcal{A}$ returns a value $y$, completing the first round of interaction.
		$\mathcal{B}$ now chooses to request the projective measurement of the $x$ register, and stores the result as $x_0$.
		Letting $p_{x_0}$ be the probability that $x_0$ is a valid preimage, by definition of $p_x$ we have $p_{x_0} = p_x$.

		Next, $\mathcal{B}$ \emph{rewinds} the execution of $\mathcal{A}$, to its state before $x_0$ was requested.
		Crucially, rewinding is possible because $\mathcal{A}$ is a classical algorithm.
		$\mathcal{B}$ now proceeds by running $\mathcal{A}$ through the second and third rounds of the protocol for many different values of the bitstring $r$ (Fig.~1),  rewinding each time.

		We now show that, for  $r$ selected uniformly at random,  $\mathcal{B}$ can extract the value of the inner product $r\cdot x_1$ with probability $p_{r \cdot x_1} \geq 1 - 2(1-p_\mathrm{CHSH})$.
		$\mathcal{B}$ begins by sending $r$ to $\mathcal{A}$, and receiving the bitstring $d$.
		$\mathcal{B}$ then requests the measurement result in both the  $\theta =\pi/4$ and $\theta = -\pi/4$ bases, by rewinding in between.
		Supposing that both the received values are ``correct'' (i.e.~would be accepted by the real verifier), they uniquely determine the single-qubit state $\ket{\psi} \in \{\ket{0}, \ket{1}, \ket{+}, \ket{-}\}$ that would be held by an honest quantum prover.
		This state reveals whether $r\cdot x_0 = r\cdot x_1$, and because $\mathcal{B}$ already holds $x_0$,  $\mathcal{B}$ can compute $r\cdot x_1$.
		We may define the probability (taken over all randomness except the choice of $\theta$) that the prover returns an accepting value in the cases $\theta=\pi/4$ and $\theta=-\pi/4$ as $p_{\pi/4}$ and $p_{-\pi/4}$ respectively.
		Then, via union bound, the probability that both are indeed correct is $p_{r \cdot x_1} \geq 1 - (1-p_{\pi/4}) - (1-p_{-\pi/4})$.
		Considering that $p_{\mathrm{CHSH}} = (p_{\pi/4} + p_{-\pi/4})/2$, we have $p_{r \cdot x_1} \geq 1 - 2(1-p_\mathrm{CHSH})$.

		Now, we show that extracting $r\cdot x_1$ in this way allows $x_1$ to be determined in full even in the presence of noise, by rewinding many times and querying for specific (correlated) choices of $r$.
		In particular, the above construction is a noisy oracle to the encoding of $x_1$ under the Hadamard code.
		By the Goldreich-Levin theorem \cite{goldreich_hard-core_1989}, list decoding applied to such an oracle will generate a polynomial-length list of candidates for $x_1$.
		If the noise rate of the oracle is noticeably less than $1/2$, $x_1$ will be contained in that list; $\mathcal{B}$ can iterate through the candidates until it finds one for which $f(x_1) = y$.

		By Lemma~1 in the Methods, for a particular iteration of the protocol, the probability that list decoding succeeds is bounded by $p_{x_1} > 2p_{r\cdot x_1} - 1 - 2\mu'(n)$, for a noticeable function $\mu'(n)$ of our choice \footnote{The oracle's noise rate is \emph{not} simply $p_{r \cdot x_1}$: that is the probability that any single value $r \cdot x_1$ is correct, but all of the queries to the oracle are correlated (they are for the same iteration of the protocol, and thus the same value of $y$).}.
		Setting $\mu'(n) = \mu(n)/4$ and combining with the previous result yields $p_{x_1} > 1 - 4(1-p_\mathrm{CHSH}) - \mu(n)/2$.

		Finally, via union bound, the probability that $\mathcal{B}$ returns a claw is
		\[{P}_\mathcal{B} \geq 1 - (1 - p_{x_0}) - (1-p_{x_1}) > p_x + 4 p_\mathrm{CHSH} - 4 - \mu(n)/2 \]
		and via the assumption that $p_x + 4 p_\mathrm{CHSH} - 4 > \mu(n)$ we have
		\[{P}_\mathcal{B} > \mu(n)/2 \]
		a contradiction.
	\end{proof}

	If we let $p_x = 1$, the bound requires that $p_\mathrm{CHSH} < 3/4 + \epsilon(n)$ for a classical device, while $p_\mathrm{CHSH} \approx 0.85$ for a quantum device, matching the classical and quantum success probabilities of the CHSH game.
	In the supplementary information~\cite{SM}, we provide an example of a classical algorithm saturating the bound with $p_x = 1$ and $p_\mathrm{CHSH} = 3/4$.

	\subsection{Robustness:~Error mitigation via postselection}
	\label{subsec:postselection}

	The existence of a finite gap between the classical and quantum success probabilities implies that our protocol can tolerate a certain amount of noise.
	A direct implementation of our interactive protocol on a noisy quantum device would require an overall fidelity of  $\sim 83$\% in order to exceed the classical bound~\footnote{This number comes from solving the classical bound (Equation~\ref{eq:classical-bound}) for circuit fidelity $\mathcal{F}$, with $p_x = \mathcal{F}$ and $p_\mathrm{CHSH} = \nicefrac{1}{2} + \mathcal{F}/2$.}.
	To allow devices with lower fidelities to demonstrate quantum advantage, our protocol allows for a natural tradeoff between fidelity and runtime, such that the classical bound can, in principle, be exceeded with only a small [e.g. $1/\mathrm{poly}(n)$] amount of coherence in the quantum device \footnote{This is true even if the coherence is exponentially small in $n$. Of course, with arbitrarily low coherence the runtime may become excessively large such that quantum advantage cannot be demonstrated---the point is that regardless of runtime, the classical probability bound can be exceeded with a device that has arbitrarily low circuit fidelity.}.

\begin{figure}
	\centering
	\includegraphics[width=0.85\linewidth]{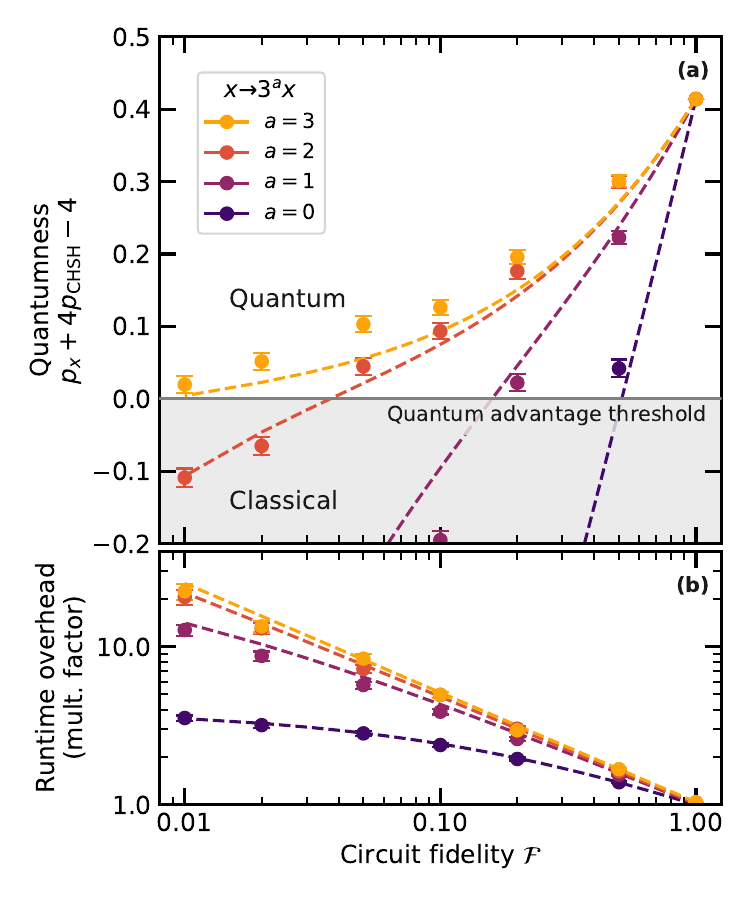}
	\caption{\label{fig:postselection} Performance of our post-selection scheme when redundancy is added to the function $x^2 \bmod N$ by mapping it to $(3^a x)^2 \bmod 3^{2a}N$. 
	Numerical simulations are performed on a quantum circuit implementing the Karatsuba algorithm for $a=\{{0,1,2,3}\}$ (see supplemental information~\cite{SM}). {\bf (a)}  
	``Quantumness'' measured in terms of the classical bound from Eqn.~\ref{eq:classical-bound} as a function of the total circuit fidelity. 
	 With $a=3$, even a quantum device with only 1\% circuit fidelity can demonstrate quantum advantage. \textbf{(b)} Depicts the increased runtime associated the post-selection scheme, which arises from a combination of slightly larger circuit sizes and the need to re-run the circuit multiple times.
	  The latter is by far the dominant effect.
	 Dashed  lines are a theory prediction with no fit parameters; points are the result of numerical simulations at $n=512$ bits and error bars depict $2\sigma$ uncertainty. 
	  }
	\end{figure}

	The key idea is based upon postselection.
	For most TCFs, there are many bitstrings of the correct length that are not valid outputs of $f$.
	Thus, if the prover detects such a  $y$ value in step 3 (Fig.~1), they can simply discard it and try again
	\footnote{This scheme will only remove errors in the first round of the protocol, but fortunately, one expects the overwhelming majority of the quantum computation, and thus also the majority of errors, to occur in that round.}.
	In principle, the verifier can even use their trapdoor data to silently detect and discard iterations of the protocol with invalid $y$ \footnote{This procedure does not leak data to a classical cheater, because the verifier does not communicate which runs were discarded. Furthermore, it does not affect the soundness of Theorem~\ref{thm:classical-success}, because the machine $\mathcal{B}$ in that theorem's proof can simply iterate until it encounters a valid $y$.}.
	Since $y$ is a function of $x_0$ and $x_1$, one might hope that this postselection scheme also rejects states where $x_0$ or  $x_1$ has become corrupt.
	Although this may not always be the case, we demonstrate numerically that this assumption holds for a specific implementation of $x^2 \bmod N$ in the following subsection.
	 One could also compute a classical checksum of $x_0$ and $x_1$ before and after the main circuit to ensure that they have not changed during its execution.
	Assuming that such bit-flip errors are indeed rejected, the possibility remains of an error in the phase between $\ket{x_0}$ and $\ket{x_1}$.
	In the supplementary information~\cite{SM}, we demonstrate that a prover holding the correct bitstrings but with an error in the phase can still saturate the classical bound; if the prover can avoid phase errors even a small fraction of the time, they will push past the classical threshold.

	\subsubsection{Numerical analysis of the postselection scheme for \texorpdfstring{$x^2 \bmod N$}{x² mod N}}
	\label{subsec:numerical-postselection}

	Focusing on the function $f(x) = x^2 \bmod N$, we now explicitly analyze the effectiveness of the postselection scheme.
	Let $m$ be the length of the outputs of this function.
	In this case, approximately $1/4$ of the bitstrings of length $m$ are valid outputs, so one would  naively expect to reject about $3/4$ of  corrupted bitstrings.
	 By introducing additional redundancy into the outputs of $f$ and thus increasing $m$, one can further decrease the probability that a corrupted $y$ will incorrectly be accepted.
	As an example, let us consider mapping $x^2 \bmod N$ to the function $(kx)^2 \bmod k^2 N$ for some integer $k$.
	This is particularly convenient because the prover can validate $y$ by simply checking whether it is a multiple of $k^2$.
	Moreover, the mapping adds only $\log k$ bits to the size of the problem, while rejecting a fraction $1 - 1/k^2$ of corrupted bitstrings.

	We perform extensive numerical simulations demonstrating that postselection allows for quantum advantage to be achieved using noisy devices with low circuit fidelities (Fig.~2).
	 We simulate quantum circuits for $(kx)^2 \bmod k^2 N$ at a problem size of $n=512$ bits.
	 Assuming a uniform gate fidelity across the circuit, we analyze the success rate of a quantum prover for $k = 3^a$ and $a = \{0,1,2,3\}$.
	 For these simulations we use our implementation of the Karatsuba algorithm (see Section~\ref{subsec:near-term-circuits}) because it is the most efficient in terms of gate count and depth. 
	 The choice of $k=3^a$, and details of the simulation, are explained in the supplementary information~\cite{SM}.

	For $a=0$, the circuit implements our original function $x^2 \bmod N$, where in the absence of postselection, an overall circuit fidelity of $\mathcal{F} \sim 0.83$ is required to achieve quantum advantage.
	As depicted in Fig.~\ref{fig:postselection}(a), even for $a=0$, our postselection scheme improves the advantage threshold down to $\mathcal{F} \sim 0.51$.
	For $a=2$,  circuit fidelities with $\mathcal{F} \gtrsim 0.1$ remain well above the quantum advantage threshold, while for $a=3$ the required circuit fidelity drops below $1\%$.

	However, there is a tradeoff. In particular, one expects the overall runtime to increase for two reasons: (i) there will be a slight increase in the circuit size for  $a>0$ and (ii) one may need to re-run the quantum circuit many times until a valid $y$ is measured.
	Somewhat remarkably, a runtime overhead of only $4.7$x already enables
 quantum advantage to be achieved with an overall circuit fidelity of  $10\%$ [Fig.~\ref{fig:postselection}(b)].
	Crucially, this increase in runtime is overwhelmingly due to re-running the quantum circuit and  does not imply the need for longer experimental coherence times.

	\subsection{Efficient quantum evaluation of irreversible classical circuits}
	\label{subsec:garbage}

	The central computational step in our interactive protocol (i.e.~step 2, Fig.~\ref{fig:main-protocol}) is for the prover  to apply  a unitary of the form:
	\begin{equation}
	\label{eq:Uf}
	\mathcal{U}_{f_i} \sum_x \ket{x}_\mathsf{x}\ket{0^{\otimes m}}_\mathsf{y} = \sum_x \ket{x}_\mathsf{x} \ket{f_i(x)}_\mathsf{y},
	\end{equation}
	where $f_i(x)$ is a classical function and $m$ is the length of the output register.
	This type of unitary operation is ubiquitous across quantum algorithms, and a common strategy for its implementation is to  convert the gates of a classical circuit into quantum gates.
	Generically, this process induces substantial overhead in both time and space complexity owing to the need to make the circuit reversible to preserve unitarity \cite{bennett_timespace_1989, levine_note_1990}.
	This reversibility is often achieved by using an additional register, $\mathsf{g}$, of so-called ``garbage bits'' and implementing:
	$\mathcal{U}_{f_i}^\prime \sum_x \ket{x}_\mathsf{x}\ket{0^{\otimes m}}_\mathsf{y}\ket{0^{\otimes l}}_\mathsf{g} = \sum_x \ket{x}_\mathsf{x} \ket{f_i(x)}_\mathsf{y} \ket{g_i(x)}_\mathsf{g}$.
	For each gate in the classical circuit, enough garbage bits are added to make the operation injective.
	In general, to maintain coherence, these bits cannot be discarded but must be ``uncomputed'' later, adding significant complexity to the circuits.

	A particularly appealing feature of our protocol is the existence of a measurement scheme to  discard garbage bits, allowing for the direct mapping of classical to quantum circuits with no overhead.
	Specifically, we envision  the prover measuring the qubits of the $\mathsf{g}$ register in the Hadamard basis and storing the results as a bitstring $h$, yielding the state,
	\begin{equation}
	\label{eq:phase-result}
	\ket{\psi} = \sum_x (-1)^{h \cdot g_i(x)}\ket{x}_\mathsf{x} \ket{f_i(x)}_\mathsf{y}.
	\end{equation}
	The prover has avoided the need to do any uncomputation of the garbage bits, at the expense of introducing phase flips onto some elements of the superposition.
	These phase flips do not affect the protocol, so long as the verifier can determine them.
	While classically computing $h \cdot g_i(x)$ is efficient for any $x$, computing it for \emph{all} terms in the superposition is infeasible for the verifier.
	However, our protocol provides a natural way around this.
	The verifier can wait until the prover has collapsed the superposition onto $x_0$ and $x_1$, before evaluating $g_i(x)$ only on those two inputs~\footnote{This is true because $g_i(x)$ is the result of adding extra output bits to the gates of a classical circuit, which is efficient to evaluate on any input.}.

	Crucially, the prover can measure away garbage qubits as soon as they would be discarded classically, instead of waiting until the computation has completed.
	If these qubits are then reused, the quantum circuit will use no more space than the classical one.
	This feature allows for significant improvements in both gate depth and qubit number for practical implementations of the protocol (see last rows of Table~I in Methods).
	We note that performing many individual measurements on a subset of the qubits is difficult on some experimental systems, which may make this technique challenging to use in practice.
	However, recent hardware advances have demonstrated these ``intermediate measurements'' in practice with high fidelity, for example by spatially shuttling trapped ions \cite{zhu2021demonstration,ryan-anderson_realization_2021}.
	We thus expect that the capability to perform partial measurements will not be a barrier in the near term.
	This issue can also be mitigated somewhat by collecting ancilla qubits and measuring them in batches rather than one-by-one, allowing for a direct trade-off between ancilla usage and the number of partial measurements.

	\section{Quantum circuits for trapdoor claw-free functions}
	\label{sec:implementation}
	
	\begin{figure*}
		\centering
		\includegraphics[width=0.9\textwidth]{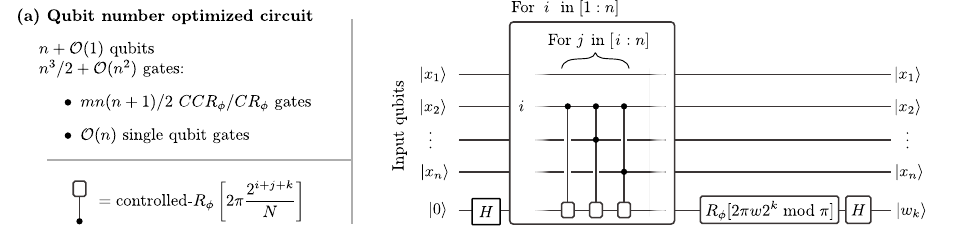}
		\vspace{3mm}
		
		\includegraphics[width=0.9\textwidth]{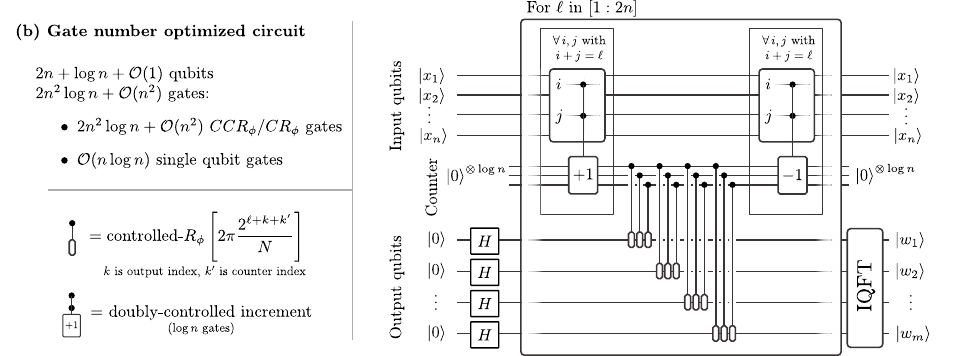}
		\caption{Quantum circuits implementing step 2 of our interactive protocol for $f(x) = x^2 \bmod N$.
		$n$ is the length of the input register, and $m = n + \mathcal{O}(1)$ is the length of the output register.
		{\bf (a)} Depicts a quantum circuit optimized for qubit number. The circuit shown computes the $k^\mathrm{th}$ bit of $w=x^2/N$ and should be iterated for $k$.
		This iteration should begin at the least significant bit to ensure that the final phase rotation can be estimated classically. 
		 Note that the only entangling operations necessary for the circuit are  doubly-controlled gates, which can be natively implemented using the Rydberg blockade (see Section~\ref{subsec:rydberg}).
		  {\bf (b)}~Depicts a quantum circuit optimized for gate number. By combining gates of equal phase, one can reduce the overall circuit complexity to $\mathcal{O}(n^2 \log n)$ gates. 
		  We note that neither circuit requires use of the ``garbage bit'' procedure described in Section~\ref{subsec:garbage}; this design choice reduces measurement complexity.
		  If desired, that procedure could be applied to the counter register of circuit (b) in place of the controlled-decrement operation.
		}
		\label{fig:circuit-1}
	\end{figure*}

While all of the trapdoor, claw-free functions listed in Table~\ref{tab:tcfs} can be utilized within our interactive protocol, each has its own set of advantages and disadvantages. 
For example, the TCF based on the Diffie-Hellman problem (described in the Methods) already enables a demonstration of quantum advantage at a key size of 160 bits (with a hardness equivalent to 1024 bit integer factorization~\cite{barker_recommendation_2016}); however, building a   circuit for this TCF requires a quantum implementation of Euclid's algorithm, which is  challenging~\cite{haner_improved_2020}.
Thus, we focus on designing quantum circuits implementing Rabin's function, $x^2 \bmod N$.
	
	\subsection{Quantum circuits for \texorpdfstring{$x^2 \bmod N$}{x² mod N}}
	\label{subsec:near-term-circuits}
	
	We explore four different circuits (implementations of these algorithms in \texttt{Python} using the \texttt{Cirq} library are included as supplementary files~\footnote{Code is available at \url{https://github.com/GregDMeyer/quantum-advantage} and is archived on Zenodo~\cite{cirq_code}}).
	The first two are quantum implementations of the Karatsuba and ``schoolbook'' classical integer multiplication algorithms, where we leverage the reversibility optimizations described in Section~\ref{subsec:garbage} (see supplementary information~\cite{SM}).
	The latter pair, which we call the ``phase circuits'' and describe below, are intrinsically quantum algorithms that use Ising interactions to directly compute $x^2 \bmod N$ in the phase.
	Using those circuits, we propose a near-term demonstration of our interactive protocol on a Rydberg-based quantum computer~\cite{levine_parallel_2019,browaeys2020many}; crucially, the so-called ``Rydberg blockade'' interaction \emph{natively} realizes multi-qubit controlled phase rotations, from which the entire circuits shown in Figure 3 are built (up to single qubit rotations).
	A comparison of approximate gate counts for each of the four  circuits can be seen in Table~I in the Methods.
	The Karatsuba algorithm is the most efficient in total gates and circuit depth, while the phase circuits are most efficient in terms of qubit usage and measurement complexity.

	We now describe the two  circuits, amenable to near-term quantum devices, that utilize quantum phase estimation to implement the function $f(x) = x^2 \bmod N$. The intuition behind our approach is as follows: we will compute $x^2/N$ in the phase and  transfer it to an output register via an inverse quantum Fourier transform~\cite{draper_addition_2000, beauregard_circuit_2003}; the modulo operation occurs automatically as the phase wraps around the unit circle, avoiding the need for a separate reduction step.

	In order to implement $\sum_x \ket{x}_\mathsf{x} \ket{x^2 \bmod N}_\mathsf{y}$, we  design a circuit to compute:
	\begin{equation}
	(\mathbb{I} \otimes \mathrm{IQFT}) \; \tilde{\mathcal{U}}_{w_N} \; (\mathbb{I} \otimes \mathrm{H^{\otimes m}}) \ket{x}\ket{0^{\otimes m}} = \ket{x}\ket{w}
	\end{equation}
	where $\mathrm{H}$ is a Hadamard gate, $\mathrm{IQFT}$ represents an inverse quantum Fourier transform, $w \equiv x^2/N = 0.w_1 w_2 \cdots w_m$ is an $m$-bit binary fraction \footnote{We must take $m > n + \mathcal{O}(1)$ to sufficiently resolve the value $x^2 \bmod N$ in post-processing},
	and $\tilde{\mathcal{U}}_{w_N}$ is the diagonal unitary,
	\begin{equation}
	\tilde{\mathcal{U}}_{w_N} \ket{x}\ket{z} = \exp \left(2 \pi i\frac{x^2}{N} z\right) \ket{x}\ket{z}.
	\label{eq:Utilde}
	\end{equation}
	By performing a binary decomposition of the phase in Eqn.~\ref{eq:Utilde}:
	\begin{equation}
	\label{eq:full-phase}
	\exp{\left(2 \pi i\frac{x^2}{N} z\right)} = \prod_{i,j,k} \exp{\left( 2 \pi i \frac{2^{i+j+k}}{N}  x_i x_j z_k\right)},
	\end{equation}
	one immediately finds that  $\tilde{\mathcal{U}}_{w_N}$ is equivalent to applying a series of controlled-controlled-phase rotation  gates of angle,
	\begin{equation}
	\label{eq:single-phase}
	\phi_{ijk} = \frac{2 \pi 2^{i+j+k}}{N} \pmod{2\pi}.
	\end{equation}
	Here, the control qubits are $i,j$ in the $\mathsf{x}$ register, while the target qubit is $k$ in the $\mathsf{y}$ register.
	Crucially, the value of this phase for any $i,j,k$ can be computed classically when the circuit is compiled.

	As depicted in Figure~\ref{fig:circuit-1}, we propose two explicit circuits to implement $\tilde{\mathcal{U}}_{w_N}$, one optimizing for qubit count, and the other optimizing for gate count.
	The first circuit [Fig.~\ref{fig:circuit-1}(a)] takes  advantage of the fact that the output register is  measured immediately after it is computed; this allows one to replace the $m$ output qubits with a single qubit that is measured and reused $m$ times.
	Moreover, by replacing groups of doubly-controlled gates with a Toffoli and a series of singly-controlled gates, one ultimately arrives at an implementation, which requires $n^3/2 + \mathcal{O}(n^2)$ gates, but only $n + \mathcal{O}(1)$ qubits.
	We note that this does require individual measurement and re-use of qubits, which has been a challenge for experiments; recent experiments however have demonstrated this capability \cite{zhu2021demonstration, ryan-anderson_realization_2021}.

	Our second circuit [Fig.~\ref{fig:circuit-1}(b)], which optimizes for gate count, leverages the fact that $\phi_{ijk}$ (Eqn.~\ref{eq:single-phase}) only depends on $i+j+k$,  allowing one to combine gates with a common sum.
	In this case, one can define  $\ell=i+j$ and then, for each value of $\ell$, simply ``count'' the number of values of $i,j$ for which both control qubits are 1.
	By then performing controlled gates off of the qubits of the counter register, one can reduce the total gate complexity by a factor of $n / \log n$, leading to a implementation with $2n^2 \log n + \mathcal{O}(n^2)$ gates.

	\begin{figure}[t!]
	    \centering
		\includegraphics[width=1.0\columnwidth]{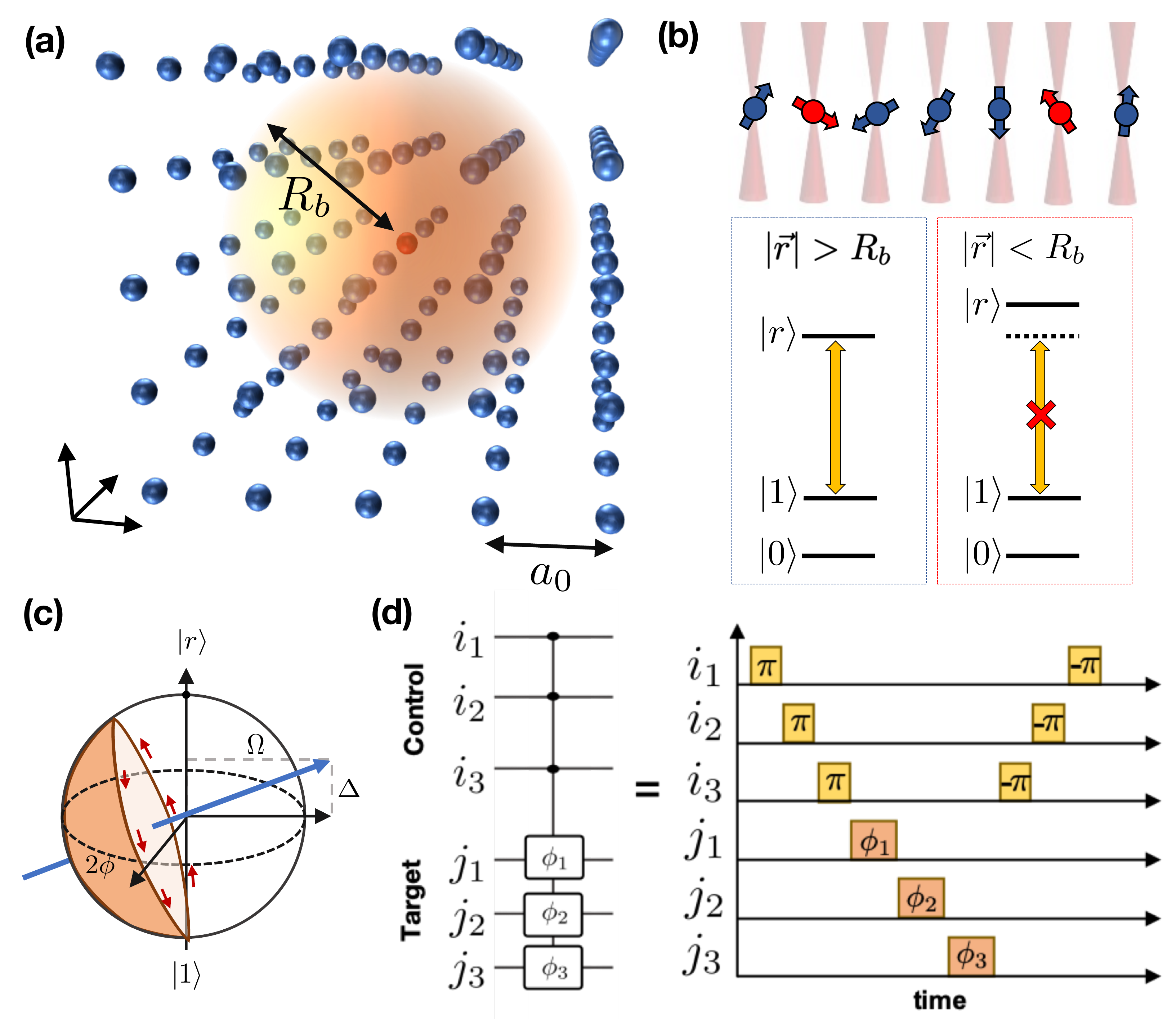}
		\caption{\textbf{a,} Schematic illustration of a three dimensional array of neutral atoms with Rydberg blockade interactions. The blockade radius can be significantly larger than the inter-atom spacing, enabling multi-qubit entangling operations. 
		\textbf{b,} As an example, Rydberg atoms can be trapped in an optical tweezer array. The presence of an atom in a Rydberg excited state (red)  shifts the energy levels of nearby atoms (blue), preventing the driving field  (yellow arrow) from exciting them to their Rydberg state, $|r\rangle$.
		\textbf{c,} A single qubit phase rotation can be implemented by an off-resonant Rabi oscillation between one of the qubit states, e.g., $\ket{1}$, and the Rydberg excited state.  This imprints a tunable, geometric phase $\phi$, which is determined by the detuning $\Delta$ and Rabi frequency $\Omega$.
		\textbf{d,} Multi-qubit controlled-phase rotations are  implemented via a sequence of $\pi$-pulses between the $\ket{0}\leftrightarrow \ket{r}$ transition of control atoms (yellow) and off-resonant Rabi oscillations on the target atoms (orange).  }
		\label{fig:blockade_gates}
	\end{figure}

	\subsection{Experimental implementation}
	\label{subsec:rydberg}

	Motivated by recent advances in the creation and control of many-body entanglement in programmable quantum systems~\cite{zhang_observation_2017,arute_quantum_2019,scholl2020programmable,ebadi2020quantum}, we propose an experimental implementation of our interactive protocol based upon neutral atoms coupled to Rydberg states~\cite{browaeys2020many}.
	We envision a three dimensional system of either alkali or alkaline-earth atoms trapped in an optical lattice or optical  tweezer array [Fig.~\ref{fig:blockade_gates}(a)]~\cite{PhysRevLett.115.043003,wang2016single,kumar2018sorting}.
To be specific, we consider $^{87}\rm Rb$ with an effective qubit degree of freedom encoded in hyperfine states:~$\ket{0} = \ket{F=1,m_F=0}$ and $\ket{1} = \ket{F=2,m_F=0}$.
Gates between atoms are mediated by coupling to a highly-excited Rydberg state $\ket{r}$, whose
large polarizability leads to strong van der Waals interactions.
This microscopic interaction enables the so-called Rydberg ``blockade'' mechanism---when a single atom is driven to its Rydberg state, all other atoms within a blockade radius, $R_b$, become off-resonant from the drive, thereby suppressing their excitation [Fig.~\ref{fig:blockade_gates}(a,b)]~\cite{saffman_quantum_2016}.

Somewhat remarkably, this blockade interaction enables the \emph{native} implementation of all multi-qubit-controlled phase gates depicted in the circuits in Figure \ref{fig:circuit-1}.
In particular,  consider the goal of applying a $C^kR^\ell_\phi$ gate; this gate applies phase rotations, $\{\phi_1, \phi_2, \dots, \phi_\ell\}$, to target qubits $\{j_1, j_2, \dots j_\ell\}$ if all $k$ control qubits $\{i_1, i_2, \dots i_k\}$ are in the $\ket{1}$ state [Fig.~\ref{fig:blockade_gates}(d)].
Experimentally, this can be implemented as follows: (i) sequentially apply (in any order) resonant $\pi$-pulses on the $\ket{0}\leftrightarrow \ket{r}$ transition  for the  $k$ desired control atoms, (ii) off-resonantly drive the $\ket{1} \leftrightarrow \ket{r}$ transition of each target atom with detuning $\Delta$ and Rabi frequency $\Omega$ for a time duration $T = 2\pi /(\Omega^2 + \Delta^2)^{1/2}$ [Fig.~\ref{fig:blockade_gates}(c)], (iii) sequentially apply [in the opposite order as in (i)] resonant $-\pi$-pulses (i.e.~$\pi$-pulses with the opposite  phase)  to the $k$  control atoms to bring them back to their original state.
	The intuition for why this experimental sequence implements the $C^kR^\ell_\phi$ gate is straightforward.   The first step creates a blockade if any of the control qubits are in the $\ket{0}$ state, while the second step imprints a phase, $\phi = \pi ( 1- \Delta/\sqrt{\Delta^2 + \Omega^2})$, on the $\ket{1}$ state, only in the absence of a  blockade.
	Note that tuning the values of $\phi_i$ for each of the target qubits simply corresponds to adjusting the detuning and Rabi frequency of the off-resonant drive in the second step [Fig.~\ref{fig:blockade_gates}(c,d)].

	Demonstrations of our protocol can already be implemented in current generation Rydberg experiments, where a number of essential features have recently been shown, including: 1) the coherent manipulation of individual qubits trapped in a 3D tweezer array~\cite{PhysRevLett.115.043003,wang2016single}, 2) the deterministic loading of atoms in a 3D optical lattice~\cite{kumar2018sorting}, and 3) fast entangling gate operations with fidelities, $F \geq 0.974$~\cite{levine_parallel_2019, graham_rydberg-mediated_2019, madjarov_high-fidelity_2020}.
	In order to estimate the number of entangling gates achievable within decoherence time scales,  let us imagine choosing a Rydberg state with a principal  quantum number $n\approx 70$.
	This yields a strong van der Waals interaction, $V(\vec{r}) = C_6/r^6$, with a $C_6$ coefficient $\sim(2\pi)\,880$~GHz$\cdot\mu$m$^6$~\cite{low2012experimental}.
Combined with a  coherent driving field of Rabi frequency $\Omega \sim (2\pi)\,1-10$~MHz, the  van der Waals interaction can lead to a blockade radius of up to, $R_b = \left( C_6 /\Omega\right)^{1/6} \sim 10\,\mu$m. 
Within this radius, one can arrange $\sim 10^2$ all-to-all interacting qubits, assuming an atom-to-atom spacing of approximately, $a_0\approx 2\mu$m \footnote{We note that this spacing is ultimately limited by a combination of the optical diffraction limit and the orbital size of $n\approx 70$ Rydberg states.}.
In current experiments, the decoherence associated with the Rydberg transition is typically limited by a combination of inhomogeneous Doppler shifts and laser phase/intensity noise, leading to $1/T_2 \sim 10-100$~kHz \cite{de_leseleuc_analysis_2018, levine_parallel_2019, liu_influence_2020}.
Taking everything together, one should be able to perform $\sim 10^3$ entangling gates before decoherence occurs (this is comparable to the number of two-qubit entangling gates possible in other state-of-the-art platforms~\cite{arute_quantum_2019, schafer_fast_2018}).
While this falls short of enabling an immediate full-scale demonstration of classically verifiable quantum advantage, we hasten to emphasize that the ability to directly perform multi-qubit entangling operations significantly reduces the cost of implementing our interactive protocol.
For example,  the standard decomposition of a Toffoli gate uses 6 CNOT gates and 7 $T$ and $T^\dag$ gates, with a gate depth of 12 \cite{nielsen_quantum_2011, shende_cnot-cost_2009, barenco_elementary_1995};  an equivalent three qubit gate can be performed in a single step via the Rydberg blockade mechanism.

	\section{Conclusion and Outlook}
	\label{sec:outlook}

	The interplay between classical and quantum complexities ultimately determines the threshold for any quantum advantage scheme.
	Here, we have proposed a novel interactive protocol for classically verifiable quantum advantage based upon trapdoor claw-free functions; in addition to proposing two new TCFs  [Table~\ref{tab:tcfs}], we also provide explicit quantum circuits that leverage the microscopic interactions present in a Rydberg-based quantum computer.
	Our work allows near-term quantum devices to move one step closer toward a loophole-free demonstration of quantum advantage and also opens the door to a number of promising future directions.

	First, our proof of soundness only applies to classical adversaries; whether it is possible to extend our  protocol's security to quantum adversaries remains an  open question~\footnote{By ``secure against quantum adversaries'' we mean that it is possible to ensure that the quantum prover is actually following the prescribed protocol, to whatever extent is necessary for the desired cryptographic task.}. 
A quantum-secure proof could enable our protocol's use in a number of applications, such as certifiable random number generation~\cite{brakerski_cryptographic_2019} and the verification of arbitrary quantum computations~\cite{mahadev_classical_2018}.
	Second, our work motivates the search for new trapdoor claw-free functions, which can be evaluated in the smallest possible quantum volume.
	Cryptographic primitives such as Learning Parity with Noise (LPN), which are designed for use in low-power  devices such as RFID cards, represent a promising path forward~\cite{pietrzak_cryptography_2012}.
	More broadly, one could also attempt to build modified protocols, which simplify either  the requirements on the cryptographic function or the interactions; 	interestingly, recent work has demonstrated that using random oracles can remove the need for interactions in a TCF-based proof of quantumness~\cite{brakerski_simpler_2020}. 
	Finally, while we have focused our experimental discussions on Rydberg atoms, a number of other  platforms also exhibit  features that facilitate the protocol's implementation.
	For example,  both trapped ions and cavity-QED systems can allow all-to-all connectivity, while superconducting qubits can be engineered to have biased noise~\cite{puri_bias-preserving_2020}.
	This latter feature would allow noise to be concentrated into error modes detectable by our proposed post-selection scheme.

        We gratefully acknowledge the insights of and discussions with A.~Bouland, S.~Garg, A.~Gheorghiu, Z.~Landau, L.~Lewis, and T.~Vidick. 
    We are particularly indebted to Joonhee Choi for insights about Rydberg-based quantum computing.
	This work was supported by the NSF QLCI program through grant number OMA-2016245, and the DOD through MURI grant number FA9550-18-1-0161.
	NYY acknowledges support from the David and Lucile Packard foundation and a Google research award. 
	GDKM acknowledges support from the Department of Defense (DOD) through the National Defense Science \& Engineering Graduate Fellowship (NDSEG) Program.
	SC acknowledges support from the Miller Institute for Basic Research in Science.

	\bibliography{refs}

\begin{thebibliography}{89}%
\makeatletter
\providecommand \@ifxundefined [1]{%
 \@ifx{#1\undefined}
}%
\providecommand \@ifnum [1]{%
 \ifnum #1\expandafter \@firstoftwo
 \else \expandafter \@secondoftwo
 \fi
}%
\providecommand \@ifx [1]{%
 \ifx #1\expandafter \@firstoftwo
 \else \expandafter \@secondoftwo
 \fi
}%
\providecommand \natexlab [1]{#1}%
\providecommand \enquote  [1]{``#1''}%
\providecommand \bibnamefont  [1]{#1}%
\providecommand \bibfnamefont [1]{#1}%
\providecommand \citenamefont [1]{#1}%
\providecommand \href@noop [0]{\@secondoftwo}%
\providecommand \href [0]{\begingroup \@sanitize@url \@href}%
\providecommand \@href[1]{\@@startlink{#1}\@@href}%
\providecommand \@@href[1]{\endgroup#1\@@endlink}%
\providecommand \@sanitize@url [0]{\catcode `\\12\catcode `\$12\catcode
  `\&12\catcode `\#12\catcode `\^12\catcode `\_12\catcode `\%12\relax}%
\providecommand \@@startlink[1]{}%
\providecommand \@@endlink[0]{}%
\providecommand \url  [0]{\begingroup\@sanitize@url \@url }%
\providecommand \@url [1]{\endgroup\@href {#1}{\urlprefix }}%
\providecommand \urlprefix  [0]{URL }%
\providecommand \Eprint [0]{\href }%
\providecommand \doibase [0]{https://doi.org/}%
\providecommand \selectlanguage [0]{\@gobble}%
\providecommand \bibinfo  [0]{\@secondoftwo}%
\providecommand \bibfield  [0]{\@secondoftwo}%
\providecommand \translation [1]{[#1]}%
\providecommand \BibitemOpen [0]{}%
\providecommand \bibitemStop [0]{}%
\providecommand \bibitemNoStop [0]{.\EOS\space}%
\providecommand \EOS [0]{\spacefactor3000\relax}%
\providecommand \BibitemShut  [1]{\csname bibitem#1\endcsname}%
\let\auto@bib@innerbib\@empty
\bibitem [{\citenamefont {Aaronson}\ and\ \citenamefont
  {Arkhipov}(2011)}]{aaronson_computational_2011}%
  \BibitemOpen
  \bibfield  {author} {\bibinfo {author} {\bibfnamefont {S.}~\bibnamefont
  {Aaronson}}\ and\ \bibinfo {author} {\bibfnamefont {A.}~\bibnamefont
  {Arkhipov}},\ }in\ \href {https://doi.org/10.1145/1993636.1993682} {\emph
  {\bibinfo {booktitle} {Proceedings of the forty-third annual {ACM} symposium
  on {Theory} of computing}}},\ \bibinfo {series and number} {{STOC} '11}\
  (\bibinfo  {publisher} {Association for Computing Machinery},\ \bibinfo
  {address} {New York, NY, USA},\ \bibinfo {year} {2011})\ pp.\ \bibinfo
  {pages} {333--342}\BibitemShut {NoStop}%
\bibitem [{\citenamefont {Farhi}\ and\ \citenamefont
  {Harrow}(2016)}]{farhi2016quantum}%
  \BibitemOpen
  \bibfield  {author} {\bibinfo {author} {\bibfnamefont {E.}~\bibnamefont
  {Farhi}}\ and\ \bibinfo {author} {\bibfnamefont {A.~W.}\ \bibnamefont
  {Harrow}},\ }\href@noop {} {\bibfield  {journal} {\bibinfo  {journal}
  {arXiv:1602.07674}\ } (\bibinfo {year} {2016})}\BibitemShut {NoStop}%
\bibitem [{\citenamefont {Bremner}\ \emph {et~al.}(2016)\citenamefont
  {Bremner}, \citenamefont {Montanaro},\ and\ \citenamefont
  {Shepherd}}]{bremner_average-case_2016}%
  \BibitemOpen
  \bibfield  {author} {\bibinfo {author} {\bibfnamefont {M.~J.}\ \bibnamefont
  {Bremner}}, \bibinfo {author} {\bibfnamefont {A.}~\bibnamefont {Montanaro}},\
  and\ \bibinfo {author} {\bibfnamefont {D.~J.}\ \bibnamefont {Shepherd}},\
  }\href {https://doi.org/10.1103/PhysRevLett.117.080501} {\bibfield  {journal}
  {\bibinfo  {journal} {Physical Review Letters}\ }\textbf {\bibinfo {volume}
  {117}},\ \bibinfo {pages} {080501} (\bibinfo {year} {2016})}\BibitemShut
  {NoStop}%
\bibitem [{\citenamefont {Lund}\ \emph {et~al.}(2017)\citenamefont {Lund},
  \citenamefont {Bremner},\ and\ \citenamefont {Ralph}}]{lund2017quantum}%
  \BibitemOpen
  \bibfield  {author} {\bibinfo {author} {\bibfnamefont {A.~P.}\ \bibnamefont
  {Lund}}, \bibinfo {author} {\bibfnamefont {M.~J.}\ \bibnamefont {Bremner}},\
  and\ \bibinfo {author} {\bibfnamefont {T.~C.}\ \bibnamefont {Ralph}},\
  }\href@noop {} {\bibfield  {journal} {\bibinfo  {journal} {npj Quantum
  Information}\ }\textbf {\bibinfo {volume} {3}},\ \bibinfo {pages} {1}
  (\bibinfo {year} {2017})}\BibitemShut {NoStop}%
\bibitem [{\citenamefont {Harrow}\ and\ \citenamefont
  {Montanaro}(2017)}]{harrow2017quantum}%
  \BibitemOpen
  \bibfield  {author} {\bibinfo {author} {\bibfnamefont {A.~W.}\ \bibnamefont
  {Harrow}}\ and\ \bibinfo {author} {\bibfnamefont {A.}~\bibnamefont
  {Montanaro}},\ }\href@noop {} {\bibfield  {journal} {\bibinfo  {journal}
  {Nature}\ }\textbf {\bibinfo {volume} {549}},\ \bibinfo {pages} {203}
  (\bibinfo {year} {2017})}\BibitemShut {NoStop}%
\bibitem [{\citenamefont {Terhal}(2018)}]{terhal2018quantum}%
  \BibitemOpen
  \bibfield  {author} {\bibinfo {author} {\bibfnamefont {B.~M.}\ \bibnamefont
  {Terhal}},\ }\href@noop {} {\bibfield  {journal} {\bibinfo  {journal} {Nature
  Physics}\ }\textbf {\bibinfo {volume} {14}},\ \bibinfo {pages} {530}
  (\bibinfo {year} {2018})}\BibitemShut {NoStop}%
\bibitem [{\citenamefont {Boixo}\ \emph {et~al.}(2018)\citenamefont {Boixo},
  \citenamefont {Isakov}, \citenamefont {Smelyanskiy}, \citenamefont {Babbush},
  \citenamefont {Ding}, \citenamefont {Jiang}, \citenamefont {Bremner},
  \citenamefont {Martinis},\ and\ \citenamefont
  {Neven}}]{boixo_characterizing_2018}%
  \BibitemOpen
  \bibfield  {author} {\bibinfo {author} {\bibfnamefont {S.}~\bibnamefont
  {Boixo}}, \bibinfo {author} {\bibfnamefont {S.~V.}\ \bibnamefont {Isakov}},
  \bibinfo {author} {\bibfnamefont {V.~N.}\ \bibnamefont {Smelyanskiy}},
  \bibinfo {author} {\bibfnamefont {R.}~\bibnamefont {Babbush}}, \bibinfo
  {author} {\bibfnamefont {N.}~\bibnamefont {Ding}}, \bibinfo {author}
  {\bibfnamefont {Z.}~\bibnamefont {Jiang}}, \bibinfo {author} {\bibfnamefont
  {M.~J.}\ \bibnamefont {Bremner}}, \bibinfo {author} {\bibfnamefont {J.~M.}\
  \bibnamefont {Martinis}},\ and\ \bibinfo {author} {\bibfnamefont
  {H.}~\bibnamefont {Neven}},\ }\href
  {https://doi.org/10.1038/s41567-018-0124-x} {\bibfield  {journal} {\bibinfo
  {journal} {Nature Physics}\ }\textbf {\bibinfo {volume} {14}},\ \bibinfo
  {pages} {595} (\bibinfo {year} {2018})}\BibitemShut {NoStop}%
\bibitem [{\citenamefont {Bouland}\ \emph {et~al.}(2019)\citenamefont
  {Bouland}, \citenamefont {Fefferman}, \citenamefont {Nirkhe},\ and\
  \citenamefont {Vazirani}}]{bouland_complexity_2019}%
  \BibitemOpen
  \bibfield  {author} {\bibinfo {author} {\bibfnamefont {A.}~\bibnamefont
  {Bouland}}, \bibinfo {author} {\bibfnamefont {B.}~\bibnamefont {Fefferman}},
  \bibinfo {author} {\bibfnamefont {C.}~\bibnamefont {Nirkhe}},\ and\ \bibinfo
  {author} {\bibfnamefont {U.}~\bibnamefont {Vazirani}},\ }\href
  {https://doi.org/10.1038/s41567-018-0318-2} {\bibfield  {journal} {\bibinfo
  {journal} {Nature Physics}\ }\textbf {\bibinfo {volume} {15}},\ \bibinfo
  {pages} {159} (\bibinfo {year} {2019})}\BibitemShut {NoStop}%
\bibitem [{\citenamefont {Aaronson}\ and\ \citenamefont
  {Chen}(2017)}]{aaronson_complexity-theoretic_2017}%
  \BibitemOpen
  \bibfield  {author} {\bibinfo {author} {\bibfnamefont {S.}~\bibnamefont
  {Aaronson}}\ and\ \bibinfo {author} {\bibfnamefont {L.}~\bibnamefont
  {Chen}},\ }in\ \href {https://doi.org/10.4230/LIPIcs.CCC.2017.22} {\emph
  {\bibinfo {booktitle} {32nd {Computational} {Complexity} {Conference} ({CCC}
  2017)}}},\ \bibinfo {series} {Leibniz {International} {Proceedings} in
  {Informatics} ({LIPIcs})}, Vol.~\bibinfo {volume} {79},\ \bibinfo {editor}
  {edited by\ \bibinfo {editor} {\bibfnamefont {R.}~\bibnamefont {O'Donnell}}}\
  (\bibinfo  {publisher} {Schloss Dagstuhl–Leibniz-Zentrum fuer Informatik},\
  \bibinfo {address} {Dagstuhl, Germany},\ \bibinfo {year} {2017})\ pp.\
  \bibinfo {pages} {22:1--22:67}\BibitemShut {NoStop}%
\bibitem [{\citenamefont {Neill}\ \emph {et~al.}(2018)\citenamefont {Neill},
  \citenamefont {Roushan}, \citenamefont {Kechedzhi}, \citenamefont {Boixo},
  \citenamefont {Isakov}, \citenamefont {Smelyanskiy}, \citenamefont {Megrant},
  \citenamefont {Chiaro}, \citenamefont {Dunsworth}, \citenamefont {Arya},
  \citenamefont {Barends}, \citenamefont {Burkett}, \citenamefont {Chen},
  \citenamefont {Chen} \emph {et~al.}}]{neill_blueprint_2018}%
  \BibitemOpen
  \bibfield  {author} {\bibinfo {author} {\bibfnamefont {C.}~\bibnamefont
  {Neill}}, \bibinfo {author} {\bibfnamefont {P.}~\bibnamefont {Roushan}},
  \bibinfo {author} {\bibfnamefont {K.}~\bibnamefont {Kechedzhi}}, \bibinfo
  {author} {\bibfnamefont {S.}~\bibnamefont {Boixo}}, \bibinfo {author}
  {\bibfnamefont {S.~V.}\ \bibnamefont {Isakov}}, \bibinfo {author}
  {\bibfnamefont {V.}~\bibnamefont {Smelyanskiy}}, \bibinfo {author}
  {\bibfnamefont {A.}~\bibnamefont {Megrant}}, \bibinfo {author} {\bibfnamefont
  {B.}~\bibnamefont {Chiaro}}, \bibinfo {author} {\bibfnamefont
  {A.}~\bibnamefont {Dunsworth}}, \bibinfo {author} {\bibfnamefont
  {K.}~\bibnamefont {Arya}}, \bibinfo {author} {\bibfnamefont {R.}~\bibnamefont
  {Barends}}, \bibinfo {author} {\bibfnamefont {B.}~\bibnamefont {Burkett}},
  \bibinfo {author} {\bibfnamefont {Y.}~\bibnamefont {Chen}}, \bibinfo {author}
  {\bibfnamefont {Z.}~\bibnamefont {Chen}}, \emph {et~al.},\ }\href
  {https://doi.org/10.1126/science.aao4309} {\bibfield  {journal} {\bibinfo
  {journal} {Science}\ }\textbf {\bibinfo {volume} {360}},\ \bibinfo {pages}
  {195} (\bibinfo {year} {2018})}\BibitemShut {NoStop}%
\bibitem [{\citenamefont {Arute}\ \emph {et~al.}(2019)\citenamefont {Arute},
  \citenamefont {Arya}, \citenamefont {Babbush}, \citenamefont {Bacon},
  \citenamefont {Bardin}, \citenamefont {Barends}, \citenamefont {Biswas},
  \citenamefont {Boixo}, \citenamefont {Brandao}, \citenamefont {Buell},
  \citenamefont {Burkett}, \citenamefont {Chen}, \citenamefont {Chen},
  \citenamefont {Chiaro} \emph {et~al.}}]{arute_quantum_2019}%
  \BibitemOpen
  \bibfield  {author} {\bibinfo {author} {\bibfnamefont {F.}~\bibnamefont
  {Arute}}, \bibinfo {author} {\bibfnamefont {K.}~\bibnamefont {Arya}},
  \bibinfo {author} {\bibfnamefont {R.}~\bibnamefont {Babbush}}, \bibinfo
  {author} {\bibfnamefont {D.}~\bibnamefont {Bacon}}, \bibinfo {author}
  {\bibfnamefont {J.~C.}\ \bibnamefont {Bardin}}, \bibinfo {author}
  {\bibfnamefont {R.}~\bibnamefont {Barends}}, \bibinfo {author} {\bibfnamefont
  {R.}~\bibnamefont {Biswas}}, \bibinfo {author} {\bibfnamefont
  {S.}~\bibnamefont {Boixo}}, \bibinfo {author} {\bibfnamefont {F.~G. S.~L.}\
  \bibnamefont {Brandao}}, \bibinfo {author} {\bibfnamefont {D.~A.}\
  \bibnamefont {Buell}}, \bibinfo {author} {\bibfnamefont {B.}~\bibnamefont
  {Burkett}}, \bibinfo {author} {\bibfnamefont {Y.}~\bibnamefont {Chen}},
  \bibinfo {author} {\bibfnamefont {Z.}~\bibnamefont {Chen}}, \bibinfo {author}
  {\bibfnamefont {B.}~\bibnamefont {Chiaro}}, \emph {et~al.},\ }\href
  {https://doi.org/10.1038/s41586-019-1666-5} {\bibfield  {journal} {\bibinfo
  {journal} {Nature}\ }\textbf {\bibinfo {volume} {574}},\ \bibinfo {pages}
  {505} (\bibinfo {year} {2019})}\BibitemShut {NoStop}%
\bibitem [{\citenamefont {Zhong}\ \emph {et~al.}(2020)\citenamefont {Zhong},
  \citenamefont {Wang}, \citenamefont {Deng}, \citenamefont {Chen},
  \citenamefont {Peng}, \citenamefont {Luo}, \citenamefont {Qin}, \citenamefont
  {Wu}, \citenamefont {Ding}, \citenamefont {Hu}, \citenamefont {Hu},
  \citenamefont {Yang}, \citenamefont {Zhang}, \citenamefont {Li} \emph
  {et~al.}}]{zhong_quantum_2020}%
  \BibitemOpen
  \bibfield  {author} {\bibinfo {author} {\bibfnamefont {H.-S.}\ \bibnamefont
  {Zhong}}, \bibinfo {author} {\bibfnamefont {H.}~\bibnamefont {Wang}},
  \bibinfo {author} {\bibfnamefont {Y.-H.}\ \bibnamefont {Deng}}, \bibinfo
  {author} {\bibfnamefont {M.-C.}\ \bibnamefont {Chen}}, \bibinfo {author}
  {\bibfnamefont {L.-C.}\ \bibnamefont {Peng}}, \bibinfo {author}
  {\bibfnamefont {Y.-H.}\ \bibnamefont {Luo}}, \bibinfo {author} {\bibfnamefont
  {J.}~\bibnamefont {Qin}}, \bibinfo {author} {\bibfnamefont {D.}~\bibnamefont
  {Wu}}, \bibinfo {author} {\bibfnamefont {X.}~\bibnamefont {Ding}}, \bibinfo
  {author} {\bibfnamefont {Y.}~\bibnamefont {Hu}}, \bibinfo {author}
  {\bibfnamefont {P.}~\bibnamefont {Hu}}, \bibinfo {author} {\bibfnamefont
  {X.-Y.}\ \bibnamefont {Yang}}, \bibinfo {author} {\bibfnamefont {W.-J.}\
  \bibnamefont {Zhang}}, \bibinfo {author} {\bibfnamefont {H.}~\bibnamefont
  {Li}}, \emph {et~al.},\ }\href {https://doi.org/10.1126/science.abe8770}
  {\bibfield  {journal} {\bibinfo  {journal} {Science}\ }\textbf {\bibinfo
  {volume} {370}},\ \bibinfo {pages} {1460} (\bibinfo {year}
  {2020})}\BibitemShut {NoStop}%
\bibitem [{\citenamefont {Bravyi}\ \emph {et~al.}(2018)\citenamefont {Bravyi},
  \citenamefont {Gosset},\ and\ \citenamefont {König}}]{bravyi_quantum_2018}%
  \BibitemOpen
  \bibfield  {author} {\bibinfo {author} {\bibfnamefont {S.}~\bibnamefont
  {Bravyi}}, \bibinfo {author} {\bibfnamefont {D.}~\bibnamefont {Gosset}},\
  and\ \bibinfo {author} {\bibfnamefont {R.}~\bibnamefont {König}},\ }\href
  {https://doi.org/10.1126/science.aar3106} {\bibfield  {journal} {\bibinfo
  {journal} {Science}\ }\textbf {\bibinfo {volume} {362}},\ \bibinfo {pages}
  {308} (\bibinfo {year} {2018})}\BibitemShut {NoStop}%
\bibitem [{\citenamefont {Bravyi}\ \emph {et~al.}(2019)\citenamefont {Bravyi},
  \citenamefont {Gosset}, \citenamefont {Koenig},\ and\ \citenamefont
  {Tomamichel}}]{bravyi_quantum_2019}%
  \BibitemOpen
  \bibfield  {author} {\bibinfo {author} {\bibfnamefont {S.}~\bibnamefont
  {Bravyi}}, \bibinfo {author} {\bibfnamefont {D.}~\bibnamefont {Gosset}},
  \bibinfo {author} {\bibfnamefont {R.}~\bibnamefont {Koenig}},\ and\ \bibinfo
  {author} {\bibfnamefont {M.}~\bibnamefont {Tomamichel}},\ }\href
  {http://arxiv.org/abs/1904.01502} {\bibfield  {journal} {\bibinfo  {journal}
  {arXiv:1904.01502 [quant-ph]}\ } (\bibinfo {year} {2019})}\BibitemShut
  {NoStop}%
\bibitem [{Note1()}]{Note1}%
  \BibitemOpen
  \bibinfo {note} {They also have some other caveats: a provable speedup of
  $\protect \mathcal {O}(1)$ quantum complexity over $\protect \mathcal {O}(n)$
  classical complexity is promising, but just reading the input may require
  $\protect \mathcal {O}(n)$ time, hiding the computational speedup in
  practice.}\BibitemShut {Stop}%
\bibitem [{\citenamefont {Shor}(1997)}]{shor_polynomial-time_1997}%
  \BibitemOpen
  \bibfield  {author} {\bibinfo {author} {\bibfnamefont {P.~W.}\ \bibnamefont
  {Shor}},\ }\href {https://doi.org/10.1137/S0097539795293172} {\bibfield
  {journal} {\bibinfo  {journal} {SIAM Journal on Computing}\ }\textbf
  {\bibinfo {volume} {26}},\ \bibinfo {pages} {1484} (\bibinfo {year}
  {1997})}\BibitemShut {NoStop}%
\bibitem [{\citenamefont {Brakerski}\ \emph {et~al.}(2019)\citenamefont
  {Brakerski}, \citenamefont {Christiano}, \citenamefont {Mahadev},
  \citenamefont {Vazirani},\ and\ \citenamefont
  {Vidick}}]{brakerski_cryptographic_2019}%
  \BibitemOpen
  \bibfield  {author} {\bibinfo {author} {\bibfnamefont {Z.}~\bibnamefont
  {Brakerski}}, \bibinfo {author} {\bibfnamefont {P.}~\bibnamefont
  {Christiano}}, \bibinfo {author} {\bibfnamefont {U.}~\bibnamefont {Mahadev}},
  \bibinfo {author} {\bibfnamefont {U.}~\bibnamefont {Vazirani}},\ and\
  \bibinfo {author} {\bibfnamefont {T.}~\bibnamefont {Vidick}},\ }\href
  {http://arxiv.org/abs/1804.00640} {\bibfield  {journal} {\bibinfo  {journal}
  {arXiv:1804.00640 [quant-ph]}\ } (\bibinfo {year} {2019})}\BibitemShut
  {NoStop}%
\bibitem [{\citenamefont {Brakerski}\ \emph {et~al.}(2020)\citenamefont
  {Brakerski}, \citenamefont {Koppula}, \citenamefont {Vazirani},\ and\
  \citenamefont {Vidick}}]{brakerski_simpler_2020}%
  \BibitemOpen
  \bibfield  {author} {\bibinfo {author} {\bibfnamefont {Z.}~\bibnamefont
  {Brakerski}}, \bibinfo {author} {\bibfnamefont {V.}~\bibnamefont {Koppula}},
  \bibinfo {author} {\bibfnamefont {U.}~\bibnamefont {Vazirani}},\ and\
  \bibinfo {author} {\bibfnamefont {T.}~\bibnamefont {Vidick}},\ }\href
  {http://arxiv.org/abs/2005.04826} {\bibfield  {journal} {\bibinfo  {journal}
  {arXiv:2005.04826 [quant-ph]}\ } (\bibinfo {year} {2020})}\BibitemShut
  {NoStop}%
\bibitem [{\citenamefont {Aharonov}\ \emph {et~al.}(2017)\citenamefont
  {Aharonov}, \citenamefont {Ben-Or}, \citenamefont {Eban},\ and\ \citenamefont
  {Mahadev}}]{aharonov2017interactive}%
  \BibitemOpen
  \bibfield  {author} {\bibinfo {author} {\bibfnamefont {D.}~\bibnamefont
  {Aharonov}}, \bibinfo {author} {\bibfnamefont {M.}~\bibnamefont {Ben-Or}},
  \bibinfo {author} {\bibfnamefont {E.}~\bibnamefont {Eban}},\ and\ \bibinfo
  {author} {\bibfnamefont {U.}~\bibnamefont {Mahadev}},\ }\href@noop {}
  {\bibfield  {journal} {\bibinfo  {journal} {arXiv:1704.04487}\ } (\bibinfo
  {year} {2017})}\BibitemShut {NoStop}%
\bibitem [{\citenamefont {Watrous}(1999)}]{watrous_pspace_1999}%
  \BibitemOpen
  \bibfield  {author} {\bibinfo {author} {\bibfnamefont {J.}~\bibnamefont
  {Watrous}},\ }\href {http://arxiv.org/abs/cs/9901015} {\bibfield  {journal}
  {\bibinfo  {journal} {arXiv:cs/9901015}\ } (\bibinfo {year}
  {1999})}\BibitemShut {NoStop}%
\bibitem [{\citenamefont {Kitaev}\ and\ \citenamefont
  {Watrous}(2000)}]{kitaev2000parallelization}%
  \BibitemOpen
  \bibfield  {author} {\bibinfo {author} {\bibfnamefont {A.}~\bibnamefont
  {Kitaev}}\ and\ \bibinfo {author} {\bibfnamefont {J.}~\bibnamefont
  {Watrous}},\ }in\ \href@noop {} {\emph {\bibinfo {booktitle} {Proceedings of
  the thirty-second annual ACM symposium on Theory of computing}}}\ (\bibinfo
  {year} {2000})\ pp.\ \bibinfo {pages} {608--617}\BibitemShut {NoStop}%
\bibitem [{\citenamefont {Kobayashi}\ and\ \citenamefont
  {Matsumoto}(2003)}]{kobayashi2003quantum}%
  \BibitemOpen
  \bibfield  {author} {\bibinfo {author} {\bibfnamefont {H.}~\bibnamefont
  {Kobayashi}}\ and\ \bibinfo {author} {\bibfnamefont {K.}~\bibnamefont
  {Matsumoto}},\ }\href@noop {} {\bibfield  {journal} {\bibinfo  {journal}
  {Journal of Computer and System Sciences}\ }\textbf {\bibinfo {volume}
  {66}},\ \bibinfo {pages} {429} (\bibinfo {year} {2003})}\BibitemShut
  {NoStop}%
\bibitem [{\citenamefont {Fitzsimons}\ and\ \citenamefont
  {Vidick}(2015)}]{fitzsimons2015multiprover}%
  \BibitemOpen
  \bibfield  {author} {\bibinfo {author} {\bibfnamefont {J.}~\bibnamefont
  {Fitzsimons}}\ and\ \bibinfo {author} {\bibfnamefont {T.}~\bibnamefont
  {Vidick}},\ }in\ \href@noop {} {\emph {\bibinfo {booktitle} {Proceedings of
  the 2015 Conference on Innovations in Theoretical Computer Science}}}\
  (\bibinfo {year} {2015})\ pp.\ \bibinfo {pages} {103--112}\BibitemShut
  {NoStop}%
\bibitem [{\citenamefont {Markov}\ \emph {et~al.}(2018)\citenamefont {Markov},
  \citenamefont {Fatima}, \citenamefont {Isakov},\ and\ \citenamefont
  {Boixo}}]{markov2018quantum}%
  \BibitemOpen
  \bibfield  {author} {\bibinfo {author} {\bibfnamefont {I.~L.}\ \bibnamefont
  {Markov}}, \bibinfo {author} {\bibfnamefont {A.}~\bibnamefont {Fatima}},
  \bibinfo {author} {\bibfnamefont {S.~V.}\ \bibnamefont {Isakov}},\ and\
  \bibinfo {author} {\bibfnamefont {S.}~\bibnamefont {Boixo}},\ }\href@noop {}
  {\bibfield  {journal} {\bibinfo  {journal} {arXiv:1807.10749}\ } (\bibinfo
  {year} {2018})}\BibitemShut {NoStop}%
\bibitem [{\citenamefont {Regev}(2005)}]{regev_lattices_2005}%
  \BibitemOpen
  \bibfield  {author} {\bibinfo {author} {\bibfnamefont {O.}~\bibnamefont
  {Regev}},\ }in\ \href {https://doi.org/10.1145/1060590.1060603} {\emph
  {\bibinfo {booktitle} {Proceedings of the thirty-seventh annual {ACM}
  symposium on {Theory} of computing}}},\ \bibinfo {series and number} {{STOC}
  '05}\ (\bibinfo  {publisher} {Association for Computing Machinery},\ \bibinfo
  {address} {New York, NY, USA},\ \bibinfo {year} {2005})\ pp.\ \bibinfo
  {pages} {84--93}\BibitemShut {NoStop}%
\bibitem [{Note2()}]{Note2}%
  \BibitemOpen
  \bibinfo {note} {``Claw-free'' is often used to refer to a \protect \emph
  {pair} of functions $f_0, f_1$ such that for appropriate $x_0, x_1$ we have
  $f_0(x_0) = f_1(x_1)$. Here, we use the slightly more general idea of a
  single 2-to-1 function $f$ for which it is hard to find $x_0, x_1$ such that
  $f(x_0) = f(x_1)$. This is a special case of a ``collision-resistant
  function,'' which could potentially be many-to-one. We also note that a
  claw-free pair of functions can be converted into a single claw-free function
  by defining $f(b || x) = f_b(x)$, where $||$ denotes
  concatenation.}\BibitemShut {Stop}%
\bibitem [{SM()}]{SM}%
  \BibitemOpen
  \href@noop {} {}\bibinfo {note} {See Supplementary Information for additional
  details and supporting derivations.}\BibitemShut {Stop}%
\bibitem [{Note3()}]{Note3}%
  \BibitemOpen
  \bibinfo {note} {To be precise, it is hard to find both $x_0$ and the parity
  of any subset of the bits of $x_1$.}\BibitemShut {Stop}%
\bibitem [{\citenamefont {Bell}(1964)}]{bell_einstein_1964}%
  \BibitemOpen
  \bibfield  {author} {\bibinfo {author} {\bibfnamefont {J.~S.}\ \bibnamefont
  {Bell}},\ }\href {https://doi.org/10.1103/PhysicsPhysiqueFizika.1.195}
  {\bibfield  {journal} {\bibinfo  {journal} {Physics Physique Fizika}\
  }\textbf {\bibinfo {volume} {1}},\ \bibinfo {pages} {195} (\bibinfo {year}
  {1964})}\BibitemShut {NoStop}%
\bibitem [{\citenamefont {Clauser}\ \emph {et~al.}(1969)\citenamefont
  {Clauser}, \citenamefont {Horne}, \citenamefont {Shimony},\ and\
  \citenamefont {Holt}}]{clauser_proposed_1969}%
  \BibitemOpen
  \bibfield  {author} {\bibinfo {author} {\bibfnamefont {J.~F.}\ \bibnamefont
  {Clauser}}, \bibinfo {author} {\bibfnamefont {M.~A.}\ \bibnamefont {Horne}},
  \bibinfo {author} {\bibfnamefont {A.}~\bibnamefont {Shimony}},\ and\ \bibinfo
  {author} {\bibfnamefont {R.~A.}\ \bibnamefont {Holt}},\ }\href
  {https://doi.org/10.1103/PhysRevLett.23.880} {\bibfield  {journal} {\bibinfo
  {journal} {Physical Review Letters}\ }\textbf {\bibinfo {volume} {23}},\
  \bibinfo {pages} {880} (\bibinfo {year} {1969})}\BibitemShut {NoStop}%
\bibitem [{\citenamefont {Diffie}\ and\ \citenamefont
  {Hellman}(1976)}]{diffie_new_1976}%
  \BibitemOpen
  \bibfield  {author} {\bibinfo {author} {\bibfnamefont {W.}~\bibnamefont
  {Diffie}}\ and\ \bibinfo {author} {\bibfnamefont {M.}~\bibnamefont
  {Hellman}},\ }\href {https://doi.org/10.1109/TIT.1976.1055638} {\bibfield
  {journal} {\bibinfo  {journal} {IEEE Transactions on Information Theory}\
  }\textbf {\bibinfo {volume} {22}},\ \bibinfo {pages} {644} (\bibinfo {year}
  {1976})}\BibitemShut {NoStop}%
\bibitem [{\citenamefont {Peikert}\ and\ \citenamefont
  {Waters}(2008)}]{peikert_lossy_2008}%
  \BibitemOpen
  \bibfield  {author} {\bibinfo {author} {\bibfnamefont {C.}~\bibnamefont
  {Peikert}}\ and\ \bibinfo {author} {\bibfnamefont {B.}~\bibnamefont
  {Waters}},\ }in\ \href {https://doi.org/10.1145/1374376.1374406} {\emph
  {\bibinfo {booktitle} {Proceedings of the fortieth annual {ACM} symposium on
  {Theory} of computing}}},\ \bibinfo {series and number} {{STOC} '08}\
  (\bibinfo  {publisher} {Association for Computing Machinery},\ \bibinfo
  {address} {New York, NY, USA},\ \bibinfo {year} {2008})\ pp.\ \bibinfo
  {pages} {187--196}\BibitemShut {NoStop}%
\bibitem [{\citenamefont {Freeman}\ \emph {et~al.}(2010)\citenamefont
  {Freeman}, \citenamefont {Goldreich}, \citenamefont {Kiltz}, \citenamefont
  {Rosen},\ and\ \citenamefont {Segev}}]{freeman_more_2010}%
  \BibitemOpen
  \bibfield  {author} {\bibinfo {author} {\bibfnamefont {D.~M.}\ \bibnamefont
  {Freeman}}, \bibinfo {author} {\bibfnamefont {O.}~\bibnamefont {Goldreich}},
  \bibinfo {author} {\bibfnamefont {E.}~\bibnamefont {Kiltz}}, \bibinfo
  {author} {\bibfnamefont {A.}~\bibnamefont {Rosen}},\ and\ \bibinfo {author}
  {\bibfnamefont {G.}~\bibnamefont {Segev}},\ }in\ \href
  {https://doi.org/10.1007/978-3-642-13013-7_17} {\emph {\bibinfo {booktitle}
  {Public {Key} {Cryptography} – {PKC} 2010}}},\ \bibinfo {series and number}
  {Lecture {Notes} in {Computer} {Science}},\ \bibinfo {editor} {edited by\
  \bibinfo {editor} {\bibfnamefont {P.~Q.}\ \bibnamefont {Nguyen}}\ and\
  \bibinfo {editor} {\bibfnamefont {D.}~\bibnamefont {Pointcheval}}}\ (\bibinfo
   {publisher} {Springer},\ \bibinfo {address} {Berlin, Heidelberg},\ \bibinfo
  {year} {2010})\ pp.\ \bibinfo {pages} {279--295}\BibitemShut {NoStop}%
\bibitem [{\citenamefont {Rabin}(1979)}]{rabin_digitalized_1979}%
  \BibitemOpen
  \bibfield  {author} {\bibinfo {author} {\bibfnamefont {M.~O.}\ \bibnamefont
  {Rabin}},\ }\href@noop {} {\emph {\bibinfo {title} {Digitalized signatures
  and public-key functions as intractable as factorization}}},\ \bibinfo {type}
  {Technical {Report}}\ (\bibinfo  {institution} {Massachusetts Institute of
  Technology},\ \bibinfo {address} {USA},\ \bibinfo {year} {1979})\BibitemShut
  {NoStop}%
\bibitem [{\citenamefont {Goldwasser}\ \emph {et~al.}(1988)\citenamefont
  {Goldwasser}, \citenamefont {Micali},\ and\ \citenamefont
  {Rivest}}]{goldwasser_digital_1988}%
  \BibitemOpen
  \bibfield  {author} {\bibinfo {author} {\bibfnamefont {S.}~\bibnamefont
  {Goldwasser}}, \bibinfo {author} {\bibfnamefont {S.}~\bibnamefont {Micali}},\
  and\ \bibinfo {author} {\bibfnamefont {R.~L.}\ \bibnamefont {Rivest}},\
  }\href {https://doi.org/10.1137/0217017} {\bibfield  {journal} {\bibinfo
  {journal} {SIAM Journal on Computing}\ }\textbf {\bibinfo {volume} {17}},\
  \bibinfo {pages} {281} (\bibinfo {year} {1988})}\BibitemShut {NoStop}%
\bibitem [{\citenamefont {Miller}(1986)}]{miller_use_1986}%
  \BibitemOpen
  \bibfield  {author} {\bibinfo {author} {\bibfnamefont {V.~S.}\ \bibnamefont
  {Miller}},\ }in\ \href {https://doi.org/10.1007/3-540-39799-X_31} {\emph
  {\bibinfo {booktitle} {Advances in {Cryptology} — {CRYPTO} ’85
  {Proceedings}}}},\ \bibinfo {series and number} {Lecture {Notes} in
  {Computer} {Science}},\ \bibinfo {editor} {edited by\ \bibinfo {editor}
  {\bibfnamefont {H.~C.}\ \bibnamefont {Williams}}}\ (\bibinfo  {publisher}
  {Springer},\ \bibinfo {address} {Berlin, Heidelberg},\ \bibinfo {year}
  {1986})\ pp.\ \bibinfo {pages} {417--426}\BibitemShut {NoStop}%
\bibitem [{\citenamefont {Koblitz}(1987)}]{koblitz_elliptic_1987}%
  \BibitemOpen
  \bibfield  {author} {\bibinfo {author} {\bibfnamefont {N.}~\bibnamefont
  {Koblitz}},\ }\href {https://doi.org/10.1090/S0025-5718-1987-0866109-5}
  {\bibfield  {journal} {\bibinfo  {journal} {Mathematics of Computation}\
  }\textbf {\bibinfo {volume} {48}},\ \bibinfo {pages} {203} (\bibinfo {year}
  {1987})}\BibitemShut {NoStop}%
\bibitem [{\citenamefont {Barker}(2016)}]{barker_recommendation_2016}%
  \BibitemOpen
  \bibfield  {author} {\bibinfo {author} {\bibfnamefont {E.}~\bibnamefont
  {Barker}},\ }\href {https://doi.org/10.6028/NIST.SP.800-57pt1r4} {\emph
  {\bibinfo {title} {Recommendation for {Key} {Management} {Part} 1:
  {General}}}},\ \bibinfo {type} {Tech. Rep.}\ \bibinfo {number} {NIST SP
  800-57pt1r4}\ (\bibinfo  {institution} {National Institute of Standards and
  Technology},\ \bibinfo {year} {2016})\BibitemShut {NoStop}%
\bibitem [{\citenamefont {Bennett}(1989)}]{bennett_timespace_1989}%
  \BibitemOpen
  \bibfield  {author} {\bibinfo {author} {\bibfnamefont {C.~H.}\ \bibnamefont
  {Bennett}},\ }\href {https://doi.org/10.1137/0218053} {\bibfield  {journal}
  {\bibinfo  {journal} {SIAM Journal on Computing}\ }\textbf {\bibinfo {volume}
  {18}},\ \bibinfo {pages} {766} (\bibinfo {year} {1989})}\BibitemShut
  {NoStop}%
\bibitem [{\citenamefont {Levine}\ and\ \citenamefont
  {Sherman}(1990)}]{levine_note_1990}%
  \BibitemOpen
  \bibfield  {author} {\bibinfo {author} {\bibfnamefont {R.~Y.}\ \bibnamefont
  {Levine}}\ and\ \bibinfo {author} {\bibfnamefont {A.~T.}\ \bibnamefont
  {Sherman}},\ }\href {https://doi.org/10.1137/0219046} {\bibfield  {journal}
  {\bibinfo  {journal} {SIAM Journal on Computing}\ }\textbf {\bibinfo {volume}
  {19}},\ \bibinfo {pages} {673} (\bibinfo {year} {1990})}\BibitemShut
  {NoStop}%
\bibitem [{\citenamefont {Aharonov}\ \emph {et~al.}(1998)\citenamefont
  {Aharonov}, \citenamefont {Kitaev},\ and\ \citenamefont
  {Nisan}}]{aharonov1998quantum}%
  \BibitemOpen
  \bibfield  {author} {\bibinfo {author} {\bibfnamefont {D.}~\bibnamefont
  {Aharonov}}, \bibinfo {author} {\bibfnamefont {A.}~\bibnamefont {Kitaev}},\
  and\ \bibinfo {author} {\bibfnamefont {N.}~\bibnamefont {Nisan}},\ }in\
  \href@noop {} {\emph {\bibinfo {booktitle} {Proceedings of the thirtieth
  annual ACM symposium on Theory of computing}}}\ (\bibinfo {year} {1998})\
  pp.\ \bibinfo {pages} {20--30}\BibitemShut {NoStop}%
\bibitem [{\citenamefont {Babu}\ \emph {et~al.}(2004)\citenamefont {Babu},
  \citenamefont {Islam}, \citenamefont {Chowdhury},\ and\ \citenamefont
  {Chowdhury}}]{babu2004synthesis}%
  \BibitemOpen
  \bibfield  {author} {\bibinfo {author} {\bibfnamefont {H.~M.~H.}\
  \bibnamefont {Babu}}, \bibinfo {author} {\bibfnamefont {M.~R.}\ \bibnamefont
  {Islam}}, \bibinfo {author} {\bibfnamefont {S.~M.~A.}\ \bibnamefont
  {Chowdhury}},\ and\ \bibinfo {author} {\bibfnamefont {A.~R.}\ \bibnamefont
  {Chowdhury}},\ }in\ \href@noop {} {\emph {\bibinfo {booktitle} {17th
  International Conference on VLSI Design. Proceedings.}}}\ (\bibinfo
  {organization} {IEEE},\ \bibinfo {year} {2004})\ pp.\ \bibinfo {pages}
  {757--760}\BibitemShut {NoStop}%
\bibitem [{\citenamefont {Kotiyal}\ \emph {et~al.}(2014)\citenamefont
  {Kotiyal}, \citenamefont {Thapliyal},\ and\ \citenamefont
  {Ranganathan}}]{kotiyal2014circuit}%
  \BibitemOpen
  \bibfield  {author} {\bibinfo {author} {\bibfnamefont {S.}~\bibnamefont
  {Kotiyal}}, \bibinfo {author} {\bibfnamefont {H.}~\bibnamefont {Thapliyal}},\
  and\ \bibinfo {author} {\bibfnamefont {N.}~\bibnamefont {Ranganathan}},\ }in\
  \href@noop {} {\emph {\bibinfo {booktitle} {2014 27th international
  conference on VLSI design and 2014 13th international conference on embedded
  systems}}}\ (\bibinfo {organization} {IEEE},\ \bibinfo {year} {2014})\ pp.\
  \bibinfo {pages} {545--550}\BibitemShut {NoStop}%
\bibitem [{\citenamefont {Saffman}(2016)}]{saffman_quantum_2016}%
  \BibitemOpen
  \bibfield  {author} {\bibinfo {author} {\bibfnamefont {M.}~\bibnamefont
  {Saffman}},\ }\href@noop {} {\bibfield  {journal} {\bibinfo  {journal}
  {Journal of Physics B: Atomic, Molecular and Optical Physics}\ }\textbf
  {\bibinfo {volume} {49}},\ \bibinfo {pages} {202001} (\bibinfo {year}
  {2016})}\BibitemShut {NoStop}%
\bibitem [{\citenamefont {Levine}\ \emph {et~al.}(2019)\citenamefont {Levine},
  \citenamefont {Keesling}, \citenamefont {Semeghini}, \citenamefont {Omran},
  \citenamefont {Wang}, \citenamefont {Ebadi}, \citenamefont {Bernien},
  \citenamefont {Greiner}, \citenamefont {Vuletić}, \citenamefont {Pichler},\
  and\ \citenamefont {Lukin}}]{levine_parallel_2019}%
  \BibitemOpen
  \bibfield  {author} {\bibinfo {author} {\bibfnamefont {H.}~\bibnamefont
  {Levine}}, \bibinfo {author} {\bibfnamefont {A.}~\bibnamefont {Keesling}},
  \bibinfo {author} {\bibfnamefont {G.}~\bibnamefont {Semeghini}}, \bibinfo
  {author} {\bibfnamefont {A.}~\bibnamefont {Omran}}, \bibinfo {author}
  {\bibfnamefont {T.~T.}\ \bibnamefont {Wang}}, \bibinfo {author}
  {\bibfnamefont {S.}~\bibnamefont {Ebadi}}, \bibinfo {author} {\bibfnamefont
  {H.}~\bibnamefont {Bernien}}, \bibinfo {author} {\bibfnamefont
  {M.}~\bibnamefont {Greiner}}, \bibinfo {author} {\bibfnamefont
  {V.}~\bibnamefont {Vuletić}}, \bibinfo {author} {\bibfnamefont
  {H.}~\bibnamefont {Pichler}},\ and\ \bibinfo {author} {\bibfnamefont {M.~D.}\
  \bibnamefont {Lukin}},\ }\href
  {https://doi.org/10.1103/PhysRevLett.123.170503} {\bibfield  {journal}
  {\bibinfo  {journal} {Phys. Rev. Lett.}\ }\textbf {\bibinfo {volume} {123}},\
  \bibinfo {pages} {170503} (\bibinfo {year} {2019})}\BibitemShut {NoStop}%
\bibitem [{\citenamefont {Graham}\ \emph {et~al.}(2019)\citenamefont {Graham},
  \citenamefont {Kwon}, \citenamefont {Grinkemeyer}, \citenamefont {Marra},
  \citenamefont {Jiang}, \citenamefont {Lichtman}, \citenamefont {Sun},
  \citenamefont {Ebert},\ and\ \citenamefont
  {Saffman}}]{graham_rydberg-mediated_2019}%
  \BibitemOpen
  \bibfield  {author} {\bibinfo {author} {\bibfnamefont {T.}~\bibnamefont
  {Graham}}, \bibinfo {author} {\bibfnamefont {M.}~\bibnamefont {Kwon}},
  \bibinfo {author} {\bibfnamefont {B.}~\bibnamefont {Grinkemeyer}}, \bibinfo
  {author} {\bibfnamefont {Z.}~\bibnamefont {Marra}}, \bibinfo {author}
  {\bibfnamefont {X.}~\bibnamefont {Jiang}}, \bibinfo {author} {\bibfnamefont
  {M.}~\bibnamefont {Lichtman}}, \bibinfo {author} {\bibfnamefont
  {Y.}~\bibnamefont {Sun}}, \bibinfo {author} {\bibfnamefont {M.}~\bibnamefont
  {Ebert}},\ and\ \bibinfo {author} {\bibfnamefont {M.}~\bibnamefont
  {Saffman}},\ }\href@noop {} {\bibfield  {journal} {\bibinfo  {journal}
  {Physical Review Letters}\ }\textbf {\bibinfo {volume} {123}},\ \bibinfo
  {pages} {230501} (\bibinfo {year} {2019})}\BibitemShut {NoStop}%
\bibitem [{\citenamefont {Madjarov}\ \emph {et~al.}(2020)\citenamefont
  {Madjarov}, \citenamefont {Covey}, \citenamefont {Shaw}, \citenamefont
  {Choi}, \citenamefont {Kale}, \citenamefont {Cooper}, \citenamefont
  {Pichler}, \citenamefont {Schkolnik}, \citenamefont {Williams},\ and\
  \citenamefont {Endres}}]{madjarov_high-fidelity_2020}%
  \BibitemOpen
  \bibfield  {author} {\bibinfo {author} {\bibfnamefont {I.~S.}\ \bibnamefont
  {Madjarov}}, \bibinfo {author} {\bibfnamefont {J.~P.}\ \bibnamefont {Covey}},
  \bibinfo {author} {\bibfnamefont {A.~L.}\ \bibnamefont {Shaw}}, \bibinfo
  {author} {\bibfnamefont {J.}~\bibnamefont {Choi}}, \bibinfo {author}
  {\bibfnamefont {A.}~\bibnamefont {Kale}}, \bibinfo {author} {\bibfnamefont
  {A.}~\bibnamefont {Cooper}}, \bibinfo {author} {\bibfnamefont
  {H.}~\bibnamefont {Pichler}}, \bibinfo {author} {\bibfnamefont
  {V.}~\bibnamefont {Schkolnik}}, \bibinfo {author} {\bibfnamefont {J.~R.}\
  \bibnamefont {Williams}},\ and\ \bibinfo {author} {\bibfnamefont
  {M.}~\bibnamefont {Endres}},\ }\href
  {https://doi.org/10.1038/s41567-020-0903-z} {\bibfield  {journal} {\bibinfo
  {journal} {Nature Physics}\ }\textbf {\bibinfo {volume} {16}},\ \bibinfo
  {pages} {857} (\bibinfo {year} {2020})}\BibitemShut {NoStop}%
\bibitem [{\citenamefont {Browaeys}\ and\ \citenamefont
  {Lahaye}(2020)}]{browaeys2020many}%
  \BibitemOpen
  \bibfield  {author} {\bibinfo {author} {\bibfnamefont {A.}~\bibnamefont
  {Browaeys}}\ and\ \bibinfo {author} {\bibfnamefont {T.}~\bibnamefont
  {Lahaye}},\ }\href@noop {} {\bibfield  {journal} {\bibinfo  {journal} {Nature
  Physics}\ }\textbf {\bibinfo {volume} {16}},\ \bibinfo {pages} {132}
  (\bibinfo {year} {2020})}\BibitemShut {NoStop}%
\bibitem [{\citenamefont
  {Mahadev}(2018{\natexlab{a}})}]{mahadev_classical_homomorphic_2018}%
  \BibitemOpen
  \bibfield  {author} {\bibinfo {author} {\bibfnamefont {U.}~\bibnamefont
  {Mahadev}},\ }\href {http://arxiv.org/abs/1708.02130} {\bibfield  {journal}
  {\bibinfo  {journal} {arXiv:1708.02130 [quant-ph]}\ } (\bibinfo {year}
  {2018}{\natexlab{a}})},\ \bibinfo {note} {arXiv: 1708.02130}\BibitemShut
  {NoStop}%
\bibitem [{\citenamefont
  {Mahadev}(2018{\natexlab{b}})}]{mahadev_classical_2018}%
  \BibitemOpen
  \bibfield  {author} {\bibinfo {author} {\bibfnamefont {U.}~\bibnamefont
  {Mahadev}},\ }\href {http://arxiv.org/abs/1804.01082} {\bibfield  {journal}
  {\bibinfo  {journal} {arXiv:1804.01082 [quant-ph]}\ } (\bibinfo {year}
  {2018}{\natexlab{b}})}\BibitemShut {NoStop}%
\bibitem [{\citenamefont {Gheorghiu}\ and\ \citenamefont
  {Vidick}(2019)}]{gheorghiu_computationally-secure_2019}%
  \BibitemOpen
  \bibfield  {author} {\bibinfo {author} {\bibfnamefont {A.}~\bibnamefont
  {Gheorghiu}}\ and\ \bibinfo {author} {\bibfnamefont {T.}~\bibnamefont
  {Vidick}},\ }\href {http://arxiv.org/abs/1904.06320} {\bibfield  {journal}
  {\bibinfo  {journal} {arXiv:1904.06320 [quant-ph]}\ } (\bibinfo {year}
  {2019})},\ \bibinfo {note} {arXiv: 1904.06320}\BibitemShut {NoStop}%
\bibitem [{Note4()}]{Note4}%
  \BibitemOpen
  \bibinfo {note} {Replacing the random oracle with a hash function is termed a
  \protect \emph {heuristic} rather than an \protect \emph {assumption} because
  the security of this procedure generally holds in practice but is not
  provable---in fact, there exist constructions that are provably secure in the
  random oracle model but trivially insecure when instantiated with a hash
  function \cite {canetti_random_1998}.}\BibitemShut {Stop}%
\bibitem [{\citenamefont {Canetti}\ \emph {et~al.}(1998)\citenamefont
  {Canetti}, \citenamefont {Goldreich},\ and\ \citenamefont
  {Halevi}}]{canetti_random_1998}%
  \BibitemOpen
  \bibfield  {author} {\bibinfo {author} {\bibfnamefont {R.}~\bibnamefont
  {Canetti}}, \bibinfo {author} {\bibfnamefont {O.}~\bibnamefont {Goldreich}},\
  and\ \bibinfo {author} {\bibfnamefont {S.}~\bibnamefont {Halevi}},\ }\href
  {https://eprint.iacr.org/1998/011} {\emph {\bibinfo {title} {The {Random}
  {Oracle} {Methodology}, {Revisited}}}},\ \bibinfo {type} {Tech. Rep.}\
  \bibinfo {number} {011}\ (\bibinfo {year} {1998})\BibitemShut {NoStop}%
\bibitem [{\citenamefont {Koblitz}\ and\ \citenamefont
  {Menezes}(2015)}]{koblitz_random_2015}%
  \BibitemOpen
  \bibfield  {author} {\bibinfo {author} {\bibfnamefont {N.}~\bibnamefont
  {Koblitz}}\ and\ \bibinfo {author} {\bibfnamefont {A.~J.}\ \bibnamefont
  {Menezes}},\ }\href {https://doi.org/10.1007/s10623-015-0094-2} {\bibfield
  {journal} {\bibinfo  {journal} {Designs, Codes and Cryptography}\ }\textbf
  {\bibinfo {volume} {77}},\ \bibinfo {pages} {587} (\bibinfo {year}
  {2015})}\BibitemShut {NoStop}%
\bibitem [{\citenamefont {Aaronson}\ and\ \citenamefont
  {Chen}(2016)}]{aaronson_complexity-theoretic_2016}%
  \BibitemOpen
  \bibfield  {author} {\bibinfo {author} {\bibfnamefont {S.}~\bibnamefont
  {Aaronson}}\ and\ \bibinfo {author} {\bibfnamefont {L.}~\bibnamefont
  {Chen}},\ }\href {http://arxiv.org/abs/1612.05903} {\bibfield  {journal}
  {\bibinfo  {journal} {arXiv:1612.05903 [quant-ph]}\ } (\bibinfo {year}
  {2016})}\BibitemShut {NoStop}%
\bibitem [{\citenamefont {Liu}\ and\ \citenamefont
  {Gheorghiu}(2021)}]{liu_depth-efficient_2021}%
  \BibitemOpen
  \bibfield  {author} {\bibinfo {author} {\bibfnamefont {Z.}~\bibnamefont
  {Liu}}\ and\ \bibinfo {author} {\bibfnamefont {A.}~\bibnamefont
  {Gheorghiu}},\ }\href {http://arxiv.org/abs/2107.02163} {\bibfield  {journal}
  {\bibinfo  {journal} {arXiv:2107.02163 [quant-ph]}\ } (\bibinfo {year}
  {2021})},\ \bibinfo {note} {arXiv: 2107.02163}\BibitemShut {NoStop}%
\bibitem [{\citenamefont {Hirahara}\ and\ \citenamefont
  {Gall}(2021)}]{hirahara_test_2021}%
  \BibitemOpen
  \bibfield  {author} {\bibinfo {author} {\bibfnamefont {S.}~\bibnamefont
  {Hirahara}}\ and\ \bibinfo {author} {\bibfnamefont {F.~L.}\ \bibnamefont
  {Gall}},\ }\href {http://arxiv.org/abs/2105.05500} {\bibfield  {journal}
  {\bibinfo  {journal} {arXiv:2105.05500 [quant-ph]}\ } (\bibinfo {year}
  {2021})},\ \bibinfo {note} {arXiv: 2105.05500}\BibitemShut {NoStop}%
\bibitem [{\citenamefont {Goldreich}\ and\ \citenamefont
  {Levint}(1989)}]{goldreich_hard-core_1989}%
  \BibitemOpen
  \bibfield  {author} {\bibinfo {author} {\bibfnamefont {O.}~\bibnamefont
  {Goldreich}}\ and\ \bibinfo {author} {\bibfnamefont {L.~A.}\ \bibnamefont
  {Levint}},\ }in\ \href@noop {} {\emph {\bibinfo {booktitle} {In {Proceedings}
  of the {Twenty} {First} {Annual} {ACM} {Symposium} on {Theory} of
  {Computing}}}}\ (\bibinfo {year} {1989})\ pp.\ \bibinfo {pages}
  {25--32}\BibitemShut {NoStop}%
\bibitem [{Note5()}]{Note5}%
  \BibitemOpen
  \bibinfo {note} {The oracle's noise rate is \protect \emph {not} simply $p_{r
  \cdot x_1}$: that is the probability that any single value $r \cdot x_1$ is
  correct, but all of the queries to the oracle are correlated (they are for
  the same iteration of the protocol, and thus the same value of
  $y$).}\BibitemShut {Stop}%
\bibitem [{Note6()}]{Note6}%
  \BibitemOpen
  \bibinfo {note} {This number comes from solving the classical bound
  (Equation~\ref {eq:classical-bound}) for circuit fidelity $\protect \mathcal
  {F}$, with $p_x = \protect \mathcal {F}$ and $p_\protect \mathrm {CHSH} =
  \protect \nicefrac {1}{2} + \protect \mathcal {F}/2$.}\BibitemShut {Stop}%
\bibitem [{Note7()}]{Note7}%
  \BibitemOpen
  \bibinfo {note} {This is true even if the coherence is exponentially small in
  $n$. Of course, with arbitrarily low coherence the runtime may become
  excessively large such that quantum advantage cannot be demonstrated---the
  point is that regardless of runtime, the classical probability bound can be
  exceeded with a device that has arbitrarily low circuit
  fidelity.}\BibitemShut {Stop}%
\bibitem [{Note8()}]{Note8}%
  \BibitemOpen
  \bibinfo {note} {This scheme will only remove errors in the first round of
  the protocol, but fortunately, one expects the overwhelming majority of the
  quantum computation, and thus also the majority of errors, to occur in that
  round.}\BibitemShut {Stop}%
\bibitem [{Note9()}]{Note9}%
  \BibitemOpen
  \bibinfo {note} {This procedure does not leak data to a classical cheater,
  because the verifier does not communicate which runs were discarded.
  Furthermore, it does not affect the soundness of Theorem~\ref
  {thm:classical-success}, because the machine $\protect \mathcal {B}$ in that
  theorem's proof can simply iterate until it encounters a valid
  $y$.}\BibitemShut {Stop}%
\bibitem [{Note10()}]{Note10}%
  \BibitemOpen
  \bibinfo {note} {This is true because $g_i(x)$ is the result of adding extra
  output bits to the gates of a classical circuit, which is efficient to
  evaluate on any input.}\BibitemShut {Stop}%
\bibitem [{\citenamefont {Zhu}\ \emph {et~al.}(2021)\citenamefont {Zhu},
  \citenamefont {Noel}, \citenamefont {Risinger}, \citenamefont {Egan},
  \citenamefont {Biswas}, \citenamefont {Wang}, \citenamefont {Nam},
  \citenamefont {Meyer}, \citenamefont {Vazirani}, \citenamefont {Yao} \emph
  {et~al.}}]{zhu2021demonstration}%
  \BibitemOpen
  \bibfield  {author} {\bibinfo {author} {\bibfnamefont {D.}~\bibnamefont
  {Zhu}}, \bibinfo {author} {\bibfnamefont {C.}~\bibnamefont {Noel}}, \bibinfo
  {author} {\bibfnamefont {A.}~\bibnamefont {Risinger}}, \bibinfo {author}
  {\bibfnamefont {L.}~\bibnamefont {Egan}}, \bibinfo {author} {\bibfnamefont
  {D.}~\bibnamefont {Biswas}}, \bibinfo {author} {\bibfnamefont
  {Q.}~\bibnamefont {Wang}}, \bibinfo {author} {\bibfnamefont {Y.}~\bibnamefont
  {Nam}}, \bibinfo {author} {\bibfnamefont {G.}~\bibnamefont {Meyer}}, \bibinfo
  {author} {\bibfnamefont {U.}~\bibnamefont {Vazirani}}, \bibinfo {author}
  {\bibfnamefont {N.}~\bibnamefont {Yao}}, \emph {et~al.},\ }\href@noop {}
  {\bibfield  {journal} {\bibinfo  {journal} {Bulletin of the American Physical
  Society}\ } (\bibinfo {year} {2021})}\BibitemShut {NoStop}%
\bibitem [{\citenamefont {Ryan-Anderson}\ \emph {et~al.}(2021)\citenamefont
  {Ryan-Anderson}, \citenamefont {Bohnet}, \citenamefont {Lee}, \citenamefont
  {Gresh}, \citenamefont {Hankin}, \citenamefont {Gaebler}, \citenamefont
  {Francois}, \citenamefont {Chernoguzov}, \citenamefont {Lucchetti},
  \citenamefont {Brown}, \citenamefont {Gatterman}, \citenamefont {Halit},
  \citenamefont {Gilmore}, \citenamefont {Gerber}, \citenamefont {Neyenhuis},
  \citenamefont {Hayes},\ and\ \citenamefont
  {Stutz}}]{ryan-anderson_realization_2021}%
  \BibitemOpen
  \bibfield  {author} {\bibinfo {author} {\bibfnamefont {C.}~\bibnamefont
  {Ryan-Anderson}}, \bibinfo {author} {\bibfnamefont {J.~G.}\ \bibnamefont
  {Bohnet}}, \bibinfo {author} {\bibfnamefont {K.}~\bibnamefont {Lee}},
  \bibinfo {author} {\bibfnamefont {D.}~\bibnamefont {Gresh}}, \bibinfo
  {author} {\bibfnamefont {A.}~\bibnamefont {Hankin}}, \bibinfo {author}
  {\bibfnamefont {J.~P.}\ \bibnamefont {Gaebler}}, \bibinfo {author}
  {\bibfnamefont {D.}~\bibnamefont {Francois}}, \bibinfo {author}
  {\bibfnamefont {A.}~\bibnamefont {Chernoguzov}}, \bibinfo {author}
  {\bibfnamefont {D.}~\bibnamefont {Lucchetti}}, \bibinfo {author}
  {\bibfnamefont {N.~C.}\ \bibnamefont {Brown}}, \bibinfo {author}
  {\bibfnamefont {T.~M.}\ \bibnamefont {Gatterman}}, \bibinfo {author}
  {\bibfnamefont {S.~K.}\ \bibnamefont {Halit}}, \bibinfo {author}
  {\bibfnamefont {K.}~\bibnamefont {Gilmore}}, \bibinfo {author} {\bibfnamefont
  {J.}~\bibnamefont {Gerber}}, \bibinfo {author} {\bibfnamefont
  {B.}~\bibnamefont {Neyenhuis}}, \bibinfo {author} {\bibfnamefont
  {D.}~\bibnamefont {Hayes}},\ and\ \bibinfo {author} {\bibfnamefont {R.~P.}\
  \bibnamefont {Stutz}},\ }\href {http://arxiv.org/abs/2107.07505} {\bibfield
  {journal} {\bibinfo  {journal} {arXiv:2107.07505 [quant-ph]}\ } (\bibinfo
  {year} {2021})},\ \bibinfo {note} {arXiv: 2107.07505}\BibitemShut {NoStop}%
\bibitem [{\citenamefont {Häner}\ \emph {et~al.}(2020)\citenamefont {Häner},
  \citenamefont {Jaques}, \citenamefont {Naehrig}, \citenamefont {Roetteler},\
  and\ \citenamefont {Soeken}}]{haner_improved_2020}%
  \BibitemOpen
  \bibfield  {author} {\bibinfo {author} {\bibfnamefont {T.}~\bibnamefont
  {Häner}}, \bibinfo {author} {\bibfnamefont {S.}~\bibnamefont {Jaques}},
  \bibinfo {author} {\bibfnamefont {M.}~\bibnamefont {Naehrig}}, \bibinfo
  {author} {\bibfnamefont {M.}~\bibnamefont {Roetteler}},\ and\ \bibinfo
  {author} {\bibfnamefont {M.}~\bibnamefont {Soeken}},\ }\href
  {http://arxiv.org/abs/2001.09580} {\bibfield  {journal} {\bibinfo  {journal}
  {arXiv:2001.09580 [quant-ph]}\ } (\bibinfo {year} {2020})}\BibitemShut
  {NoStop}%
\bibitem [{Note11()}]{Note11}%
  \BibitemOpen
  \bibinfo {note} {Code is available at \protect \url
  {https://github.com/GregDMeyer/quantum-advantage} and is archived on
  Zenodo~\cite {cirq_code}}\BibitemShut {NoStop}%
\bibitem [{\citenamefont {Draper}(2000)}]{draper_addition_2000}%
  \BibitemOpen
  \bibfield  {author} {\bibinfo {author} {\bibfnamefont {T.~G.}\ \bibnamefont
  {Draper}},\ }\href {http://arxiv.org/abs/quant-ph/0008033} {\bibfield
  {journal} {\bibinfo  {journal} {arXiv:quant-ph/0008033}\ } (\bibinfo {year}
  {2000})}\BibitemShut {NoStop}%
\bibitem [{\citenamefont {Beauregard}(2003)}]{beauregard_circuit_2003}%
  \BibitemOpen
  \bibfield  {author} {\bibinfo {author} {\bibfnamefont {S.}~\bibnamefont
  {Beauregard}},\ }\href {http://arxiv.org/abs/quant-ph/0205095} {\bibfield
  {journal} {\bibinfo  {journal} {arXiv:quant-ph/0205095}\ } (\bibinfo {year}
  {2003})}\BibitemShut {NoStop}%
\bibitem [{Note12()}]{Note12}%
  \BibitemOpen
  \bibinfo {note} {We must take $m > n + \protect \mathcal {O}(1)$ to
  sufficiently resolve the value $x^2 \nonscript \mskip -\medmuskip \mkern
  5mu\mathbin {\mathgroup \symoperators mod}\penalty 900 \mkern 5mu\nonscript
  \mskip -\medmuskip N$ in post-processing}\BibitemShut {NoStop}%
\bibitem [{\citenamefont {Zhang}\ \emph {et~al.}(2017)\citenamefont {Zhang},
  \citenamefont {Pagano}, \citenamefont {Hess}, \citenamefont {Kyprianidis},
  \citenamefont {Becker}, \citenamefont {Kaplan}, \citenamefont {Gorshkov},
  \citenamefont {Gong},\ and\ \citenamefont {Monroe}}]{zhang_observation_2017}%
  \BibitemOpen
  \bibfield  {author} {\bibinfo {author} {\bibfnamefont {J.}~\bibnamefont
  {Zhang}}, \bibinfo {author} {\bibfnamefont {G.}~\bibnamefont {Pagano}},
  \bibinfo {author} {\bibfnamefont {P.~W.}\ \bibnamefont {Hess}}, \bibinfo
  {author} {\bibfnamefont {A.}~\bibnamefont {Kyprianidis}}, \bibinfo {author}
  {\bibfnamefont {P.}~\bibnamefont {Becker}}, \bibinfo {author} {\bibfnamefont
  {H.}~\bibnamefont {Kaplan}}, \bibinfo {author} {\bibfnamefont {A.~V.}\
  \bibnamefont {Gorshkov}}, \bibinfo {author} {\bibfnamefont {Z.-X.}\
  \bibnamefont {Gong}},\ and\ \bibinfo {author} {\bibfnamefont
  {C.}~\bibnamefont {Monroe}},\ }\href@noop {} {\bibfield  {journal} {\bibinfo
  {journal} {Nature}\ }\textbf {\bibinfo {volume} {551}},\ \bibinfo {pages}
  {601} (\bibinfo {year} {2017})}\BibitemShut {NoStop}%
\bibitem [{\citenamefont {Scholl}\ \emph {et~al.}(2020)\citenamefont {Scholl},
  \citenamefont {Schuler}, \citenamefont {Williams}, \citenamefont
  {Eberharter}, \citenamefont {Barredo}, \citenamefont {Schymik}, \citenamefont
  {Lienhard}, \citenamefont {Henry}, \citenamefont {Lang}, \citenamefont
  {Lahaye}, \citenamefont {Läuchli},\ and\ \citenamefont
  {Browaeys}}]{scholl2020programmable}%
  \BibitemOpen
  \bibfield  {author} {\bibinfo {author} {\bibfnamefont {P.}~\bibnamefont
  {Scholl}}, \bibinfo {author} {\bibfnamefont {M.}~\bibnamefont {Schuler}},
  \bibinfo {author} {\bibfnamefont {H.~J.}\ \bibnamefont {Williams}}, \bibinfo
  {author} {\bibfnamefont {A.~A.}\ \bibnamefont {Eberharter}}, \bibinfo
  {author} {\bibfnamefont {D.}~\bibnamefont {Barredo}}, \bibinfo {author}
  {\bibfnamefont {K.-N.}\ \bibnamefont {Schymik}}, \bibinfo {author}
  {\bibfnamefont {V.}~\bibnamefont {Lienhard}}, \bibinfo {author}
  {\bibfnamefont {L.-P.}\ \bibnamefont {Henry}}, \bibinfo {author}
  {\bibfnamefont {T.~C.}\ \bibnamefont {Lang}}, \bibinfo {author}
  {\bibfnamefont {T.}~\bibnamefont {Lahaye}}, \bibinfo {author} {\bibfnamefont
  {A.~M.}\ \bibnamefont {Läuchli}},\ and\ \bibinfo {author} {\bibfnamefont
  {A.}~\bibnamefont {Browaeys}},\ }\href {http://arxiv.org/abs/2012.12268}
  {\bibfield  {journal} {\bibinfo  {journal} {arXiv:2012.12268 [cond-mat,
  physics:physics, physics:quant-ph]}\ } (\bibinfo {year} {2020})}\BibitemShut
  {NoStop}%
\bibitem [{\citenamefont {Ebadi}\ \emph {et~al.}(2020)\citenamefont {Ebadi},
  \citenamefont {Wang}, \citenamefont {Levine}, \citenamefont {Keesling},
  \citenamefont {Semeghini}, \citenamefont {Omran}, \citenamefont {Bluvstein},
  \citenamefont {Samajdar}, \citenamefont {Pichler}, \citenamefont {Ho},
  \citenamefont {Choi}, \citenamefont {Sachdev}, \citenamefont {Greiner},
  \citenamefont {Vuletic},\ and\ \citenamefont {Lukin}}]{ebadi2020quantum}%
  \BibitemOpen
  \bibfield  {author} {\bibinfo {author} {\bibfnamefont {S.}~\bibnamefont
  {Ebadi}}, \bibinfo {author} {\bibfnamefont {T.~T.}\ \bibnamefont {Wang}},
  \bibinfo {author} {\bibfnamefont {H.}~\bibnamefont {Levine}}, \bibinfo
  {author} {\bibfnamefont {A.}~\bibnamefont {Keesling}}, \bibinfo {author}
  {\bibfnamefont {G.}~\bibnamefont {Semeghini}}, \bibinfo {author}
  {\bibfnamefont {A.}~\bibnamefont {Omran}}, \bibinfo {author} {\bibfnamefont
  {D.}~\bibnamefont {Bluvstein}}, \bibinfo {author} {\bibfnamefont
  {R.}~\bibnamefont {Samajdar}}, \bibinfo {author} {\bibfnamefont
  {H.}~\bibnamefont {Pichler}}, \bibinfo {author} {\bibfnamefont {W.~W.}\
  \bibnamefont {Ho}}, \bibinfo {author} {\bibfnamefont {S.}~\bibnamefont
  {Choi}}, \bibinfo {author} {\bibfnamefont {S.}~\bibnamefont {Sachdev}},
  \bibinfo {author} {\bibfnamefont {M.}~\bibnamefont {Greiner}}, \bibinfo
  {author} {\bibfnamefont {V.}~\bibnamefont {Vuletic}},\ and\ \bibinfo {author}
  {\bibfnamefont {M.~D.}\ \bibnamefont {Lukin}},\ }\href
  {http://arxiv.org/abs/2012.12281} {\bibfield  {journal} {\bibinfo  {journal}
  {arXiv:2012.12281 [cond-mat, physics:physics, physics:quant-ph]}\ } (\bibinfo
  {year} {2020})}\BibitemShut {NoStop}%
\bibitem [{\citenamefont {Wang}\ \emph {et~al.}(2015)\citenamefont {Wang},
  \citenamefont {Zhang}, \citenamefont {Corcovilos}, \citenamefont {Kumar},\
  and\ \citenamefont {Weiss}}]{PhysRevLett.115.043003}%
  \BibitemOpen
  \bibfield  {author} {\bibinfo {author} {\bibfnamefont {Y.}~\bibnamefont
  {Wang}}, \bibinfo {author} {\bibfnamefont {X.}~\bibnamefont {Zhang}},
  \bibinfo {author} {\bibfnamefont {T.~A.}\ \bibnamefont {Corcovilos}},
  \bibinfo {author} {\bibfnamefont {A.}~\bibnamefont {Kumar}},\ and\ \bibinfo
  {author} {\bibfnamefont {D.~S.}\ \bibnamefont {Weiss}},\ }\href
  {https://doi.org/10.1103/PhysRevLett.115.043003} {\bibfield  {journal}
  {\bibinfo  {journal} {Phys. Rev. Lett.}\ }\textbf {\bibinfo {volume} {115}},\
  \bibinfo {pages} {043003} (\bibinfo {year} {2015})}\BibitemShut {NoStop}%
\bibitem [{\citenamefont {Wang}\ \emph {et~al.}(2016)\citenamefont {Wang},
  \citenamefont {Kumar}, \citenamefont {Wu},\ and\ \citenamefont
  {Weiss}}]{wang2016single}%
  \BibitemOpen
  \bibfield  {author} {\bibinfo {author} {\bibfnamefont {Y.}~\bibnamefont
  {Wang}}, \bibinfo {author} {\bibfnamefont {A.}~\bibnamefont {Kumar}},
  \bibinfo {author} {\bibfnamefont {T.-Y.}\ \bibnamefont {Wu}},\ and\ \bibinfo
  {author} {\bibfnamefont {D.~S.}\ \bibnamefont {Weiss}},\ }\href@noop {}
  {\bibfield  {journal} {\bibinfo  {journal} {Science}\ }\textbf {\bibinfo
  {volume} {352}},\ \bibinfo {pages} {1562} (\bibinfo {year}
  {2016})}\BibitemShut {NoStop}%
\bibitem [{\citenamefont {Kumar}\ \emph {et~al.}(2018)\citenamefont {Kumar},
  \citenamefont {Wu}, \citenamefont {Giraldo},\ and\ \citenamefont
  {Weiss}}]{kumar2018sorting}%
  \BibitemOpen
  \bibfield  {author} {\bibinfo {author} {\bibfnamefont {A.}~\bibnamefont
  {Kumar}}, \bibinfo {author} {\bibfnamefont {T.-Y.}\ \bibnamefont {Wu}},
  \bibinfo {author} {\bibfnamefont {F.}~\bibnamefont {Giraldo}},\ and\ \bibinfo
  {author} {\bibfnamefont {D.~S.}\ \bibnamefont {Weiss}},\ }\href@noop {}
  {\bibfield  {journal} {\bibinfo  {journal} {Nature}\ }\textbf {\bibinfo
  {volume} {561}},\ \bibinfo {pages} {83} (\bibinfo {year} {2018})}\BibitemShut
  {NoStop}%
\bibitem [{\citenamefont {L{\"o}w}\ \emph {et~al.}(2012)\citenamefont
  {L{\"o}w}, \citenamefont {Weimer}, \citenamefont {Nipper}, \citenamefont
  {Balewski}, \citenamefont {Butscher}, \citenamefont {B{\"u}chler},\ and\
  \citenamefont {Pfau}}]{low2012experimental}%
  \BibitemOpen
  \bibfield  {author} {\bibinfo {author} {\bibfnamefont {R.}~\bibnamefont
  {L{\"o}w}}, \bibinfo {author} {\bibfnamefont {H.}~\bibnamefont {Weimer}},
  \bibinfo {author} {\bibfnamefont {J.}~\bibnamefont {Nipper}}, \bibinfo
  {author} {\bibfnamefont {J.~B.}\ \bibnamefont {Balewski}}, \bibinfo {author}
  {\bibfnamefont {B.}~\bibnamefont {Butscher}}, \bibinfo {author}
  {\bibfnamefont {H.~P.}\ \bibnamefont {B{\"u}chler}},\ and\ \bibinfo {author}
  {\bibfnamefont {T.}~\bibnamefont {Pfau}},\ }\href@noop {} {\bibfield
  {journal} {\bibinfo  {journal} {Journal of Physics B: Atomic, Molecular and
  Optical Physics}\ }\textbf {\bibinfo {volume} {45}},\ \bibinfo {pages}
  {113001} (\bibinfo {year} {2012})}\BibitemShut {NoStop}%
\bibitem [{Note13()}]{Note13}%
  \BibitemOpen
  \bibinfo {note} {We note that this spacing is ultimately limited by a
  combination of the optical diffraction limit and the orbital size of
  $n\approx 70$ Rydberg states.}\BibitemShut {Stop}%
\bibitem [{\citenamefont {de~Léséleuc}\ \emph {et~al.}(2018)\citenamefont
  {de~Léséleuc}, \citenamefont {Barredo}, \citenamefont {Lienhard},
  \citenamefont {Browaeys},\ and\ \citenamefont
  {Lahaye}}]{de_leseleuc_analysis_2018}%
  \BibitemOpen
  \bibfield  {author} {\bibinfo {author} {\bibfnamefont {S.}~\bibnamefont
  {de~Léséleuc}}, \bibinfo {author} {\bibfnamefont {D.}~\bibnamefont
  {Barredo}}, \bibinfo {author} {\bibfnamefont {V.}~\bibnamefont {Lienhard}},
  \bibinfo {author} {\bibfnamefont {A.}~\bibnamefont {Browaeys}},\ and\
  \bibinfo {author} {\bibfnamefont {T.}~\bibnamefont {Lahaye}},\ }\href
  {https://doi.org/10.1103/PhysRevA.97.053803} {\bibfield  {journal} {\bibinfo
  {journal} {Physical Review A}\ }\textbf {\bibinfo {volume} {97}},\ \bibinfo
  {pages} {053803} (\bibinfo {year} {2018})}\BibitemShut {NoStop}%
\bibitem [{\citenamefont {Liu}\ \emph {et~al.}(2020)\citenamefont {Liu},
  \citenamefont {Sun}, \citenamefont {Fu}, \citenamefont {Xu}, \citenamefont
  {Wang}, \citenamefont {He}, \citenamefont {Wang},\ and\ \citenamefont
  {Zhan}}]{liu_influence_2020}%
  \BibitemOpen
  \bibfield  {author} {\bibinfo {author} {\bibfnamefont {Y.}~\bibnamefont
  {Liu}}, \bibinfo {author} {\bibfnamefont {Y.}~\bibnamefont {Sun}}, \bibinfo
  {author} {\bibfnamefont {Z.}~\bibnamefont {Fu}}, \bibinfo {author}
  {\bibfnamefont {P.}~\bibnamefont {Xu}}, \bibinfo {author} {\bibfnamefont
  {X.}~\bibnamefont {Wang}}, \bibinfo {author} {\bibfnamefont {X.}~\bibnamefont
  {He}}, \bibinfo {author} {\bibfnamefont {J.}~\bibnamefont {Wang}},\ and\
  \bibinfo {author} {\bibfnamefont {M.}~\bibnamefont {Zhan}},\ }\href
  {http://arxiv.org/abs/2012.12589} {\bibfield  {journal} {\bibinfo  {journal}
  {arXiv:2012.12589 [quant-ph]}\ } (\bibinfo {year} {2020})}\BibitemShut
  {NoStop}%
\bibitem [{\citenamefont {Schäfer}\ \emph {et~al.}(2018)\citenamefont
  {Schäfer}, \citenamefont {Ballance}, \citenamefont {Thirumalai},
  \citenamefont {Stephenson}, \citenamefont {Ballance}, \citenamefont
  {Steane},\ and\ \citenamefont {Lucas}}]{schafer_fast_2018}%
  \BibitemOpen
  \bibfield  {author} {\bibinfo {author} {\bibfnamefont {V.~M.}\ \bibnamefont
  {Schäfer}}, \bibinfo {author} {\bibfnamefont {C.~J.}\ \bibnamefont
  {Ballance}}, \bibinfo {author} {\bibfnamefont {K.}~\bibnamefont
  {Thirumalai}}, \bibinfo {author} {\bibfnamefont {L.~J.}\ \bibnamefont
  {Stephenson}}, \bibinfo {author} {\bibfnamefont {T.~G.}\ \bibnamefont
  {Ballance}}, \bibinfo {author} {\bibfnamefont {A.~M.}\ \bibnamefont
  {Steane}},\ and\ \bibinfo {author} {\bibfnamefont {D.~M.}\ \bibnamefont
  {Lucas}},\ }\href {https://doi.org/10.1038/nature25737} {\bibfield  {journal}
  {\bibinfo  {journal} {Nature}\ }\textbf {\bibinfo {volume} {555}},\ \bibinfo
  {pages} {75} (\bibinfo {year} {2018})}\BibitemShut {NoStop}%
\bibitem [{\citenamefont {Nielsen}\ and\ \citenamefont
  {Chuang}(2011)}]{nielsen_quantum_2011}%
  \BibitemOpen
  \bibfield  {author} {\bibinfo {author} {\bibfnamefont {M.~A.}\ \bibnamefont
  {Nielsen}}\ and\ \bibinfo {author} {\bibfnamefont {I.~L.}\ \bibnamefont
  {Chuang}},\ }\href@noop {} {\emph {\bibinfo {title} {Quantum {Computation}
  and {Quantum} {Information}: 10th {Anniversary} {Edition}}}},\ \bibinfo
  {edition} {10th}\ ed.\ (\bibinfo  {publisher} {Cambridge University Press},\
  \bibinfo {address} {USA},\ \bibinfo {year} {2011})\BibitemShut {NoStop}%
\bibitem [{\citenamefont {Shende}\ and\ \citenamefont
  {Markov}(2009)}]{shende_cnot-cost_2009}%
  \BibitemOpen
  \bibfield  {author} {\bibinfo {author} {\bibfnamefont {V.~V.}\ \bibnamefont
  {Shende}}\ and\ \bibinfo {author} {\bibfnamefont {I.~L.}\ \bibnamefont
  {Markov}},\ }\href@noop {} {\bibfield  {journal} {\bibinfo  {journal}
  {Quantum Information \& Computation}\ }\textbf {\bibinfo {volume} {9}},\
  \bibinfo {pages} {461} (\bibinfo {year} {2009})}\BibitemShut {NoStop}%
\bibitem [{\citenamefont {Barenco}\ \emph {et~al.}(1995)\citenamefont
  {Barenco}, \citenamefont {Bennett}, \citenamefont {Cleve}, \citenamefont
  {DiVincenzo}, \citenamefont {Margolus}, \citenamefont {Shor}, \citenamefont
  {Sleator}, \citenamefont {Smolin},\ and\ \citenamefont
  {Weinfurter}}]{barenco_elementary_1995}%
  \BibitemOpen
  \bibfield  {author} {\bibinfo {author} {\bibfnamefont {A.}~\bibnamefont
  {Barenco}}, \bibinfo {author} {\bibfnamefont {C.~H.}\ \bibnamefont
  {Bennett}}, \bibinfo {author} {\bibfnamefont {R.}~\bibnamefont {Cleve}},
  \bibinfo {author} {\bibfnamefont {D.~P.}\ \bibnamefont {DiVincenzo}},
  \bibinfo {author} {\bibfnamefont {N.}~\bibnamefont {Margolus}}, \bibinfo
  {author} {\bibfnamefont {P.}~\bibnamefont {Shor}}, \bibinfo {author}
  {\bibfnamefont {T.}~\bibnamefont {Sleator}}, \bibinfo {author} {\bibfnamefont
  {J.~A.}\ \bibnamefont {Smolin}},\ and\ \bibinfo {author} {\bibfnamefont
  {H.}~\bibnamefont {Weinfurter}},\ }\href
  {https://doi.org/10.1103/PhysRevA.52.3457} {\bibfield  {journal} {\bibinfo
  {journal} {Physical Review A}\ }\textbf {\bibinfo {volume} {52}},\ \bibinfo
  {pages} {3457} (\bibinfo {year} {1995})}\BibitemShut {NoStop}%
\bibitem [{Note14()}]{Note14}%
  \BibitemOpen
  \bibinfo {note} {By ``secure against quantum adversaries'' we mean that it is
  possible to ensure that the quantum prover is actually following the
  prescribed protocol, to whatever extent is necessary for the desired
  cryptographic task.}\BibitemShut {Stop}%
\bibitem [{\citenamefont {Pietrzak}(2012)}]{pietrzak_cryptography_2012}%
  \BibitemOpen
  \bibfield  {author} {\bibinfo {author} {\bibfnamefont {K.}~\bibnamefont
  {Pietrzak}},\ }in\ \href {https://doi.org/10.1007/978-3-642-27660-6_9} {\emph
  {\bibinfo {booktitle} {{SOFSEM} 2012: {Theory} and {Practice} of {Computer}
  {Science}}}},\ \bibinfo {series and number} {Lecture {Notes} in {Computer}
  {Science}},\ \bibinfo {editor} {edited by\ \bibinfo {editor} {\bibfnamefont
  {M.}~\bibnamefont {Bieliková}}, \bibinfo {editor} {\bibfnamefont
  {G.}~\bibnamefont {Friedrich}}, \bibinfo {editor} {\bibfnamefont
  {G.}~\bibnamefont {Gottlob}}, \bibinfo {editor} {\bibfnamefont
  {S.}~\bibnamefont {Katzenbeisser}},\ and\ \bibinfo {editor} {\bibfnamefont
  {G.}~\bibnamefont {Turán}}}\ (\bibinfo  {publisher} {Springer},\ \bibinfo
  {address} {Berlin, Heidelberg},\ \bibinfo {year} {2012})\ pp.\ \bibinfo
  {pages} {99--114}\BibitemShut {NoStop}%
\bibitem [{\citenamefont {Puri}\ \emph {et~al.}(2020)\citenamefont {Puri},
  \citenamefont {St-Jean}, \citenamefont {Gross}, \citenamefont {Grimm},
  \citenamefont {Frattini}, \citenamefont {Iyer}, \citenamefont {Krishna},
  \citenamefont {Touzard}, \citenamefont {Jiang}, \citenamefont {Blais},
  \citenamefont {Flammia},\ and\ \citenamefont
  {Girvin}}]{puri_bias-preserving_2020}%
  \BibitemOpen
  \bibfield  {author} {\bibinfo {author} {\bibfnamefont {S.}~\bibnamefont
  {Puri}}, \bibinfo {author} {\bibfnamefont {L.}~\bibnamefont {St-Jean}},
  \bibinfo {author} {\bibfnamefont {J.~A.}\ \bibnamefont {Gross}}, \bibinfo
  {author} {\bibfnamefont {A.}~\bibnamefont {Grimm}}, \bibinfo {author}
  {\bibfnamefont {N.~E.}\ \bibnamefont {Frattini}}, \bibinfo {author}
  {\bibfnamefont {P.~S.}\ \bibnamefont {Iyer}}, \bibinfo {author}
  {\bibfnamefont {A.}~\bibnamefont {Krishna}}, \bibinfo {author} {\bibfnamefont
  {S.}~\bibnamefont {Touzard}}, \bibinfo {author} {\bibfnamefont
  {L.}~\bibnamefont {Jiang}}, \bibinfo {author} {\bibfnamefont
  {A.}~\bibnamefont {Blais}}, \bibinfo {author} {\bibfnamefont {S.~T.}\
  \bibnamefont {Flammia}},\ and\ \bibinfo {author} {\bibfnamefont {S.~M.}\
  \bibnamefont {Girvin}},\ }\href {https://doi.org/10.1126/sciadv.aay5901}
  {\bibfield  {journal} {\bibinfo  {journal} {Science Advances}\ }\textbf
  {\bibinfo {volume} {6}},\ \bibinfo {pages} {eaay5901} (\bibinfo {year}
  {2020})}\BibitemShut {NoStop}%
\bibitem [{\citenamefont {Meyer}(2022)}]{cirq_code}%
  \BibitemOpen
  \bibfield  {author} {\bibinfo {author} {\bibfnamefont {G.}~\bibnamefont
  {Meyer}},\ }\href {https://doi.org/10.5281/zenodo.6519250} {\bibinfo {title}
  {Gregdmeyer/quantum-advantage: v1.1}} (\bibinfo {year} {2022})\BibitemShut
  {NoStop}%
\end{thebibliography}%


\begin{thebibliography}{12}%
\makeatletter
\providecommand \@ifxundefined [1]{%
 \@ifx{#1\undefined}
}%
\providecommand \@ifnum [1]{%
 \ifnum #1\expandafter \@firstoftwo
 \else \expandafter \@secondoftwo
 \fi
}%
\providecommand \@ifx [1]{%
 \ifx #1\expandafter \@firstoftwo
 \else \expandafter \@secondoftwo
 \fi
}%
\providecommand \natexlab [1]{#1}%
\providecommand \enquote  [1]{``#1''}%
\providecommand \bibnamefont  [1]{#1}%
\providecommand \bibfnamefont [1]{#1}%
\providecommand \citenamefont [1]{#1}%
\providecommand \href@noop [0]{\@secondoftwo}%
\providecommand \href [0]{\begingroup \@sanitize@url \@href}%
\providecommand \@href[1]{\@@startlink{#1}\@@href}%
\providecommand \@@href[1]{\endgroup#1\@@endlink}%
\providecommand \@sanitize@url [0]{\catcode `\\12\catcode `\$12\catcode
  `\&12\catcode `\#12\catcode `\^12\catcode `\_12\catcode `\%12\relax}%
\providecommand \@@startlink[1]{}%
\providecommand \@@endlink[0]{}%
\providecommand \url  [0]{\begingroup\@sanitize@url \@url }%
\providecommand \@url [1]{\endgroup\@href {#1}{\urlprefix }}%
\providecommand \urlprefix  [0]{URL }%
\providecommand \Eprint [0]{\href }%
\providecommand \doibase [0]{https://doi.org/}%
\providecommand \selectlanguage [0]{\@gobble}%
\providecommand \bibinfo  [0]{\@secondoftwo}%
\providecommand \bibfield  [0]{\@secondoftwo}%
\providecommand \translation [1]{[#1]}%
\providecommand \BibitemOpen [0]{}%
\providecommand \bibitemStop [0]{}%
\providecommand \bibitemNoStop [0]{.\EOS\space}%
\providecommand \EOS [0]{\spacefactor3000\relax}%
\providecommand \BibitemShut  [1]{\csname bibitem#1\endcsname}%
\let\auto@bib@innerbib\@empty
\bibitem [{\citenamefont {Goldreich}\ and\ \citenamefont
  {Levint}(1989)}]{goldreich_hard-core_1989}%
  \BibitemOpen
  \bibfield  {author} {\bibinfo {author} {\bibfnamefont {O.}~\bibnamefont
  {Goldreich}}\ and\ \bibinfo {author} {\bibfnamefont {L.~A.}\ \bibnamefont
  {Levint}},\ }in\ \href@noop {} {\emph {\bibinfo {booktitle} {In {Proceedings}
  of the {Twenty} {First} {Annual} {ACM} {Symposium} on {Theory} of
  {Computing}}}}\ (\bibinfo {year} {1989})\ pp.\ \bibinfo {pages}
  {25--32}\BibitemShut {NoStop}%
\bibitem [{\citenamefont {Rabin}(1979)}]{rabin_digitalized_1979}%
  \BibitemOpen
  \bibfield  {author} {\bibinfo {author} {\bibfnamefont {M.~O.}\ \bibnamefont
  {Rabin}},\ }\href@noop {} {\emph {\bibinfo {title} {Digitalized signatures
  and public-key functions as intractable as factorization}}},\ \bibinfo {type}
  {Technical {Report}}\ (\bibinfo  {institution} {Massachusetts Institute of
  Technology},\ \bibinfo {address} {USA},\ \bibinfo {year} {1979})\BibitemShut
  {NoStop}%
\bibitem [{\citenamefont {Goldwasser}\ \emph {et~al.}(1988)\citenamefont
  {Goldwasser}, \citenamefont {Micali},\ and\ \citenamefont
  {Rivest}}]{goldwasser_digital_1988}%
  \BibitemOpen
  \bibfield  {author} {\bibinfo {author} {\bibfnamefont {S.}~\bibnamefont
  {Goldwasser}}, \bibinfo {author} {\bibfnamefont {S.}~\bibnamefont {Micali}},\
  and\ \bibinfo {author} {\bibfnamefont {R.~L.}\ \bibnamefont {Rivest}},\
  }\href {https://doi.org/10.1137/0217017} {\bibfield  {journal} {\bibinfo
  {journal} {SIAM Journal on Computing}\ }\textbf {\bibinfo {volume} {17}},\
  \bibinfo {pages} {281} (\bibinfo {year} {1988})}\BibitemShut {NoStop}%
\bibitem [{Note1()}]{Note1}%
  \BibitemOpen
  \bibinfo {note} {In practice, $p$ and $q$ must be selected with some care
  such that Fermat factorization and Pollard's $p-1$ algorithm~\cite
  {pollard_theorems_1974} cannot be used to efficiently factor $N$ classically.
  Selecting $p$ and $q$ in the same manner as for RSA encryption would be
  effective ~\cite {rivest_method_1978}.}\BibitemShut {Stop}%
\bibitem [{\citenamefont {Peikert}\ and\ \citenamefont
  {Waters}(2008)}]{peikert_lossy_2008}%
  \BibitemOpen
  \bibfield  {author} {\bibinfo {author} {\bibfnamefont {C.}~\bibnamefont
  {Peikert}}\ and\ \bibinfo {author} {\bibfnamefont {B.}~\bibnamefont
  {Waters}},\ }in\ \href {https://doi.org/10.1145/1374376.1374406} {\emph
  {\bibinfo {booktitle} {Proceedings of the fortieth annual {ACM} symposium on
  {Theory} of computing}}},\ \bibinfo {series and number} {{STOC} '08}\
  (\bibinfo  {publisher} {Association for Computing Machinery},\ \bibinfo
  {address} {New York, NY, USA},\ \bibinfo {year} {2008})\ pp.\ \bibinfo
  {pages} {187--196}\BibitemShut {NoStop}%
\bibitem [{\citenamefont {Freeman}\ \emph {et~al.}(2010)\citenamefont
  {Freeman}, \citenamefont {Goldreich}, \citenamefont {Kiltz}, \citenamefont
  {Rosen},\ and\ \citenamefont {Segev}}]{freeman_more_2010}%
  \BibitemOpen
  \bibfield  {author} {\bibinfo {author} {\bibfnamefont {D.~M.}\ \bibnamefont
  {Freeman}}, \bibinfo {author} {\bibfnamefont {O.}~\bibnamefont {Goldreich}},
  \bibinfo {author} {\bibfnamefont {E.}~\bibnamefont {Kiltz}}, \bibinfo
  {author} {\bibfnamefont {A.}~\bibnamefont {Rosen}},\ and\ \bibinfo {author}
  {\bibfnamefont {G.}~\bibnamefont {Segev}},\ }in\ \href
  {https://doi.org/10.1007/978-3-642-13013-7_17} {\emph {\bibinfo {booktitle}
  {Public {Key} {Cryptography} – {PKC} 2010}}},\ \bibinfo {series and number}
  {Lecture {Notes} in {Computer} {Science}},\ \bibinfo {editor} {edited by\
  \bibinfo {editor} {\bibfnamefont {P.~Q.}\ \bibnamefont {Nguyen}}\ and\
  \bibinfo {editor} {\bibfnamefont {D.}~\bibnamefont {Pointcheval}}}\ (\bibinfo
   {publisher} {Springer},\ \bibinfo {address} {Berlin, Heidelberg},\ \bibinfo
  {year} {2010})\ pp.\ \bibinfo {pages} {279--295}\BibitemShut {NoStop}%
\bibitem [{\citenamefont {Gidney}(2019)}]{gidney_asymptotically_2019}%
  \BibitemOpen
  \bibfield  {author} {\bibinfo {author} {\bibfnamefont {C.}~\bibnamefont
  {Gidney}},\ }\href {http://arxiv.org/abs/1904.07356} {\bibfield  {journal}
  {\bibinfo  {journal} {arXiv:1904.07356 [quant-ph]}\ } (\bibinfo {year}
  {2019})}\BibitemShut {NoStop}%
\bibitem [{\citenamefont {Gidney}\ and\ \citenamefont
  {Ekerå}(2019)}]{gidney_how_2019}%
  \BibitemOpen
  \bibfield  {author} {\bibinfo {author} {\bibfnamefont {C.}~\bibnamefont
  {Gidney}}\ and\ \bibinfo {author} {\bibfnamefont {M.}~\bibnamefont
  {Ekerå}},\ }\href {http://arxiv.org/abs/1905.09749} {\bibfield  {journal}
  {\bibinfo  {journal} {arXiv:1905.09749 [quant-ph]}\ } (\bibinfo {year}
  {2019})}\BibitemShut {NoStop}%
\bibitem [{noa()}]{noauthor_cado-nfs_nodate}%
  \BibitemOpen
  \href {http://cado-nfs.gforge.inria.fr/} {\bibinfo {title} {{CADO}-{NFS}}},\
  \bibinfo {howpublished} {\url{http://cado-nfs.gforge.inria.fr/}},\ \bibinfo
  {note} {accessed: 2020-06-27}\BibitemShut {NoStop}%
\bibitem [{\citenamefont
  {Zimmermann}(2020)}]{zimmermann_cado-nfs-discuss_2020}%
  \BibitemOpen
  \bibfield  {author} {\bibinfo {author} {\bibfnamefont {P.}~\bibnamefont
  {Zimmermann}},\ }\href@noop {} {\bibinfo {title} {{Factorization} of
  {RSA}-250}},\ \bibinfo {howpublished}
  {\url{https://lists.gforge.inria.fr/pipermail/cado-nfs-discuss/2020-February/001166.html}}
  (\bibinfo {year} {2020}),\ \bibinfo {note} {accessed: 2020-06-27}\BibitemShut
  {NoStop}%
\bibitem [{\citenamefont {Pollard}(1974)}]{pollard_theorems_1974}%
  \BibitemOpen
  \bibfield  {author} {\bibinfo {author} {\bibfnamefont {J.~M.}\ \bibnamefont
  {Pollard}},\ }\href {https://doi.org/10.1017/S0305004100049252} {\bibfield
  {journal} {\bibinfo  {journal} {Mathematical Proceedings of the Cambridge
  Philosophical Society}\ }\textbf {\bibinfo {volume} {76}},\ \bibinfo {pages}
  {521} (\bibinfo {year} {1974})}\BibitemShut {NoStop}%
\bibitem [{\citenamefont {Rivest}\ \emph {et~al.}(1978)\citenamefont {Rivest},
  \citenamefont {Shamir},\ and\ \citenamefont {Adleman}}]{rivest_method_1978}%
  \BibitemOpen
  \bibfield  {author} {\bibinfo {author} {\bibfnamefont {R.~L.}\ \bibnamefont
  {Rivest}}, \bibinfo {author} {\bibfnamefont {A.}~\bibnamefont {Shamir}},\
  and\ \bibinfo {author} {\bibfnamefont {L.}~\bibnamefont {Adleman}},\ }\href
  {https://doi.org/10.1145/359340.359342} {\bibfield  {journal} {\bibinfo
  {journal} {Communications of the ACM}\ }\textbf {\bibinfo {volume} {21}},\
  \bibinfo {pages} {120} (\bibinfo {year} {1978})}\BibitemShut {NoStop}%
\end{thebibliography}%


\begin{thebibliography}{22}%
\makeatletter
\providecommand \@ifxundefined [1]{%
 \@ifx{#1\undefined}
}%
\providecommand \@ifnum [1]{%
 \ifnum #1\expandafter \@firstoftwo
 \else \expandafter \@secondoftwo
 \fi
}%
\providecommand \@ifx [1]{%
 \ifx #1\expandafter \@firstoftwo
 \else \expandafter \@secondoftwo
 \fi
}%
\providecommand \natexlab [1]{#1}%
\providecommand \enquote  [1]{``#1''}%
\providecommand \bibnamefont  [1]{#1}%
\providecommand \bibfnamefont [1]{#1}%
\providecommand \citenamefont [1]{#1}%
\providecommand \href@noop [0]{\@secondoftwo}%
\providecommand \href [0]{\begingroup \@sanitize@url \@href}%
\providecommand \@href[1]{\@@startlink{#1}\@@href}%
\providecommand \@@href[1]{\endgroup#1\@@endlink}%
\providecommand \@sanitize@url [0]{\catcode `\\12\catcode `\$12\catcode
  `\&12\catcode `\#12\catcode `\^12\catcode `\_12\catcode `\%12\relax}%
\providecommand \@@startlink[1]{}%
\providecommand \@@endlink[0]{}%
\providecommand \url  [0]{\begingroup\@sanitize@url \@url }%
\providecommand \@url [1]{\endgroup\@href {#1}{\urlprefix }}%
\providecommand \urlprefix  [0]{URL }%
\providecommand \Eprint [0]{\href }%
\providecommand \doibase [0]{https://doi.org/}%
\providecommand \selectlanguage [0]{\@gobble}%
\providecommand \bibinfo  [0]{\@secondoftwo}%
\providecommand \bibfield  [0]{\@secondoftwo}%
\providecommand \translation [1]{[#1]}%
\providecommand \BibitemOpen [0]{}%
\providecommand \bibitemStop [0]{}%
\providecommand \bibitemNoStop [0]{.\EOS\space}%
\providecommand \EOS [0]{\spacefactor3000\relax}%
\providecommand \BibitemShut  [1]{\csname bibitem#1\endcsname}%
\let\auto@bib@innerbib\@empty
\bibitem [{\citenamefont {Brakerski}\ \emph {et~al.}(2019)\citenamefont
  {Brakerski}, \citenamefont {Christiano}, \citenamefont {Mahadev},
  \citenamefont {Vazirani},\ and\ \citenamefont
  {Vidick}}]{brakerski_cryptographic_2019}%
  \BibitemOpen
  \bibfield  {author} {\bibinfo {author} {\bibfnamefont {Z.}~\bibnamefont
  {Brakerski}}, \bibinfo {author} {\bibfnamefont {P.}~\bibnamefont
  {Christiano}}, \bibinfo {author} {\bibfnamefont {U.}~\bibnamefont {Mahadev}},
  \bibinfo {author} {\bibfnamefont {U.}~\bibnamefont {Vazirani}},\ and\
  \bibinfo {author} {\bibfnamefont {T.}~\bibnamefont {Vidick}},\ }\href
  {http://arxiv.org/abs/1804.00640} {\bibfield  {journal} {\bibinfo  {journal}
  {arXiv:1804.00640 [quant-ph]}\ } (\bibinfo {year} {2019})}\BibitemShut
  {NoStop}%
\bibitem [{\citenamefont {Rabin}(1979)}]{rabin_digitalized_1979}%
  \BibitemOpen
  \bibfield  {author} {\bibinfo {author} {\bibfnamefont {M.~O.}\ \bibnamefont
  {Rabin}},\ }\href@noop {} {\emph {\bibinfo {title} {Digitalized signatures
  and public-key functions as intractable as factorization}}},\ \bibinfo {type}
  {Technical {Report}}\ (\bibinfo  {institution} {Massachusetts Institute of
  Technology},\ \bibinfo {address} {USA},\ \bibinfo {year} {1979})\BibitemShut
  {NoStop}%
\bibitem [{\citenamefont {Peikert}\ and\ \citenamefont
  {Waters}(2008)}]{peikert_lossy_2008}%
  \BibitemOpen
  \bibfield  {author} {\bibinfo {author} {\bibfnamefont {C.}~\bibnamefont
  {Peikert}}\ and\ \bibinfo {author} {\bibfnamefont {B.}~\bibnamefont
  {Waters}},\ }in\ \href {https://doi.org/10.1145/1374376.1374406} {\emph
  {\bibinfo {booktitle} {Proceedings of the fortieth annual {ACM} symposium on
  {Theory} of computing}}},\ \bibinfo {series and number} {{STOC} '08}\
  (\bibinfo  {publisher} {Association for Computing Machinery},\ \bibinfo
  {address} {New York, NY, USA},\ \bibinfo {year} {2008})\ pp.\ \bibinfo
  {pages} {187--196}\BibitemShut {NoStop}%
\bibitem [{\citenamefont {Freeman}\ \emph {et~al.}(2010)\citenamefont
  {Freeman}, \citenamefont {Goldreich}, \citenamefont {Kiltz}, \citenamefont
  {Rosen},\ and\ \citenamefont {Segev}}]{freeman_more_2010}%
  \BibitemOpen
  \bibfield  {author} {\bibinfo {author} {\bibfnamefont {D.~M.}\ \bibnamefont
  {Freeman}}, \bibinfo {author} {\bibfnamefont {O.}~\bibnamefont {Goldreich}},
  \bibinfo {author} {\bibfnamefont {E.}~\bibnamefont {Kiltz}}, \bibinfo
  {author} {\bibfnamefont {A.}~\bibnamefont {Rosen}},\ and\ \bibinfo {author}
  {\bibfnamefont {G.}~\bibnamefont {Segev}},\ }in\ \href
  {https://doi.org/10.1007/978-3-642-13013-7_17} {\emph {\bibinfo {booktitle}
  {Public {Key} {Cryptography} – {PKC} 2010}}},\ \bibinfo {series and number}
  {Lecture {Notes} in {Computer} {Science}},\ \bibinfo {editor} {edited by\
  \bibinfo {editor} {\bibfnamefont {P.~Q.}\ \bibnamefont {Nguyen}}\ and\
  \bibinfo {editor} {\bibfnamefont {D.}~\bibnamefont {Pointcheval}}}\ (\bibinfo
   {publisher} {Springer},\ \bibinfo {address} {Berlin, Heidelberg},\ \bibinfo
  {year} {2010})\ pp.\ \bibinfo {pages} {279--295}\BibitemShut {NoStop}%
\bibitem [{\citenamefont {Regev}(2005)}]{regev_lattices_2005}%
  \BibitemOpen
  \bibfield  {author} {\bibinfo {author} {\bibfnamefont {O.}~\bibnamefont
  {Regev}},\ }in\ \href {https://doi.org/10.1145/1060590.1060603} {\emph
  {\bibinfo {booktitle} {Proceedings of the thirty-seventh annual {ACM}
  symposium on {Theory} of computing}}},\ \bibinfo {series and number} {{STOC}
  '05}\ (\bibinfo  {publisher} {Association for Computing Machinery},\ \bibinfo
  {address} {New York, NY, USA},\ \bibinfo {year} {2005})\ pp.\ \bibinfo
  {pages} {84--93}\BibitemShut {NoStop}%
\bibitem [{\citenamefont {Brakerski}\ \emph {et~al.}(2020)\citenamefont
  {Brakerski}, \citenamefont {Koppula}, \citenamefont {Vazirani},\ and\
  \citenamefont {Vidick}}]{brakerski_simpler_2020}%
  \BibitemOpen
  \bibfield  {author} {\bibinfo {author} {\bibfnamefont {Z.}~\bibnamefont
  {Brakerski}}, \bibinfo {author} {\bibfnamefont {V.}~\bibnamefont {Koppula}},
  \bibinfo {author} {\bibfnamefont {U.}~\bibnamefont {Vazirani}},\ and\
  \bibinfo {author} {\bibfnamefont {T.}~\bibnamefont {Vidick}},\ }\href
  {http://arxiv.org/abs/2005.04826} {\bibfield  {journal} {\bibinfo  {journal}
  {arXiv:2005.04826 [quant-ph]}\ } (\bibinfo {year} {2020})}\BibitemShut
  {NoStop}%
\bibitem [{\citenamefont {Lyubashevsky}\ \emph {et~al.}(2012)\citenamefont
  {Lyubashevsky}, \citenamefont {Peikert},\ and\ \citenamefont
  {Regev}}]{lyubashevsky_ideal_2012}%
  \BibitemOpen
  \bibfield  {author} {\bibinfo {author} {\bibfnamefont {V.}~\bibnamefont
  {Lyubashevsky}}, \bibinfo {author} {\bibfnamefont {C.}~\bibnamefont
  {Peikert}},\ and\ \bibinfo {author} {\bibfnamefont {O.}~\bibnamefont
  {Regev}},\ }\href {https://eprint.iacr.org/2012/230} {\emph {\bibinfo {title}
  {On {Ideal} {Lattices} and {Learning} with {Errors} {Over} {Rings}}}},\
  \bibinfo {type} {Tech. Rep.}\ \bibinfo {number} {230}\ (\bibinfo {year}
  {2012})\BibitemShut {NoStop}%
\bibitem [{\citenamefont {de~Clercq}\ \emph {et~al.}(2015)\citenamefont
  {de~Clercq}, \citenamefont {Roy}, \citenamefont {Vercauteren},\ and\
  \citenamefont {Verbauwhede}}]{de_clercq_efficient_2015}%
  \BibitemOpen
  \bibfield  {author} {\bibinfo {author} {\bibfnamefont {R.}~\bibnamefont
  {de~Clercq}}, \bibinfo {author} {\bibfnamefont {S.~S.}\ \bibnamefont {Roy}},
  \bibinfo {author} {\bibfnamefont {F.}~\bibnamefont {Vercauteren}},\ and\
  \bibinfo {author} {\bibfnamefont {I.}~\bibnamefont {Verbauwhede}},\ }in\
  \href@noop {} {\emph {\bibinfo {booktitle} {Proceedings of the 2015 {Design},
  {Automation} \& {Test} in {Europe} {Conference} \& {Exhibition}}}},\ \bibinfo
  {series and number} {{DATE} '15}\ (\bibinfo  {publisher} {EDA Consortium},\
  \bibinfo {address} {San Jose, CA, USA},\ \bibinfo {year} {2015})\ pp.\
  \bibinfo {pages} {339--344}\BibitemShut {NoStop}%
\bibitem [{\citenamefont {Roy}\ \emph {et~al.}(2014)\citenamefont {Roy},
  \citenamefont {Vercauteren}, \citenamefont {Mentens}, \citenamefont {Chen},\
  and\ \citenamefont {Verbauwhede}}]{roy_compact_2014}%
  \BibitemOpen
  \bibfield  {author} {\bibinfo {author} {\bibfnamefont {S.~S.}\ \bibnamefont
  {Roy}}, \bibinfo {author} {\bibfnamefont {F.}~\bibnamefont {Vercauteren}},
  \bibinfo {author} {\bibfnamefont {N.}~\bibnamefont {Mentens}}, \bibinfo
  {author} {\bibfnamefont {D.~D.}\ \bibnamefont {Chen}},\ and\ \bibinfo
  {author} {\bibfnamefont {I.}~\bibnamefont {Verbauwhede}},\ }in\ \href
  {https://doi.org/10.1007/978-3-662-44709-3_21} {\emph {\bibinfo {booktitle}
  {Cryptographic {Hardware} and {Embedded} {Systems} – {CHES} 2014}}},\
  \bibinfo {series and number} {Lecture {Notes} in {Computer} {Science}},\
  \bibinfo {editor} {edited by\ \bibinfo {editor} {\bibfnamefont
  {L.}~\bibnamefont {Batina}}\ and\ \bibinfo {editor} {\bibfnamefont
  {M.}~\bibnamefont {Robshaw}}}\ (\bibinfo  {publisher} {Springer},\ \bibinfo
  {address} {Berlin, Heidelberg},\ \bibinfo {year} {2014})\ pp.\ \bibinfo
  {pages} {371--391}\BibitemShut {NoStop}%
\bibitem [{\citenamefont {Diffie}\ and\ \citenamefont
  {Hellman}(1976)}]{diffie_new_1976}%
  \BibitemOpen
  \bibfield  {author} {\bibinfo {author} {\bibfnamefont {W.}~\bibnamefont
  {Diffie}}\ and\ \bibinfo {author} {\bibfnamefont {M.}~\bibnamefont
  {Hellman}},\ }\href {https://doi.org/10.1109/TIT.1976.1055638} {\bibfield
  {journal} {\bibinfo  {journal} {IEEE Transactions on Information Theory}\
  }\textbf {\bibinfo {volume} {22}},\ \bibinfo {pages} {644} (\bibinfo {year}
  {1976})}\BibitemShut {NoStop}%
\bibitem [{\citenamefont {Shor}(1997)}]{shor_polynomial-time_1997}%
  \BibitemOpen
  \bibfield  {author} {\bibinfo {author} {\bibfnamefont {P.~W.}\ \bibnamefont
  {Shor}},\ }\href {https://doi.org/10.1137/S0097539795293172} {\bibfield
  {journal} {\bibinfo  {journal} {SIAM Journal on Computing}\ }\textbf
  {\bibinfo {volume} {26}},\ \bibinfo {pages} {1484} (\bibinfo {year}
  {1997})}\BibitemShut {NoStop}%
\bibitem [{\citenamefont {Zalka}(1998)}]{zalka_fast_1998}%
  \BibitemOpen
  \bibfield  {author} {\bibinfo {author} {\bibfnamefont {C.}~\bibnamefont
  {Zalka}},\ }\href {http://arxiv.org/abs/quant-ph/9806084} {\bibfield
  {journal} {\bibinfo  {journal} {arXiv:quant-ph/9806084}\ } (\bibinfo {year}
  {1998})}\BibitemShut {NoStop}%
\bibitem [{\citenamefont {Schönhage}\ and\ \citenamefont
  {Strassen}(1971)}]{schonhage_schnelle_1971}%
  \BibitemOpen
  \bibfield  {author} {\bibinfo {author} {\bibfnamefont {A.}~\bibnamefont
  {Schönhage}}\ and\ \bibinfo {author} {\bibfnamefont {V.}~\bibnamefont
  {Strassen}},\ }\href {https://doi.org/10.1007/BF02242355} {\bibfield
  {journal} {\bibinfo  {journal} {Computing}\ }\textbf {\bibinfo {volume}
  {7}},\ \bibinfo {pages} {281} (\bibinfo {year} {1971})}\BibitemShut {NoStop}%
\bibitem [{\citenamefont {Kowada}\ \emph {et~al.}(2006)\citenamefont {Kowada},
  \citenamefont {Portugal},\ and\ \citenamefont
  {Figueiredo}}]{kowada_reversible_2006}%
  \BibitemOpen
  \bibfield  {author} {\bibinfo {author} {\bibfnamefont {L.}~\bibnamefont
  {Kowada}}, \bibinfo {author} {\bibfnamefont {R.}~\bibnamefont {Portugal}},\
  and\ \bibinfo {author} {\bibfnamefont {C.}~\bibnamefont {Figueiredo}},\
  }\href@noop {} {\bibfield  {journal} {\bibinfo  {journal} {J. UCS}\ }\textbf
  {\bibinfo {volume} {12}},\ \bibinfo {pages} {499} (\bibinfo {year}
  {2006})}\BibitemShut {NoStop}%
\bibitem [{\citenamefont {Parent}\ \emph {et~al.}(2017)\citenamefont {Parent},
  \citenamefont {Roetteler},\ and\ \citenamefont
  {Mosca}}]{parent_improved_2017}%
  \BibitemOpen
  \bibfield  {author} {\bibinfo {author} {\bibfnamefont {A.}~\bibnamefont
  {Parent}}, \bibinfo {author} {\bibfnamefont {M.}~\bibnamefont {Roetteler}},\
  and\ \bibinfo {author} {\bibfnamefont {M.}~\bibnamefont {Mosca}},\ }\href
  {http://arxiv.org/abs/1706.03419} {\bibfield  {journal} {\bibinfo  {journal}
  {arXiv:1706.03419 [quant-ph]}\ } (\bibinfo {year} {2017})}\BibitemShut
  {NoStop}%
\bibitem [{\citenamefont {Gidney}(2019)}]{gidney_asymptotically_2019}%
  \BibitemOpen
  \bibfield  {author} {\bibinfo {author} {\bibfnamefont {C.}~\bibnamefont
  {Gidney}},\ }\href {http://arxiv.org/abs/1904.07356} {\bibfield  {journal}
  {\bibinfo  {journal} {arXiv:1904.07356 [quant-ph]}\ } (\bibinfo {year}
  {2019})}\BibitemShut {NoStop}%
\bibitem [{\citenamefont {Montgomery}(1985)}]{montgomery_modular_1985}%
  \BibitemOpen
  \bibfield  {author} {\bibinfo {author} {\bibfnamefont {P.~L.}\ \bibnamefont
  {Montgomery}},\ }\href {https://doi.org/10.1090/S0025-5718-1985-0777282-X}
  {\bibfield  {journal} {\bibinfo  {journal} {Mathematics of Computation}\
  }\textbf {\bibinfo {volume} {44}},\ \bibinfo {pages} {519} (\bibinfo {year}
  {1985})}\BibitemShut {NoStop}%
\bibitem [{\citenamefont {Javeed}\ \emph {et~al.}(2016)\citenamefont {Javeed},
  \citenamefont {Irwin},\ and\ \citenamefont {Wang}}]{javeed_design_2016}%
  \BibitemOpen
  \bibfield  {author} {\bibinfo {author} {\bibfnamefont {K.}~\bibnamefont
  {Javeed}}, \bibinfo {author} {\bibfnamefont {D.}~\bibnamefont {Irwin}},\ and\
  \bibinfo {author} {\bibfnamefont {X.}~\bibnamefont {Wang}},\ }in\ \href
  {https://doi.org/10.1007/978-3-319-48671-0_23} {\emph {\bibinfo {booktitle}
  {Cloud {Computing} and {Security}}}},\ \bibinfo {series and number} {Lecture
  {Notes} in {Computer} {Science}},\ \bibinfo {editor} {edited by\ \bibinfo
  {editor} {\bibfnamefont {X.}~\bibnamefont {Sun}}, \bibinfo {editor}
  {\bibfnamefont {A.}~\bibnamefont {Liu}}, \bibinfo {editor} {\bibfnamefont
  {H.-C.}\ \bibnamefont {Chao}},\ and\ \bibinfo {editor} {\bibfnamefont
  {E.}~\bibnamefont {Bertino}}}\ (\bibinfo  {publisher} {Springer International
  Publishing},\ \bibinfo {address} {Cham},\ \bibinfo {year} {2016})\ pp.\
  \bibinfo {pages} {251--260}\BibitemShut {NoStop}%
\bibitem [{\citenamefont {Morales‐Sandoval}\ and\ \citenamefont
  {Diaz‐Perez}(2016)}]{moralessandoval_scalable_2016}%
  \BibitemOpen
  \bibfield  {author} {\bibinfo {author} {\bibfnamefont {M.}~\bibnamefont
  {Morales‐Sandoval}}\ and\ \bibinfo {author} {\bibfnamefont
  {A.}~\bibnamefont {Diaz‐Perez}},\ }\href
  {https://doi.org/https://doi.org/10.1049/iet-cdt.2015.0055} {\bibfield
  {journal} {\bibinfo  {journal} {IET Computers \& Digital Techniques}\
  }\textbf {\bibinfo {volume} {10}},\ \bibinfo {pages} {102} (\bibinfo {year}
  {2016})}\BibitemShut {NoStop}%
\bibitem [{\citenamefont {Yang}\ \emph {et~al.}(2016)\citenamefont {Yang},
  \citenamefont {Wu}, \citenamefont {Li},\ and\ \citenamefont
  {Yang}}]{yang_efficient_2016}%
  \BibitemOpen
  \bibfield  {author} {\bibinfo {author} {\bibfnamefont {Y.}~\bibnamefont
  {Yang}}, \bibinfo {author} {\bibfnamefont {C.}~\bibnamefont {Wu}}, \bibinfo
  {author} {\bibfnamefont {Z.}~\bibnamefont {Li}},\ and\ \bibinfo {author}
  {\bibfnamefont {J.}~\bibnamefont {Yang}},\ }\href
  {https://doi.org/10.1016/j.micpro.2016.07.008} {\bibfield  {journal}
  {\bibinfo  {journal} {Microprocessors and Microsystems}\ }\textbf {\bibinfo
  {volume} {47}},\ \bibinfo {pages} {209} (\bibinfo {year} {2016})}\BibitemShut
  {NoStop}%
\bibitem [{Note1()}]{Note1}%
  \BibitemOpen
  \bibinfo {note} {Code is available at \protect \url
  {https://github.com/GregDMeyer/quantum-advantage} and is archived on
  Zenodo~\cite {cirq_code}}\BibitemShut {NoStop}%
\bibitem [{\citenamefont {Meyer}(2022)}]{cirq_code}%
  \BibitemOpen
  \bibfield  {author} {\bibinfo {author} {\bibfnamefont {G.}~\bibnamefont
  {Meyer}},\ }\href {https://doi.org/10.5281/zenodo.6519250} {\bibinfo {title}
  {Gregdmeyer/quantum-advantage: v1.1}} (\bibinfo {year} {2022})\BibitemShut
  {NoStop}%
\end{thebibliography}%
	\bibliographystyle{apsrev4-2}

\end{document}


\title{Supplementary Information: Classically-Verifiable Quantum Advantage from a Computational Bell Test}
	\author{Gregory D. Kahanamoku-Meyer}
	\affiliation{Department of Physics, University of California at Berkeley, Berkeley, CA 94720}
	\author{Soonwon Choi}
	\affiliation{Department of Physics, University of California at Berkeley, Berkeley, CA 94720}
	\author{Umesh V. Vazirani}
	\affiliation{Department of Electrical Engineering and Computer Science, University of California at Berkeley, Berkeley, CA 94720}
	\author{Norman Y. Yao}
	\affiliation{Department of Physics, University of California at Berkeley, Berkeley, CA 94720}
	
	\maketitle

	\section{Cryptographic proofs of TCF properties}
	\label{si:tcf-proofs}

	Here we prove the cryptographic properties of the trapdoor claw-free functions (TCFs) presented in the Methods section of the main text.
	We base our definitions on the Noisy Trapdoor Claw-free Function family (NTCF) definition given in Definition 3.1 of Ref.~\cite{brakerski_cryptographic_2019}, with certain modifications such as removing the adaptive hardcore bit requirement and the ``noisy'' nature of the functions.

	We emphasize that in the definitions below, we define security only against classical attackers.
	Both the $x^2 \bmod N$ and DDH constructions could be trivially defeated by a quantum adversary via Shor's algorithm; since the purpose of the protocol in this paper is to demonstrate quantum capability, this type of adversary is allowed.
	
	We also note that the TCF definition allows the 2-to-1 property to be ``imperfect''---that is, we allow the fraction of pre-images which have a colliding pair to be less than 1.
	In the protocol, the verifier may simply discard any runs in which the prover supplied an output $y$ value that is not part of a claw, that is, does not have two corresponding inputs.
	This will not affect the prover's ability to pass the classical threshold (since these runs are counted neither for or against the prover); it will only possibly affect the number of iterations of the protocol required to exceed the classical bound with the desired statistical significance. 
	In the definition below we require the fraction of ``good'' inputs be at least a constant (which we set to 0.9); in principle the fraction could be as low as $1/\mathrm{poly}(\lambda)$ without interfering with the protocol's effectiveness.

	\subsection{TCF definition}

	We use the following definition of a Trapdoor Claw-free Function family:

	\begin{definition}
	\label{def:tcf}

	Let $\lambda$ be a security parameter, $I$ a set of function indices, and $X_i$ and $Y_i$ finite sets for each $i \in I$.
	A family of functions
	\[ \mathcal{F} = \{ f_i : X_i \to Y_i \}_{i \in I} \]
	is called a trapdoor claw free (TCF) family if the following conditions hold:

	\begin{enumerate}
		\item \textbf{Efficient Function Generation.} There exists an efficient probabilistic algorithm $\mathsf{Gen}$ which generates a key $i \in I$ and the associated trapdoor data $t_i$:
		\[(i, t_i) \leftarrow \mathsf{Gen}(1^\lambda)\]
		\item \textbf{Trapdoor Injective Pair.} For all indices $i \in I$, the following conditions hold:
		\begin{enumerate}
			\item Injective pair: Consider the set $R_i$ of all tuples $(x_0, x_1)$ such that $f_i(x_0) = f_i(x_1)$.
			Let $X_i' \subseteq X_i$ be the set of values $x$ which appear in the elements of $R_i$.
			For all $x \in X_i'$, $x$ appears in exactly one element of $R_i$; furthermore, there exists a value $\lambda_c$ such that for all $\lambda > \lambda_c$, $|X_i '|/|X_i| > 0.9$.
			\item Trapdoor: There exists an efficient deterministic algorithm $T$ such that for all $y \in Y_i$ and $(x_0, x_1)$ such that $f_i(x_0) = f_i(x_1) = y$, $T(t_i, y) = (x_0, x_1)$.
		\end{enumerate}
		\item \textbf{Claw-free.} For any non-uniform probabilistic polynomial time (nu-PPT) classical Turing machine $\mathcal{A}$, there exists a negligible function $\epsilon(\cdot)$ such that
		\[ \Pr \left[ f_i(x_0) = f_i(x_1) \land x_0 \neq x_1 | (x_0, x_1) \leftarrow \mathcal{A}(i) \right] < \epsilon(\lambda) \]
		where the probability is over both choice of $i$ and the random coins of $\mathcal{A}$.
		\item \textbf{Efficient Superposition.} There exists an efficient quantum circuit that on input a key $i$ prepares the state
		\[ \frac{1}{\sqrt{|X_i|}} \sum_{x \in X_i} \ket{x} \ket{f_i(x)} \]
	\end{enumerate}

	\end{definition}

	\subsection{Proof of \texorpdfstring{$x^2 \bmod N$}{x² mod N} TCF}

	In this section we prove that the function family $\mathcal{F}_\mathrm{Rabin}$ (defined in Methods) is a TCF by demonstrating each of the properties of Definition~\ref{def:tcf}.
	Most of the properties follow directly from properties of the Rabin cryptosystem \cite{rabin_digitalized_1979}; we reproduce several of the arguments here for completeness.

	\begin{thm}
	The function family $\mathcal{F}_\mathrm{Rabin}$ is trapdoor claw-free, under the assumption of hardness of integer factorization.
	\end{thm}

	\begin{proof}
	We demonstrate each of the properties of Definition~\ref{def:tcf}:

	\begin{enumerate}
		\item \textbf{Efficient Function Generation.} Sampling large primes to generate $p, q$ and $N$ is efficient \cite{rabin_digitalized_1979}.

		\item \textbf{Trapdoor Injective Pair.}
		\begin{enumerate}
			\item Injective pair:
			By definition of the function, $Y_i$ is the set of quadratic residues modulo $N$. 
			For any $y \in Y_i$, consider the two values $a < p/2$ and $b < q/2$ such that $a^2 \equiv y \bmod p$ and $b^2 \equiv y \bmod q$.
			These values exist because $y$ is a quadratic residue modulo $pq$, therefore it is also a quadratic residue modulo $p$ and $q$.
			Define $c \equiv 1 \bmod p \equiv 0 \bmod q$ and $d \equiv 0 \bmod p \equiv 1 \bmod q$.
			The following four values $x$ in the range $[0, N)$ have $x^2 \equiv y \bmod N$: $ac + bd, ac-bd, -ac+bd, -ac-bd$.
			Exactly two of these values are in the domain $[N/2]$ of the TCF, and constitute the injective pair; moreover, these two values will be unique as long as $a,b \neq 0$.
			Thus we may define the set $X'_i = \{x \in [N/2] | x \not\equiv 0 \bmod p \land x \not\equiv 0 \bmod q \}$.
			There exist exactly $((p-1)+(q-1))/2$ multiples of $p$ or $q$ in the set of integers $X_i = [N/2]$, thus $|X_i '|/|X_i| = 1 - ((p-1)+(q-1))/N$.
			Recall that $p,q$ are defined to have length $\lambda/2$; if we let $\lambda_c = 12$, then $p, q > 2^5 = 32$.
			Since $1 - (31+31)/32^2 > 0.9$ and $|X_i '|/|X_i|$ increases monotonically with $\lambda$, we have $|X_i '|/|X_i| > 0.9$ for all $\lambda > \lambda_c$.
			\item Trapdoor: Because $p$ and $q$ were selected to have $p \equiv q \equiv 3 \bmod 4$, $a$ and $b$ in the expressions above can always be computed as $a = y^{(p+1)/4} \bmod p$ and $b = y^{(q+1)/4} \bmod q$, and then the preimages can be computed as defined above.
		\end{enumerate}
		\item \textbf{Claw-free.} We show that knowledge of a claw $x_0, x_1$ can be used directly to factor $N$.
		Writing the claw as $(ac+bd, ac-bd)$ using the values $a, b, c, d$ from above, we have $x_0 + x_1 = 2ac$.
		Because $c = 0 \bmod q$, $\gcd(x_0 + x_1, N) = q$ can be efficiently computed, which then also yields $p = N/q$.
		Thus, an algorithm that could be used efficiently to find claws could be equally used to efficiently factor $N$, which we assume to be hard.
		\item \textbf{Efficient Superposition.}
		The set of preimages $X_i$ is the set of integers $[N/2]$.
		A uniform superposition $\sum_{x \in X_i} \ket{x}$ may be computed by generating a uniform superposition of all bitstrings of length $n$ (via Hadamard gate on every qubit), and then evaluating a comparator circuit that generates the state $\sum \ket{x}\ket{x < N/2}$ where $\ket{x < N/2}$ is a bit on an ancilla.
		If this ancilla is then measured and the result is $\ket{1}$, the state is collapsed onto the superposition $\sum_{x \in X_i} \ket{x}$ (if the result is $\ket{0}$ the process should simply be repeated).
		Then a multiplication circuit to an empty register may be executed to generate the desired state $\sum_{x \in X} \ket{x} \ket{x^2 \bmod N}$.

	\end{enumerate}
	\end{proof}

	\subsection{Proof of Decisional Diffie-Hellman TCF}

	We now prove that $\mathcal{F}_\mathrm{DDH}$ (defined in Methods) forms a trapdoor claw-free function family.

	\begin{thm}
	The function family $\mathcal{F}_\mathrm{DDH}$ is trapdoor claw-free, under the assumption of hardness of the decisional Diffie-Hellman problem for the group $\mathbb{G}$.
	\end{thm}

	\begin{proof}
	We demonstrate each of the properties of Definition~\ref{def:tcf}:

	\begin{enumerate}
		\item \textbf{Efficient Function Generation.}
		Each step of $\mathsf{Gen}$ is efficient by inspection.

		\item \textbf{Trapdoor Injective Pair.}
		\begin{enumerate}
			\item Injective pair:
			First we note that the matrix $\mathbf{M}$ is chosen to be invertible, thus $f_0$ and $f_1$ are one-to-one.
			Therefore for all $\mathbf{x}_0 \in X_i$, at most one other preimage $\mathbf{x}_1 \in X_i$ has $f_i(\mathbf{x}_0) = f_i(\mathbf{x}_1)$.
			Furthermore, since colliding pairs have the structure $(0||\mathbf{x}_0'), (1||\mathbf{x}_1')$ with $\mathbf{x}_0' = \mathbf{x}_1' + \mathbf{s}$ and $\mathbf{s} \in \{0, 1\}^k$, the only preimages that will \emph{not} form part of a colliding pair are those where $\mathbf{x}_0'$ has a zero element at an index where $\mathbf{s}$ is nonzero, or $\mathbf{x}_1'$ has an element equal to $d-1$ where $\mathbf{s}$ is nonzero (the vector element will be outside of the range of vector elements for the other vector).
			Thus $|X_k '|/|X_k| > \left(1-\nicefrac{1}{d} \right)^{k}$.
			Since $d \sim \mathcal{O}(k^2)$ and $k \sim \mathcal{O}(\lambda)$, we have $\lim_{\lambda \to \infty} |X_k '|/|X_k| = 1$ with $|X_k '|/|X_k|$ monotonically increasing.
			Therefore, there exists a value $\lambda_c$ such that $|X_k '|/|X_k| > 0.9$ for all $\lambda > \lambda_c$.
			(We note that if we set $k=\lambda$ and $d=k^2$, then $\lambda_c = 10$.)
			\item Trapdoor: The steps of the algorithm $T$ are efficient by inspection. Crucially, the discrete logarithm of each vector element is possible by brute force, because the elements of $\mathbf{x}_0$ only take values up to polynomial in $\lambda$.
		\end{enumerate}
		\item \textbf{Claw-free.} An algorithm which could efficiently compute a claw $(0||\mathbf{x}_0', 1||\mathbf{x}_1')$ could then trivially compute the secret vector $\mathbf{s} = \mathbf{x}_0' - \mathbf{x}_1'$.
		For any matrix $\mathbf{M'}$, the existence of an algorithm to uniquely determine $\mathbf{s}$ from $(g^\mathbf{M'}, g^\mathbf{M's})$ would directly imply an algorithm for determining whether $\mathbf{M'}$ has full rank.
		But DDH implies it is computationally hard to determine whether a matrix $\mathbf{M}'$ is invertible given $g^\mathbf{M'}$~\cite{peikert_lossy_2008, freeman_more_2010}.
		Therefore DDH implies the claw-free property.
		\item \textbf{Efficient Superposition.} Because $d$ is a power of two, a superposition of all possible preimages $\mathbf{x}$ can be computed by applying Hadamard gates to every qubit in a register all initialized to $\ket{0}$.
		The function $f$ can then be computed by a quantum circuit implementing a classical algorithm for the group operation of $\mathbb{G}$.

	\end{enumerate}
	\end{proof}
	
	\section{Overview of Trapdoor Claw-free Functions}
	
	In this section, we provide a brief overview of the cryptographic concepts upon which this work relies.
	
	\medskip	
	
	Foundational to the field of cryptography is the idea of a \emph{one-way function}.
	Informally, this type of function is easy to compute, but hard to invert.
	Here, ``easy'' means that the function can be evaluated in time polynomial in the length of the input.
	By ``hard'' we mean that the cost of the best algorithm to invert the function is superpolynomial in the length of the input.
	In practice, for a given one-way function we desire that there exists a particular problem size (input length) for which the function can be evaluated fast enough that it is not overly costly to use, but for which inversion would be infeasible for even an adversary with large (but realistic) computing power.
	One way functions can be used directly to construct many useful cryptographic schemes, including pseudorandom number generators, private-key encryption, and secure digital signatures.
	
	In this work, we rely on a specific type of one-way function called a \emph{trapdoor claw-free function} (TCF).
	This class of functions has two additional features.
	
	First, it has a \emph{trapdoor}.
	This means that while the function is hard to invert in general, with the knowedge of some secret data (the trapdoor key) inversion becomes easy.
	This secret data should be easy to generate when the function is chosen (from a large family of similar functions), but should be hard to find given just the description of the function itself.
	For example, in this work we describe the function $x^2 \bmod N$, with $N$ the product of two primes.
	The trapdoor is the factorization of $N$.
	It is easy to generate this function along with the trapdoor, by simply selecting two primes and multiplying them together.
	However, under the assumption of hardness of integer factorization, given only the function description (namely the value $N$) it is computationally hard to find the trapdoor (the factors $p$ and $q$).
	
	The second additional feature of a TCF is that it is \emph{claw-free}.
	This means that the function is two-to-one (has two inputs that map to each output), but it is computationally hard to \emph{find} two such colliding inputs without the trapdoor.
	Note that if it were possible to invert the function it would be trivial to find a collision (by picking an input, computing the function to get the output corresponding to it, and then inverting the function to find the second input mapping to that output).
	However the claw-free property is a bit stronger than the hardness of inversion: there exist some two-to-one functions which are one-way but not claw-free.
	
	Importantly, in this work we only require that breaking the claw-free property is hard classically---indeed, the claw-free property of the DDH and $x^2 \bmod N$ TCFs described here can be fully broken by quantum computers.
	However, perhaps surprisingly, we do not \emph{require} that breaking the claw-free property is easy for a quantum machine.
	In fact, the claw-free property of the LWE and Ring-LWE based TCFs remains secure even against quantum attacks.
	This corresponds to a very powerful property of the protocol in this paper, and other related protocols: that a quantum computer can pass the test without actually being able to find a claw.
	This subtle distinction stems from the fact that the quantum prover generates a \emph{superposition} over two inputs that collide.
	No measurement of such a state can yield both superposed values classically in full, but the test is designed to not require both values---just the results of an appropriate measurement of the superposition.
	A classical cheater, on the other hand, still cannot pass the test because the idea of a superposition does not exist classically.

	\section{Explanation of circuit complexities}
	\label{app:circ-complexities}

	Here we describe each of the asymptotic circuit complexities listed in Table~I of the main text.
	For these estimates we drop factors of $\log\log n$ or less.
	In all cases, we assume integer multiplication can be performed in time $\mathcal{O}(n \log n)$ using the Schonhage-Strassen algorithm.

	We emphasize that the value of $n$ necessary to achieve classical hardness in practice varies widely among these functions, and 	also that the asymptotic complexities here may not be applicable at practical values of $n$.

	\textbf{LWE} \cite{brakerski_cryptographic_2019,regev_lattices_2005} The LWE cost is dominated by multiplying an $\mathcal{O}(n \log n) \times n$ matrix of integers by a length $n$ vector.
	The integers are of length $\log n$, so each multiplication is expected to take approximately $\mathcal{O}(\log n)$ time.
	Thus, the evaluation of the entire function requires $\mathcal{O}(n^2 \log^2 n)$ operations.

	$\boldsymbol{x^2 \bmod N}$ \cite{rabin_digitalized_1979} The function can be computed in time $\mathcal{O}(n \log n)$ using Schonhage-Strassen multiplication algorithm and Montgomery reduction for the modulus.

	\textbf{Ring-LWE} \cite{brakerski_simpler_2020, lyubashevsky_ideal_2012, de_clercq_efficient_2015, roy_compact_2014} Ring-LWE is dominated by the cost of multiplying one polynomial by $\log n$ other polynomials.
	Through Number Theoretic Transform techniques similar to the Schonhage-Strassen algorithm, each polynomial multiplication can be performed in time $\mathcal{O}(n \log n)$, so the total runtime is $\mathcal{O}(n \log^2 n)$.
	We note that integer multiplication and polynomial multiplication can be mapped onto each other, so the runtimes for $x^2 \bmod N$ and Ring-LWE scale identically except for the fact that Ring-LWE requires $\log n$ multiplications instead of $\mathcal{O}(1)$.

	\textbf{Diffie-Hellman} \cite{diffie_new_1976, peikert_lossy_2008, freeman_more_2010} The Diffie-Hellman based construction defined in Methods requires performing multiplication of a $k \times k$ matrix by a vector, with $k \sim \mathcal{O}(n)$.
	However, the ``addition'' operation for the matrix-vector multiply is the group operation of $\mathbb{G}$; we expect this operation to have complexity at least $\mathcal{O}(n \log n)$ (for e.g. integer multiplication).
	The exponentiation operations have exponent at most $d \sim \mathcal{O}(k^2)$, so can be performed in $\mathcal{O}(\log n)$ group operations.
	So, for each of the $k^2$ matrix elements one must perform an operation of complexity $\mathcal{O}(n \log^2 n)$, yielding a total complexity of $\mathcal{O}(n^3 \log^2 n)$.

	\textbf{Shor's Algorithm} \cite{shor_polynomial-time_1997} Allowing for the use of Schonhage-Strassen integer multiplication, Shor's algorithm requires $\mathcal{O}(n^2 \log n \log\log n)$ gates \cite{zalka_fast_1998}.

	\section{Optimal classical algorithm}
	\label{app:optimal-classical}

	Here we provide an example of a classical algorithm that saturates the probability bound of Theorem~2 of the main text.
	It has $p_x = 1$ and $p_\mathrm{CHSH} = \nicefrac{3}{4}$.

	For a TCF $f : X \to Y$, consider a classical prover that simply picks some value $x_0 \in X$, and then computes $y$ as $f(x_0)$, without ever having knowledge of $x_1$.
	If the verifier requests a projective $x$ measurement, they always return $x_0$, causing the verifier to accept with $p_x = 1$.
	In the other case (performing rounds 2 and 3 of the protocol), upon receiving $r$ they compute $b_0 = x_0 \cdot r$.
	The cheating prover now simply assumes that $x_0 \cdot r = x_1 \cdot r$, and thus that the correct single-qubit state that would be held by a quantum prover is $\ket{b_0}$, and returns measurement outcomes accordingly.
	With probability $\nicefrac{1}{2}$, $\ket{b_0}$ is in fact the correct single-qubit state; in this case they can always cause the verifier to accept.
	On the other hand, if $x_0 \cdot r \neq x_1 \cdot r$, the correct state is either $\ket{+}$ or $\ket{-}$. 
	With probability $\frac{1}{2}$, the measurement outcome reported by the cheating prover will happen to be correct for this state too.
	Overall, this cheating prover will have $p_\mathrm{CHSH} = (1 + \frac{1}{2})/2 = \frac{3}{4}$.

	Thus we see $p_x + 4p_\mathrm{CHSH} - 4 = 1 + 4\cdot\frac{3}{4} - 4 = 0$ which saturates the bound.

	\section{Quantum circuits for Karatsuba and schoolbook multiplication}
	\label{app:large-circuits}

	Classically, multiplication of large integers is generally performed using recursive algorithms such as Schonhage-Strassen \cite{schonhage_schnelle_1971} and Karatsuba which have complexity as low as $\mathcal{O}(n \log n)$.
	In the quantum setting, the need to store garbage bits at each level of recursion has limited their usefulness~\cite{kowada_reversible_2006, parent_improved_2017}.
	There does exist a reversible construction of Karatsuba multiplication that uses a linear number of qubits \cite{gidney_asymptotically_2019}, but due to overhead required for its implementation it does not begin to outperform schoolbook multiplication until the problem size reaches tens of thousands of bits.

	Leveraging the irreversibility described in Section~IID of the main text, we are able use these recursive algorithms directly, without needing to maintain garbage bits for later uncomputation.
	We implement both the $\mathcal{O}(n^{1.58})$ Karatsuba multiplication algorithm and the simple $\mathcal{O}(n^2)$ ``schoolbook'' algorithm.
	Due to efficiencies gained from discarding garbage bits, we find that the Karatsuba algorithm already begins to outcompete schoolbook multiplication at problem sizes of under 100 bits.
	Thus Karatsuba seems to be the best candidate for ``full-scale'' tests of quantum advantage at problem sizes of $n\sim500-1000$ bits.
	We also note that the Schonhage-Strassen algorithm scales even better than Karatsuba as $\mathcal{O}(n \log n \log \log n)$.
	However, even in classical applications it has too much overhead to be useful at these problem sizes.
	We leave its potential quantum implementation to a future work.

	The multiplication algorithms just described do not include the modulo $N$ operation, it must be performed in a separate step.
	We implement the modulo using only two classical-quantum multiplications and one addition via Montgomery reduction \cite{montgomery_modular_1985}.
	Montgomery reduction does introduce a constant $R'$ into the product, but this factor can be removed in classical post-processing after $y=x^2 R' \bmod N$ is measured.

	Finally, we note that at the implementation level, optimizing \emph{classical} circuits for modular integer multiplication has received significant study in the context of performing cryptography on embedded devices and FPGAs~\cite{javeed_design_2016, moralessandoval_scalable_2016, yang_efficient_2016}.
	Mapping such optimized circuits into the quantum context may be a promising avenue for further research.

	\section{Details of post-selection scheme}
	\label{app:postselection}

	In this section we describe several details of the post-selection scheme proposed in Section~IIC of the main text.

	\subsection{Quantum prover with no phase coherence saturates the classical bound}

	Consider the two states $\ket{\psi_{\pm}} = (\ket{x_0} \pm \ket{x_1})_\mathsf{x} \ket{y}_\mathsf{y}$ for some claw $(x_0, x_1)$ with $y = f_k (x_0) = f_k (x_1)$.
	Note that $\ket{\psi_+}$ is the state that would be held by a noise-free prover.
	Suppose a noisy quantum prover is capable of generating the mixed state
	\begin{equation}
	\rho_\delta = (1/2 + \delta) \ket{\psi_+} \bra{\psi_+} + (1/2 - \delta) \ket{\psi_-} \bra{\psi_-}.
	\end{equation}
	In words, they are able to generate a state that is a superposition of the correct bitstrings, but with the correct phase only $1/2 + \delta$ fraction of the time.
	Here we show that such a prover can exceed the classical threshold of Theorem~2 of the main text, whenever $\delta > 0$.
	We proceed by examining this prover's behavior during the protocol.

	First, we note that if the verifier requests a projective $x$ measurement after Round 1 of the protocol, this prover will always succeed---they simply measure the $\mathsf{x}$ register as instructed, and the phase is not relevant.
	Thus, using the notation of Theorem~2, $p_x = 1$.
	With this value set, to exceed the bound we must achieve $p_\mathrm{CHSH} > 3/4$.
	Naively performing the rest of the protocol as described in the main text does not exceed the bound when $\delta$ is small.
	However, the noisy prover can exceed the bound if they adjust the angle of their measurements in the third round of protocol (but preserve the sign of the measurement requested by the prover).
	We now demonstrate how.

	Define $\ket{\phi}$ as the ``correct'' single-qubit state at the end of Round 2---one of $\{\ket{0}, \ket{1}, \ket{+}, \ket{-}\}$.
	Let $f_\updownarrow$ be the probability that our noisy prover holds the correct state when $\ket{\phi} \in \{\ket{0}, \ket{1} \}$, and $f_\leftrightarrow$ the corresponding probability when $\ket{\phi} \in \{\ket{0}, \ket{1} \}$.
	In the first case, the potential phase error of our prover does not affect the single-qubit state, so $f_{\updownarrow} = 1$.
	In the other case, the state is only correct when the phase is correct, so $f_{\leftrightarrow} = 1/2 + \delta$.
	We see that our prover will hold the correct single-qubit state with probability greater than $3/4$.
	But, if they naively measure in the prescribed off-diagonal basis $\theta \in \{\pi/4, -\pi/4\}$ from the verifier, for small $\delta$ their success probability will be less than $3/4$.
	This can be rectified by adjusting the rotation angle of the measurement basis.

	Letting $\pm \theta'$ define the pair of measurement angles used by the prover in step 3 of the protocol (nominally $\theta' = |\theta| = \pi/4$), we can now express the prover's success probability $p_\mathrm{CHSH}$ as
	\begin{equation}
	\label{eq:pm_theta}
	p_m = \frac{1}{2} \left[ \cos^2 \left( \frac{\theta'}{2} \right) f_\updownarrow + \cos^2 \left( \frac{\theta'}{2} - \frac{\pi}{4} \right) f_\leftrightarrow + \sin^2 \left( \frac{\theta'}{2} \right) (1 - f_\updownarrow) + \sin^2 \left( \frac{\theta'}{2} - \frac{\pi}{4} \right) (1 - f_\leftrightarrow) \right]
	\end{equation}
	If the prover measures with $\theta' = \pi/4$ as prescribed in the protocol, the success rate will be $p_\mathrm{CHSH} \approx 0.68 + \mathcal{O}(\delta) < 3/4$.
	However, if they instead adjust their measurement angle to $\theta = \delta$, they instead achieve $p_\mathrm{CHSH} = 3/4 + 3 \delta^2 / 8 - \mathcal{O}(\delta^3)$, which exceeds the classical bound (provided that $\delta$ is large enough to be noticeable).

	In practice, both $f_\updownarrow$ and $f_\leftrightarrow$ are likely to be less than one; the optimal measurement angle can be determined as
	\begin{equation}
	\label{eq:opt_theta}
	\theta_\mathrm{opt}' = \tan^{-1} \left( \frac{2 f_\leftrightarrow - 1}{2 f_\updownarrow - 1} \right)
	\end{equation}
	which is the result of optimizing Equation~\ref{eq:pm_theta} over $\theta'$.
	In a real experiment, it would be most effective to empirically determine $f_\updownarrow$ and $f_\leftrightarrow$ and then use Equation~\ref{eq:opt_theta} to determine the optimal measurement angle.

	\vspace{1em}

	\subsection{Details of simulation and error model}

	We now describe the details of the numerical simulation that was used to generate Figure~2 of the main text.
	For several values of the overall circuit fidelity $\mathcal{F}$, we established a per-gate fidelity as $f = \mathcal{F}^{1/N_g}$ where $N_g$ is the number of gates in the $x^2 \bmod N$ circuit.
	We then generated a new circuit to compute the function $(3^a x)^2 \bmod 3^{2a} N$ for various values of $a$ (see next subsection for an explanation of the choice $k=3^a$).
	For each gate in the new circuit, with probability $1-f$ we added a Pauli ``error'' operator randomly chosen from $\{X, Y, Z\}$ to one of the qubits to which the gate was being applied.

	For the simulation, we randomly chose two primes $p$ and $q$ that multiplied to yield an integer $N$ of length $512$ bits.
	We then randomly chose a large set of colliding preimage pairs, and simulated the circuit separately for each such preimage (which is classically efficient, since the circuits only consist of $X$, CNOT, and Toffoli gates).
	The relative phase between each pair of preimages (due to error gates) was tracked explicitly during the simulation.
	Finally, the expected success rate of the prover was determined by analyzing the correctness of the bitstrings and their relative phase at the end of the circuit.
	
	The primes $p$ and $q$ used to generate Figure 2 of the main text are (in base 10):
	
	\begin{verbatim}
	p = 113287732919697174280284729511923238986362403955638184856698528941220766063369
	q = 98359967382337110635377957241353362183812709461386334819166502848512740692727
	\end{verbatim}

	\subsection{Choice of \texorpdfstring{$k=3^a$}{k=3{\textasciicircum}a} to improve postselection for \texorpdfstring{$x^2 \bmod N$}{x² mod N}}

	In the previous subsection, we map the TCF $f_N = x^2 \bmod N$ to the function $f_N' = (kx)^2 \bmod k^2 N$.
	To achieve this at the implementation level, we may use essentially the same circuit for modular multiplication; the only new requirement is to efficiently generate a superposition of multiples $kx$ in the $\mathsf{x}$ register.
	We generate this superposition by starting with a uniform superposition over values $x$ and then multiplying by $k$.

	Normally, quantum multiplication circuits (like those we use to evaluate $x^2 \bmod N$) perform an out-of-place multiplication, where the result is stored in a new register.
	In this case, however, it is preferable to do the multiplication ``in-place,'' where the result is stored in the input register itself---this way the $y$ value is computed directly from the input register and thus is more likely to reflect errors that may occur in the input.

	In general, performing in-place multiplication is complicated, particularly on a quantum register, because the input is being modified as it is being consumed (not to mention concerns about reversibility).
	However, multiplication by small constants is much simpler to implement.
	By setting $k$ to a power of three, we are able to implement the in-place multiplication by performing a sequence of in-place multiplications by $3$, which can each be performed quite efficiently (see implementation in the attached \texttt{Cirq} code~\footnote{Code is available at \url{https://github.com/GregDMeyer/quantum-advantage} and is archived on Zenodo~\cite{cirq_code}}).
	
	\subsection{Theory prediction of Figure 2 of the main text}
	
	For the dashed ``theory prediction'' lines of Figure 2 of the main text, we predicted the success probabilities under two assumptions (which the numerical experiments are intended to test). 
	First, among noisy runs where at least one bit flip error occurs, the output bitstring is approximately uniformly distributed.
	Second, we assume that with at least one phase flip error, the probability that the phase is correct in the final state is $1/2$.
	
	Under these assumptions, we compute the predicted success rates $p_x$ and $p_\mathrm{CHSH}$ as follows:
	\begin{enumerate}
	\item For a given overall fidelity $\mathcal{F}$ of the original $x^2 \bmod N$ circuit containing $N_g$ gates, compute a per-gate fidelity $f = \mathcal{F}^{1/N_g}$. Then compute the expected overall fidelity $\mathcal{F}'$ of running the slightly larger $(kx)^2 \bmod k^2N$ containing $N_g'$ gates as $f^{N_g'}$.
	\item Using $\mathcal{F}'$ and the given error model (see ``Details of simulation and error model'' section above), compute three disjoint probabilities: that no errors occur, that only phase errors occur, or that at least one bit flip error (and possibly also phase errors) occurs.
	\item Compute the probability that the output will pass postselection, which includes both cases with no bit flip errors and those that are corrupted but happen to pass postselection by chance.
	\item Normalizing to only those runs that pass postselection, compute $p_x$ and $p_\mathrm{CHSH}$:
	\begin{enumerate}
	\item $p_x$ is computed as the probability that no bit flip errors occurred (among those runs that pass postselection). This is a lower bound (that seems intuitively tight); it assumes a negligible probability that the measured pair $(x, y)$ still has $y = f(x)$ despite bit flip errors.
	\item $p_\mathrm{CHSH}$ is computed by finding the probability that no errors occurred that would affect the single-qubit state at the end of round 2. When the correct single-qubit state should be polarized along $Z$, this is taken to be the probability that no bit flip errors occurred (phase errors are allowed since they will not affect this state); when the correct state should be polarized along $X$, it is taken as the probability that no errors at all have occurred. In these ``no-error'' cases, we compute the verifier's probability of accepting by applying the adjusted measurement basis described in the first sub-section above, ``Quantum prover with no phase coherence saturates the classical bound''. Finally, for the case that there was an error that could affect the single-qubit state, the probability that the verifier receives a correct measurement outcome is taken to be $1/2$ (the single-qubit state is taken to be maximally mixed). 
	\end{enumerate}
	\item Compute the measure of ``quantumness'' from $p_x$ and $p_\mathrm{CHSH}$.
	\item Compute the estimate runtime by multiplying the increase in quantum circuit size by the expected number of iterations required to pass postselection (which is computed from the analysis above).
	\end{enumerate}

    \bibliography{refs}
	\bibliographystyle{apsrev4-2}


\title{Methods and Extended Data: Classically-Verifiable Quantum Advantage from a Computational Bell Test}
	\author{Gregory D. Kahanamoku-Meyer}
	\affiliation{Department of Physics, University of California at Berkeley, Berkeley, CA 94720}
	\author{Soonwon Choi}
	\affiliation{Department of Physics, University of California at Berkeley, Berkeley, CA 94720}
	\author{Umesh V. Vazirani}
	\affiliation{Department of Electrical Engineering and Computer Science, University of California at Berkeley, Berkeley, CA 94720}
	\author{Norman Y. Yao}
	\affiliation{Department of Physics, University of California at Berkeley, Berkeley, CA 94720}

	\maketitle

	\section{List decoding lemma}
	\label{app:list-decoding-lemma}

	In this section we prove a bound on the probability that list decoding will succeed for a particular value of $y$, given an oracle's noise rate over \emph{all} values of $y$.
	Recall that by the Goldreich-Levin theorem~\cite{goldreich_hard-core_1989}, list decoding of the Hadamard code is possible if the noise rate is noticeably less than $1/2$.

	\begin{lemma}
	\label{lem:list-decoding-bound}
	Consider a binary-valued function over two inputs $g : Y \times \{0,1\}^n \to \{0,1\}$, and a noisy oracle $\mathcal{G}$ to that function.
	Assuming some distribution of values $y \in Y$ and $r \in \{0,1\}^n$, define $\epsilon \equiv \Pr_{y,r} [\mathcal{G}(y,r) \neq g(y,r)]$ as the ``noise rate'' of the oracle.
	Now define the conditional noise rate for a particular $y \in Y$ as
	\begin{equation}
	\epsilon_{y} \equiv \Pr_{r}\left[\mathcal{G}(y,r) \neq g(y, r)\right]
	\end{equation}
	Then, the probability that $\epsilon_y$ is less than $1/2 - \mu(n)$ for any positive function $\mu$, over randomly selected $y$, is
	\begin{equation}
	p_\mathrm{good} \equiv \Pr_{y}\left[\epsilon_{y} < 1/2 - \mu(n) \right] \geq 1 - 2 \epsilon - 2\mu(n).
	\end{equation}

	\end{lemma}

	\begin{proof}
	Let $S \subseteq Y$ be the set of $y$ values for which $\epsilon_{y} < 1/2 - \mu(n)$.
	Then by definition we have
	\begin{equation}
	\label{eq:eps-def}
	\epsilon = p_\mathrm{good} \cdot \epsilon_{y \in S} + (1-p_\mathrm{good}) \cdot \epsilon_{y \notin S}
	\end{equation}

	Noting that we must have $\epsilon_{y} \geq 1/2 - \mu(n)$ for $y \notin S$ by definition, we may minimize the right hand side of Equation~\ref{eq:eps-def}, yielding the bound
	\begin{equation}
	\epsilon > p_\mathrm{good} \cdot 0 + (1-p_\mathrm{good}) \cdot (1/2 - \mu(n))
	\end{equation}

	Rearranging this expression we arrive at
	\[
	p_\mathrm{good} > 1 - 2\epsilon - 2\mu(n)
	\]
	which is what we desired to show.
	\end{proof}

	\section{Trapdoor claw-free function constructions}
	\label{app:tcfs}

	Here we present two trapdoor claw-free function families (TCFs) for use in the protocol of this paper.
	These families are defined by three algorithms: $\mathsf{Gen}$, a probabilistic algorithm which selects an index $i$ specifying one function in the family and outputs the corresponding trapdoor data $t$; $f_i$, the definition of the function itself; and $T$, a trapdoor algorithm which efficiently inverts $f_i$ for any $i$, given the corresponding trapdoor data $t$.
	Here we provide the definitions of the function families; proofs of their cryptographic properties are included in the supplementary information.
	In these definitions we use a security parameter $\lambda$ following the notation of cryptographic literature; $\lambda$ is informally equivalent to the ``problem size'' $n$ defined in the main text as the length of the TCF input string.

	\subsection{TCF from Rabin's function \texorpdfstring{$x^2 \bmod N$}{x² mod N}}
	\label{app:tcf-rabin-def}

	``Rabin's function'' $f_N(x) = x^2 \bmod N$, with $N$ the product of two primes, was first used in the context of public-key cryptography and digital signatures \cite{rabin_digitalized_1979, goldwasser_digital_1988}.
	We use it to define the trapdoor claw-free function family $\mathcal{F}_{\mathrm{Rabin}}$, as follows.

	\medskip
	\noindent\textbf{Function generation} \\
	\noindent\rule[\baselineskip]{\linewidth}{0.4pt}

	$\mathsf{Gen}\left(1^{\lambda}\right)$
	\begin{enumerate}
		\item Randomly choose two prime numbers $p$ and $q$ of length $\lambda/2$ bits, with $p \bmod 4 \equiv q \bmod 4 \equiv 3 \bmod 4$
		\footnote{In practice, $p$ and $q$ must be selected with some care such that Fermat factorization and Pollard's $p-1$ algorithm~\cite{pollard_theorems_1974} cannot be used to efficiently factor $N$ classically. Selecting $p$ and $q$ in the same manner as for RSA encryption would be effective ~\cite{rivest_method_1978}.}.
		\item Return $N=pq$ as the function index, and the tuple $(p,q)$ as the trapdoor data.
	\end{enumerate}

	\medskip
	\noindent\textbf{Function definition} \\
	\noindent\rule[\baselineskip]{\linewidth}{0.4pt}

	$f_N : [N/2] \to [N]$ is defined as
	\begin{align}
	f_N (x) = x^2 \bmod N
	\end{align}
	The domain is restricted to $[N/2]$ to remove extra trivial collisions of the form $(x, -x)$.

	\medskip
	\noindent\textbf{Trapdoor} \\
	\noindent\rule[\baselineskip]{\linewidth}{0.4pt}

	The trapdoor algorithm is the same as the decryption algorithm of the Rabin cryptosystem \cite{rabin_digitalized_1979}.
	On input $y$ and key $(p,q)$, the Rabin decryption algorithm returns four integers $(x_0, x_1, -x_0, -x_1)$ in the range [0, N).
	$x_0$ and $x_1$ can then be selected by choosing the two values that are smaller than $N/2$.
	See proof in supplementary information for an overview of the algorithm.

	\subsection{TCF from Decisional Diffie-Hellman}
	\label{app:tcf-ddh-def}

	We now present a trapdoor claw-free function family $\mathcal{F}_{\mathrm{DDH}}$ based on the decisional Diffie-Hellman problem (DDH).
	DDH is defined for a multiplicative group $\mathbb{G}$; informally, the DDH assumption states that for a group generator $g$ and two integers $a$ and $b$, given $g$, $g^a$, and $g^b$ it is computationally hard to distinguish $g^{ab}$ from a random group element.
	We expand on a known DDH-based trapdoor one-way function construction \cite{peikert_lossy_2008,freeman_more_2010}, adding the claw-free property to construct a TCF.

	\medskip
	\noindent\textbf{Function generation} \\
	\noindent\rule[\baselineskip]{\linewidth}{0.4pt}

	$\mathsf{Gen}\left(1^{\lambda}\right)$
	\begin{enumerate}
	\item Choose a group $\mathbb{G}$ of order $q\sim\mathcal{O}\text{\ensuremath{\left(2^{\lambda}\right)}}$,
	and a generator $g$ for that group.
	\item For dimension $k > \log_2 q$ choose a random invertible matrix $\mathbf{M} \in\mathbb{Z}_{q}^{k \times k}$.
	\item Compute $g^{\mathbf{M}}=\left(g^{\mathbf{M}_{ij}}\right)\in\mathbb{G}^{k \times k}$ (element-wise exponentiation).
	\item Choose a secret vector $\mathbf{s} \in\left\{ 0,1\right\}^{k}$; compute the vector $g^{\mathbf{M} \mathbf{s}}$ (where $\mathbf{M} \mathbf{s}$ is the matrix-vector
	product, and again the exponentiation is element-wise).
	\item Publish the pair $\left(g^{\mathbf{M}},g^{\mathbf{Ms}}\right)$, retain
	$\left(g, \mathbf{M},\mathbf{s}\right)$ as the trapdoor data.
	\end{enumerate}

	\medskip
	\noindent\textbf{Function definition} \\
	\noindent\rule[\baselineskip]{\linewidth}{0.4pt}

	Let $d$ be a power of two with $d \sim \mathcal{O}\left(k^{2}\right)$.
	We define the function $f_i$ as $f_i(b||\mathbf{x}) := f_{i,b}(\mathbf{x})$, where $||$ denotes concatenation, for a pair of functions $f_{i,b}:\mathbb{Z}_{d}^{k}\to\mathbb{G}^{k}$:

	\begin{align}
	f_{i,0}\left(\mathbf{x}\right) & =g^{\mathbf{Mx}}\\
	f_{i,1}\left(\mathbf{x}\right) & =g^{\mathbf{Mx}}g^{\mathbf{Ms}}=g^{\mathbf{M}\left(\mathbf{x}+\mathbf{s}\right)}
	\end{align}

	\medskip
	\noindent\textbf{Trapdoor} \\
	\noindent\rule[\baselineskip]{\linewidth}{0.4pt}

	The algorithm takes as input the trapdoor data $\left(g, \mathbf{M},\mathbf{s}\right)$ and a value $y=g^{\mathbf{Mx}_{0}}=g^{\mathbf{M}\left(\mathbf{x}_{1}+\mathbf{s}\right)}$, and returns the claw $\left(\mathbf{x}_{0}, \mathbf{x}_{1}\right)$.

	\begin{quote}
	$T((g, \mathbf{M}, \mathbf{s}), y)$
	\begin{enumerate}
	\item Compute $\mathbf{M}^{-1}$ using $\mathbf{M}$.
	\item Compute $g^{\mathbf{M}^{-1} \mathbf{M}\mathsf{x}_{0}}=g^{\mathbf{x}_{0}}$.
	\item Take the discrete logarithm of each element of $g^{\mathbf{x}_{0}}$,
	yielding $\mathbf{x}_{0}$. Crucially, this is possible because the
	elements of $\mathbf{x}$ are in $\mathbb{Z}_{d}$ and $d=\mathsf{poly}\left(n\right)$,
	so the discrete logarithm can be computed in polynomial time by brute
	force.
	\item Compute $\mathbf{x}_{1}=\mathbf{x}_{0}-\mathbf{s}$
	\item Return $\left(\mathbf{x}_{0}, \mathbf{x}_{1}\right)$
	\end{enumerate}
	\end{quote}

	\section{Table of circuit sizes}
	\label{app:circuit-table}

	A comparison of the resource requirements for computing $x^2 \bmod N$, for various problem sizes and circuit designs, is presented in Table~\ref{tab:gate-counts}.
	These counts are generated in the ``abstract circuit'' model, in which error correction, qubit routing, and other practical considerations are not included.
	For schoolbook and Karatsuba circuits, circuits are decomposed into a Clifford+$T$ gate set. 
	For the ``phase'' circuits, we allow controlled arbitrary phase rotations, as we expect these circuits to be appropriate for hardware (physical) qubits where these gates are native. 
	Accordingly, we do not provide $T$ gate counts for those circuits. 

	\begin{table}[h]
		\caption{Circuit sizes for various values of $n = \log N$. Values may vary for different $N$ of the same length. ``Qubit-optimized phase'' and ``gate-optimized phase'' refer to the circuits given in Figure 3(a) and 3(b) of the main text, respectively. ``Qubit measmts.'' refers to the number of times qubits are measured and then reused during execution of the circuit. ${}^*$From analytic estimate rather than building explicit circuit. ${}^\dag$Reversible circuits constructed using Q\# implementation of Ref.~\cite{gidney_asymptotically_2019}, and scaled to include Montgomery reduction. ${}^\ddagger$Estimate from \cite{gidney_how_2019}.}
		\label{tab:gate-counts}

		\vspace{5mm}
		\begin{tabular}{|c|c|c|c|c|c|c|}
			\hline
			Circuit & Qubits & \thead{Gates ($CCR_{\phi}$/\\Toffoli allowed)} & \thead{Gates\\(Clifford + $T$)} & $T$ Gates & Depth & \thead{Qubit\\measmts.}\\
			\hline
			\hline
			\multicolumn{7}{|c|}{$n = 128$ (takes seconds on a desktop \cite{noauthor_cado-nfs_nodate})} \\
			\hline
			Qubit-optimized phase & 128 & $1.1 \times 10^6$ & --- & --- & $1.1 \times 10^6$ & 128 \\
			\hline
			Gate-optimized phase & 264 & $4.3 \times 10^5$ & --- & --- & $6.3 \times 10^4$ & 0 \\
			\hline
			Schoolbook & 515 & $1.4 \times 10^5$ & $9.1 \times 10^5$ & $3.9 \times 10^5$ & $1.9 \times 10^4$ & $3.5 \times 10^4$ \\
			\hline
			Karatsuba  & 942 & $1.3 \times 10^5$ & $7.7 \times 10^5$ & $3.3 \times 10^5$ & $2.0 \times 10^3$ & $3.4 \times 10^4$\\
			\hline
			\hline
			\multicolumn{7}{|c|}{$n = 400$ (takes hours on a desktop \cite{noauthor_cado-nfs_nodate})} \\
			\hline
			Qubit-optimized phase & 400 & $3.3 \times 10^7 {}^*$ & --- & --- & $3.3 \times 10^7 {}^*$ & 400 \\
			\hline
			Gate-optimized phase & 812 & $4.2 \times 10^6 {}^*$ & --- & --- & $6.2 \times 10^5{}^*$ & 0 \\
			\hline
			Schoolbook & 1603 & $1.3 \times 10^6$ & $8.7 \times 10^6$ & $3.6 \times 10^6$ & $5.9 \times 10^4$ & $3.3 \times 10^5$ \\
			\hline
			Karatsuba  & 3051 & $8.8 \times 10^5$ & $5.4 \times 10^6$ & $2.3 \times 10^6$ & $5.3 \times 10^4$ & $2.4 \times 10^5$ \\
			\hline
			\hline
			\multicolumn{7}{|c|}{$n = 829$ (record for factoring \cite{zimmermann_cado-nfs-discuss_2020})} \\
			\hline
			Qubit-optimized phase & 829 & $3.0 \times 10^8{}^*$ & --- & --- & $2.9 \times 10^8{}^*$ & 829 \\
			\hline
			Gate-optimized phase & 1671 & $1.8 \times 10^7{}^*$ & --- & --- & $2.6 \times 10^6{}^*$ & 0 \\
			\hline
			Schoolbook & 3319 & $5.6 \times 10^6$ & $3.8 \times 10^7$ & $1.6 \times 10^7$ & $1.2 \times 10^5{}^*$ & $1.4 \times 10^6$ \\
			\hline
			Karatsuba  & 5522 & $3.0 \times 10^6$ & $1.8 \times 10^7$ & $7.7 \times 10^6$ & $1.1 \times 10^5{}^*$ & $8.0 \times 10^5$ \\
			\hline
			\hline
			\multicolumn{7}{|c|}{$n = 1024$ (exceeds factoring record)} \\
			\hline
			Qubit-optimized phase & 1024 & $5.6 \times 10^8{}^*$ & --- & --- & $5.5 \times 10^8{}^*$ & 1024 \\
			\hline
			Gate-optimized phase & 2061 & $2.7 \times 10^7{}^*$ & --- & --- & $4.0 \times 10^6{}^*$ & 0 \\
			\hline
			Schoolbook & 4097 & $8.3 \times 10^6$ & $5.7 \times 10^7$ & $2.4 \times 10^7$ & $1.5 \times 10^5{}^*$ & $2.1 \times 10^6$ \\
			\hline
			Karatsuba  & 6801 & $4.3 \times 10^6$ & $2.6 \times 10^7$ & $1.1 \times 10^7$ & $1.4 \times 10^5{}^*$ & $1.1 \times 10^6$ \\
			\hline
			\hline
			\multicolumn{7}{|c|}{Other algs. at $n = 1024$} \\
			\hline
			Rev. schoolbook ${}^\dag$ & 8192 & --- & $6.4 \times 10^8$ & $2.2 \times 10^8$ & $1.1 \times 10^8$ & 0 \\
			\hline
			Rev. Karatsuba ${}^\dag$ & 12544 & --- & $5.7 \times 10^8$ & $1.9 \times 10^8$ & $2.4 \times 10^7$ & 0 \\
			\hline
			Shor's alg. ${}^\ddagger$ & 3100 & --- & --- & $1.9 \times 10^{9}{}^*$ & --- & --- \\
			\hline
		\end{tabular}
	\end{table}
	
	\bibliography{refs}
	\bibliographystyle{apsrev4-2}